\documentclass[10pt,titlepage]{article}
\usepackage[intlimits]{amsmath}
\usepackage{fullpage}
\usepackage{feynmp,epsf}
\usepackage{calrsfs}
\usepackage{textcomp}
\usepackage{amssymb,amscd}
\usepackage{slashed}
\usepackage{amsthm}
\usepackage[dvips]{graphicx}
\usepackage{rotate}

\input xy
\xyoption{all}

\DeclareMathOperator{\prim}{prim}

\DeclareMathOperator{\Sym}{Sym}

\DeclareMathOperator{\colour}{colour}
\DeclareMathOperator{\Tr}{Tr}

\renewcommand\epsilon\varepsilon
\renewcommand\tilde\widetilde
\renewcommand\hat\widehat
\renewcommand\bar\overline

\DeclareMathOperator*\Res{Res}
\DeclareMathOperator*\Reg{Reg}
\DeclareMathOperator\skeleton{skeleton}

\newcommand\fracx[2]{\genfrac{}{}{0pt}{}{#1}{#2}}

\theoremstyle{plain}
\newtheorem{thm}{Theorem}[section]
\newtheorem{lem}[thm]{Lemma}
\newtheorem{prop}[thm]{Proposition}
\theoremstyle{definition}
\newtheorem{defn}[thm]{Definition}
\newtheorem{exmp}[thm]{Example}
\newtheorem{ex}[thm]{Example}
\newtheorem{rem}[thm]{Remark}

\def\One{\mathbb{I}}
\newcommand{\be}{\begin{equation}}
\newcommand{\ee}{\end{equation}}
\newcommand{\bea}{\begin{eqnarray}}
\newcommand{\eea}{\end{eqnarray}}
\newcommand{\beas}{\begin{eqnarray*}}
\newcommand{\eeas}{\end{eqnarray*}}

\def\dunce{\;\raisebox{-4mm}{\epsfysize=12mm\epsfbox{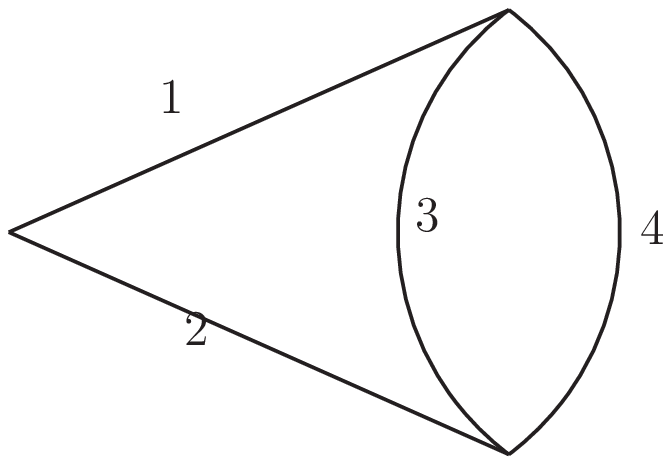}}\;}

\newcommand\markededge{{\includegraphics{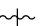}}}
\newcommand\mmarkededge{{\includegraphics{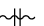}}}
\newcommand\dotsedge{{\includegraphics{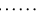}}}
\newcommand\ddotsedge{{\includegraphics{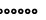}}}
\newcommand\fourvertex{{\includegraphics{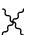}}}
\newcommand\ghostloop{{\includegraphics{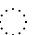}}}

\theoremstyle{plain}
\newtheorem{lemma}[thm]{Lemma}
\newtheorem{cor}[thm]{Corollary}
\theoremstyle{definition}
\newtheorem{defi}[thm]{Definition}
\newtheorem{exa}[thm]{Example}

\title{Quantization of gauge fields, graph polynomials  and graph homology
\thanks{This work is supported by an Alexander von Humboldt Professorship of the Alexander von
Humboldt Foundation and the BMBF for Dirk Kreimer. DK: kreimer@physik.hu-berlin.de, ph: +49 30 2093 7630, fax: +49 30 2093 3984}
}
\author{Dirk Kreimer \and Matthias Sars \and Walter D. van Suijlekom\\
Humboldt University, 10099 Berlin, Germany (D.K., M.S.) \and Radboud University Nijmegen, 6525 AJ Nijmegen,
The Netherlands}

\begin{document}

\begin{fmffile}{graphs}
\fmfset{thin}{.5pt}
\fmfset{thick}{1pt}
\fmfset{wiggly_len}{4pt}
\fmfset{arrow_len}{6pt}
\fmfset{dot_len}{2pt}
\fmfcmd{
 vardef bar (expr p, len, ang, sh) =
  ((-len/2,0)--(len/2,0))
     rotated (ang + angle direction length(p)/2 of p)
     shifted point length(p)/2 + sh of p
 enddef;
 style_def marked expr p =
  draw_wiggly p;
  ccutdraw bar (p, 5, 90, 0);
 enddef;
 style_def mmarked expr p =
  draw_wiggly p;
  ccutdraw bar (p, 5, 90, .08);
  ccutdraw bar (p, 5, 90, -.08);
 enddef;
}

\maketitle

\begin{abstract}
We review quantization of gauge fields using algebraic properties of 3-regular graphs. 
We derive the Feynman integrand at $n$ loops for a non-abelian gauge theory quantized in a covariant gauge
from scalar integrands for connected 3-regular graphs, obtained from the two Symanzik polynomials.

The transition to the full gauge theory amplitude is obtained by the use of a third, new, graph polynomial, the corolla polynomial.

This implies effectively a covariant quantization without ghosts, where all the relevant signs of the ghost sector 
are incorporated in a double complex furnished by the corolla polynomial ---we call it cycle homology--- and by graph homology.
\end{abstract}

\section{Introduction}
Feynman rules for scalar field theories reveal astonishing connections to algebraic geometry. These are most easily seen
upon the use of parametric representations, which furnish renormalized integrands suitable to be analysed in mathematical terms. 
Crucial are here the two Symanzik polynomials, which underlie the parametric representation \cite{BEK,BrI,BrII,BrSchI,BrSchII,BrSchY,DBr}. 

If one wants to generalize this approach into the realm of gauge theories, a pedestrian approach would be to turn every tensor integral appearing from any single graph
into the parametric representation, creating a bewildering number of tensor integrals for each and any contributing graph in a non-abelian gauge theory. While effective on-shell methods have been suggested as an alternative and indeed succeeded  at sufficiently low loop orders \cite{Dixon,Kosower}, the question how the transition from scalar to gauge field theory can be formulated mathematically remained open.

Here we establish an answer and suggest a  more succinct and, we believe, more elegant approach.   It allows us to obtain the renormalized integrand of 
a $n$-loop scattering amplitude in a gauge theory from the use of the scalar amplitude with cubic interaction and through one further graph polynomial, beyond the use of the  two Symanzik polynomials,  in one go.

We eliminate thereby the need to introduce  ghosts in the context of covariant quantization and replace their use by the use of the corolla polynomial.
This polynomial, which incorporates all the necessary signs of closed ghost loops, is still a strictly positive polynomial. This property is maintained in the pure Yang--Mills sector. We also generalize it to enable the inclusion of fermions. Positivity then depends on the representation of the gauge group which we choose for fermionic matter fields.

We believe our approach is useful as it is an approach to gauge theory amplitudes which does not rely
on on-shell recursion relations and methods of cut-reconstructibility. This has become the standard approach 
to gauge theory amplitudes in particular in the context of $N=4$ supersymmetric theories in recent years.
It nevertheless has shortcomings with respect  to the consideration of loop amplitudes, and seems rather restricted 
to the context of  planar graphs, preferably to be considered in a theory with a vanishing $\beta$-function.

Here, we offer an approach which is based on a mere study of the two Kirchhoff polynomials, the corolla polynomial of  
\cite{KrY}, the combinatorics of parametric renormalization \cite{BrKr}, and nothing else to obtain the renormalized integrand of a generic gauge theory amplitude from the amplitude of a scalar field theory with 3-valent vertices only.

In this paper, we introduce our approach. Comparison of graph- and cycle homology with BRST homology, an interpretation of Slavnov--Taylor  identities through corolla differentials, and computations of gauge theory amplitudes will be presented in future work, as will be a discussion of gravity Feynman rules from this approach.

\subsection*{Acknowledgments}
M.S.\ thanks Erik Panzer for many helpful discussions. D.K.\ is grateful to David Broadhurst, Francis Brown, Spencer Bloch, Erik Panzer, Oliver Schnetz and Karen Yeats
for shared collaborations and insights into the topics of this work. 
\subsection{Results}
We prove that the Feynman integrand at any finite loop order $n$ in gauge theory can be obtained from the study of two different complexes: graph homology and cycle homology,
acting on 3-regular connected scalar graphs.

Both complexes are reflected in the analytic structure of the Feynman integrand for an amplitude: residues in parametric space
along co-dimension $k$  hypersurfaces biject with the integrands of graphs having $k$ distinct 4-valent vertices and thus reflect graph homology, while cycle homology pairs with the structure of the corolla polynomial defined below in Sect.(\ref{Corolla}),
see also \cite{KrY}. 

In fact, we prove that there exists a corolla differential $D_\Gamma$ such that the Feynman integrand $I_\Gamma$ for connected 3-regular graphs $\Gamma$,
\begin{equation}
I_\Gamma=\frac{e^{-\frac{|N_\Gamma|_{\mathrm{Pf}}}{\psi_\Gamma}}}{\psi_\Gamma^2},\label{scalint}
\end{equation}
gives rise, when summed over connected graphs $\Gamma$, to the total gauge theory amplitude, using $D_\Gamma I_\Gamma$. The differential $D_\Gamma$ arises from the corolla polynomial $C_\Gamma$, which is a graph polynomial based on half-edge variables, upon replacing each half-edge variable by a suitable differential operator
assigned to any half-edge.\footnote{Equation (\ref{scalint}) above needs a correction for quadratic sub-divergences, which will be provided  in detail later on}

In a gauge theory, we must deal with the simultaneous requirements of unitarity of the scattering matrix $S$ and of covariance of fields which we assume to be representatives of the Poincar\'e group.
These two requirements severely restrict possible Feynman rules \cite{Weinberg1,Weinberg2,Weinberg3}. In particular, unphysical degrees of freedom 
must be eliminated from observable amplitudes.

For the 3- and 4-valent vertices of gauge boson interactions this requires that those vertices are not independent \cite{Cvit}. 
A 4-valent vertex has a Feynman rule which can be written in terms of two 3-valent vertices connected through a marked edge, see Eq.(\ref{eq:|}), summing over the three $s$, $t$ and $u$ channels.

Actually, we have to clarify a few conventions from the start:  we will mark in graphs edges and vertices in various ways and let 
\[
\mathcal{G}^{r,l}_{n\fourvertex,m\ghostloop}
\]  
be the set of all graphs with external half-edges specifying the scattering amplitude  $r$, with $l$ loops and $n$  4-gluon vertices, and $m$ ghostloops.
Similarly, we will indicate the number of marked edges and other qualifiers as needed.

If we want to leave a qualifier $l,n,m,\cdots$ unspecified (so that we consider the union of all sets with any number of such items), we replace it by $/$.
A similar notation will be used below for sums or series of graphs
\[
X^{r,l}_{n\fourvertex,m\ghostloop}
\]  
which are sums (for fixed $l$) or series (for $l=/$) of graphs weighted by symmetry, colour and other such factors as defined below. 
We dub such series combinatorial Green functions later on.

We can then consider graph homology, for shrinking an edge between two 3-valent vertices gives a 4-valent vertex. In particular, the elimination of an edge through the graph homology boundary $s$, see Def.(\ref{4-12}), must be balanced against graphs which contain that newly generated 4-valent vertex, which can again be homologically expressed through a generalized boundary $S$, $S^2=0$, defined in Def.(\ref{4-15}).  

More precisely, let $e$ be an edge connecting two 3-gluon  vertices in a graph $\Gamma$, $\chi^e_+$ be the operator which shrinks edge $e$, and we 
 extend $\chi_+^e$ to zero when acting between any other two vertices.
Let $
\chi_+=\sum_e\chi^e_+$, see Def.(\ref{4-1}), a sum over all internal  edges. Let $S$, with $S^2=0$,  be that generalized graph homology operator. Then, for a gauge theory amplitude $r$:
\begin{thm}\label{chithm}
Let $X^{r,n}_{0\fourvertex;j\ghostloop}$ be the sum of all 3-regular connected graphs, with $j$ ghost loops,   and with external legs determined by $r$
and loop number $n$, weighted by colour and symmetry, let $X^{r,n}_{/\fourvertex,j\ghostloop}$
be the same allowing for 3- and 4-valent vertices.
We have
\begin{align*}
i):\; e^{\chi_+}X^{r,n}_{0\fourvertex;j\ghostloop}&=X^{r,n}_{/\fourvertex;j\ghostloop},\\
ii):\; Se^{\chi_+}X^{r,n}_{0\fourvertex;j\ghostloop}&=0.
\end{align*}
\end{thm}
This theorem ensures that 3- and 4-valent vertices match to fulfil simultaneously the requirement of relativistic field theory and $S$-matrix theory.

These two requirements also demand that unphysical degrees of freedom propagating in closed loops do cancel. For this, we replace graph homology by cycle homology and can
proceed analogously.

Let $\delta^C_+$ be the operator which marks a cycle
$C$ through 3-valent vertices and unmarked edges, extend $\delta_+^C$ to zero on any other cycle.
Let $\delta_+=\sum_C \delta_+^C$ sum over all cycles, see Def.(\ref{4-22}). Let $t$ be the corresponding cycle differential $t$, defined in  Def.(\ref{4-27}). Let $T$, with $T^2=0$,  be the generalized cycle homology operator introduced in Def.(\ref{4-30}). Then:
\begin{thm}\label{deltathm}
Let $X^{r,n}_{j\fourvertex;0\ghostloop}$ be the sum of all connected graphs with $j$ 4-vertices contributing to amplitude $r$ and loop number $n$ and no ghost loops, weighted by colour and symmetry, $X^{r,n}_{j\fourvertex;/\ghostloop }$
be the same allowing for any possible number of ghost loops.
We have
\begin{align*}
i):\; e^{\delta_+}X^{r,n}_{j\fourvertex;0\ghostloop}&=X^{r,n}_{j\fourvertex;/\ghostloop},\\
ii):\; Te^{\delta_+}X^{r,n}_{j\fourvertex;0\ghostloop}&=0.
\end{align*}
\end{thm}
These two operations are compatible:
\begin{thm}\label{bicomplexthm}
\label{allthm}
\begin{itemize}
\item[i)] We have $[s,t]=[S,T]=0$ and 
$$
Te^{\delta_++\chi_+}X^{r,n}_{0\fourvertex;0\ghostloop}=0, \qquad Se^{\delta_++\chi_+}X^{r,n}_{0\fourvertex;0\ghostloop}=0.
$$
\item[ii)] Together, they generate the whole gauge theory amplitude from 3-regular graphs:
\begin{equation*} e^{\delta_++\chi_+}X^{r,n}_{0\fourvertex;0\ghostloop}=X^{r,n}_{/\fourvertex;/\ghostloop}=:X^{r,n}.
\end{equation*}
\end{itemize}
\end{thm}
\begin{rem}
$X^{r,n}$ is the only non-trivial element in the cycle and graph homology which we here construct. It will be an object of future study.
\end{rem}

Finally, we get the Feynman integrands in the unrenormalized and renormalized case for a gauge theory amplitude $r$ from 3-regular connected graphs of scalar fields.

\begin{thm} \label{GTFR}
The full Yang--Mills amplitude $\bar U_\Gamma$ for a graph $\Gamma$ can be obtained by acting with a corolla differential operator (see below) on the scalar integrand $U_\Gamma (\{\xi_e\})$ for $\Gamma$, setting the edge momenta $\xi_e = 0$ afterwards.  

Moreover, $\bar U_\Gamma$ gives rise to a differential form $J_\Gamma^{\bar U_\Gamma}$ and there exists a vector $H_\Gamma$ such that the unrenormalized Feynman integrand for the sum of all Feynman graphs contributing to the connected $k$-loop amplitude $r$ is  
\[
\Phi(X^{r,k})=
\sum_{|\Gamma|=k,\mathbf{res}(\Gamma)=r} \frac{\mathrm{colour}(\gamma)}{\mathrm{sym}(\Gamma)}
\int (H_\Gamma\cdot J_\Gamma^{\bar{U}_\Gamma}),\] 
The renormalized analogue is given by writing $\bar{U}_\Gamma^R$ instead of $\bar{U}_\Gamma$.
\end{thm}
\subsection{Organization of the paper}
The next section gives a detailed account of our graph-theoretic notions including graph homology. Then, we turn to the structure of Feynman rules for scalar fields
in the parametric representation, including renormalization as rigorously detailed for parametric
representations  in \cite{BrKr}. The peculiar situation in gauge theory is discussed in 
the fourth section, and a review of gauge theory from the viewpoint of Hochschild cocycles and the corresponding combinatorial Green functions is provided in the fifth. 
The sixth section combines these results with the corolla polynomial and the corolla differentials, culminating in Theorem \ref{GTFR} above. Short conclusions finish the paper.
\section{Graphs}
We first define the necessary graph theoretic notions.
\subsection{Vertices, edges, half-edges}
We consider connected graphs with labelled edges and vertices. We consider graphs as elements of a free commutative $\mathbb{Q}$-algebra $H$,
which is graded as a vectorspace by the first Betti number, the number of algebraically independent cycles in a graph.

For a graph $\Gamma$, we let $\Gamma^{[1]}=\Gamma^{[1]}_E\cup \Gamma^{[1]}_I \equiv E^\Gamma=E^\Gamma_E\cup E^\Gamma_I$ be the set of its external and internal edges
and let $\Gamma^{[0]}\equiv V^\Gamma$ be the set of its vertices.

We do not allow for internal edges which form self-loops (tadpoles): every internal edge can be considered as a pair of two distinct vertices.\footnote{The $\mathbb{Q}$-vectorspace of graphs with self-loops forms an ideal and co-ideal $I_{\mathrm{tad}}$, and we can effectively work in a quotient $H/I_{\mathrm{tad}}$.}

For a vertex $v\in V^\Gamma$, let $n(v)$ be the set of edges adjacent to $v$, and its cardinality $|n(v)|$ be the valence of $v$. For an edge $e\in E^\Gamma$,
we let $v(e)$ be the set $e^{[1]}\equiv e\cap V^\Gamma$. If $|v(e)|=2$, the edge  $e$ is an internal edge. 
If  $|v(e)|=1$, $e$ is an external edge, as self-loops are excluded. 

\begin{defi}
A pair $(v,e)$ with $e\in n(v)$ is called a half-edge. We let $H^\Gamma$ be the set of half-edges of $\Gamma$. 
\end{defi}
An internal edge $e$ defines two half-edges uniquely.\footnote{Note that we discard self-loops indeed. As a consequence, a chosen vertex and an edge incident to that vertex label a half-edge uniquely.} An ordering of the set $v(e)$ defines an orientation of that edge.
Reversing that ordering is called an edge swap. 

An oriented internal edge $e$ connects two vertices, which we call source,  $s(e)\in V^\Gamma$, and target, 
$t(e)\in V^\Gamma$, for an edge oriented  from source to target.

Half-edges will play an important role for us as our new graph polynomial, which we dub the corolla polynomial below,
is actually based on the half-edges of a graph.

\begin{defi}
The set of all half-edges incident to a given vertex, $$\mathrm{cor}(v):=\bigcup_{e\in n(v)}(v,e),$$ is called the corolla at $v$.
\end{defi}
\begin{defi}
We denote
\[
P_e:=\mathrm{cor}(s(e))\mathrm{cor}(t(e)),
\]
as an ordered pair of corollas.
\end{defi}
We call two such pairs $P_e,P_f$ disjoint if the edges $e$ and $f$ are disjoint.

External edges $e$ at $v$ are identified with the half-edge $(v,e)$ and are always regarded as oriented to the vertex.
\begin{defi}
\label{defn:j-regular}
We say that a graph is $j$-regular if all vertices have valence $j$, $|n(v)|=j$, $\forall v\in V^\Gamma$. 
\end{defi}
Let $\Gamma$ be 3-regular. Let $\mathcal{C}^\Gamma$ be the set of all its cycles (not circuits!). For $C\in \mathcal{C}^\Gamma$ and $v\in C$ a vertex, let $v_C\in H^\Gamma$ be the unique half-edge 
at $v$ not in $C_i$.

\begin{defi}
A graph is $n$-connected if it is connected after removal of any $n$ of its internal edges.
\end{defi}
\begin{rem}
A 2-connected graph is commonly called one-particle irreducible (1PI) in physics.\end{rem}

\subsection{Orientation and Cycles}
We need oriented graphs for two reasons: to define graph homology, and to have cyclic ordering at each corolla so that each half-edge incident to a vertex has a precursor and a successor at that vertex.

Let $M_k$ be an oriented  Riemann surface of genus $k$. We call a graph $k$-compatible, if it can be drawn on $M_k$ without self-intersections.
\begin{defi}
We say a graph is of genus $k$ if it is $k$-compatible but not $j$-compatible for any $j<k$. A planar graph is of genus $0$.
\end{defi}

\begin{defi}
An orientation $(\Gamma,or)$ of a graph $\Gamma$ is an ordering of $V^\Gamma$ together with an ordering of each $v(e)$ for all internal edges $e$.
Here, an ordering of $v(e)$ is the choice of one of the two possible bijections between $v(e)$ and the set $\{s(e),t(e)\}$.
\end{defi}
Two orderings distinguished by an even number of vertex permutations and edge swaps are equivalent. We write $(\Gamma,or)$
for an oriented graph. We set $(\Gamma,or)=-(\Gamma,or^\prime)$ for $or,or^\prime$ 
inequivalent orientations of the same graph.
  
\begin{lemma}(Conant, Vogtmann \cite{CV})
An ordering of a $3$-regular graph is equivalent to a cyclic ordering of all its corollas.
\end{lemma}
\begin{defi}
For an oriented 3-regular graph of genus $k$ and $e\in n(v)$, let $e_+$ and $e_-$ be the edges before ($e_-$) and after ($e_+$) edge $e$ in the cyclic ordering of the corolla at $v$, induced by the orientation of $M_k$. 
\end{defi}
\begin{rem}
Note that such an orientation of a graph is compatible with a strict ordering of edges: the ordering of vertices orders the pair $s(e),t(e)$ lexicographically
(with $s(e)<t(e)$ say), while multiple edges having the same source and target are ordered by the orientation of the underlying Riemann surface.
\end{rem}
\subsection{Graph homology}
We will use various homologies on graphs, with corresponding boundary operators, and suitable variants to study the filtrations
of graphs by the number of ghost cycles and the number of internal 4-valent vertices. 
We start with standard graph homology for scalar graphs with 3- and 4-valent vertices.
\subsubsection{Graph homology $\tilde s$ (following Conant, Vogtmann \cite{CV})}
For an edge $e$ in a graph $\Gamma$, let $\Gamma_e$ be the graph where $e$ shrinks to zero length.
Its orientation is obtained as follows:
we permute vertex labels collecting  signs until the edge $e$ connects vertex 1, $s(e)=1$, to vertex 2, $t(e)=2$. Let $\sigma$ be the sign of the necessary permutations. Then 
we shrink edge $e$ and the so-obtained vertex is labelled 1. We inherit all remaining edge orientations and the ordering of vertices remains unchanged, with vertices $3,4,\ldots,|V^\Gamma|$ relabelled to $2,3,\ldots,|V^\Gamma|-1$. 
This defines an orientation of $\Gamma_e$. If $\sigma$ is negative, we change the orientation by en edge swap.

For an oriented  graph $\Gamma$, let 
\begin{equation*}
\tilde s\Gamma=\sum_{e\in E_I}\Gamma_e,
\end{equation*}
 be a sum of  graphs obtained by shrinking edge $e$ and assigning
the orientation as above.   Graph homology comes from the classical result
\begin{thm}(graph homology)
$\tilde s\circ \tilde s=0$.
\end{thm}
\subsubsection{Graphs with marked vertices}
We can restrict graph homology to graphs with vertex valence bounded by $|n(v)|\leq 4$ by setting all terms which have 
vertices of valence higher than four  to zero in the image of $\tilde s$. 

For such graphs, let $V_3$ be the set of 3-valent vertices, $V_4$ be the set of 4-valent vertices, so $V^\Gamma=V_3\cup V_4$.

Let $W_4\subset V_4$ be a chosen subset. If a vertex $w\in W_4$, call it a marked vertex. A pair $(\Gamma, W_4)$ 
is a graph with marked vertices.

Let $u_w$ be the obvious map which removes the marking at vertex $w$, and $W_4(\Gamma)$ be the map which marks the vertices of $\Gamma$ which are in $W_4$. 
So $u_w (\Gamma,W_4)=(\Gamma, W_4-w)$. Note that an orientation of $\Gamma$
induces an ordering of the set $W_4$. Set $\sigma(w)=j$ if and only if $w\in W_4$ is in the $j$-th place in that ordering.

We extend graph homology with boundary $\tilde s$ to graphs with marked 4-valent vertices and boundary $\tilde S$ by setting 
\begin{equation}
\tilde S(\Gamma,W_4)=(s(\Gamma),W_4)+(-1)^{|V_4|}\sum_w (-1)^{\sigma(w)}(\Gamma,W_4-w).\label{S-homology} 
\end{equation}
\begin{prop}
$\tilde S$ is a differential:
$$\tilde S^2=0.$$
\end{prop}
\begin{proof}
We compute:
\begin{eqnarray*}
\tilde S^2(\Gamma,W_4)&=&
 \tilde S \left( (\tilde s(\Gamma),W_4)+(-1)^{|V_4|}\sum_w (-1)^{\sigma(w)}(\Gamma,W_4-w)\right)   \nonumber\\
 & = & (-1)^{|V_4|+1}\sum_w (-1)^{ \sigma(w)}(\tilde s(\Gamma),W_4-w)+(-1)^{|V_4|}\sum_w (-1)^\sigma(w)(\tilde s(\Gamma),W_4-w)\nonumber\\
 & & +(-1)^{|V_4|}\sum_{w\in W_4} (-1)^{\sigma(w)}\sum_{u\in W_4-w}(-1)^{\sigma(u)}(\tilde s(\Gamma),W_4-w-u)\\
  & = & 0,\nonumber
\end{eqnarray*}
where the last line vanishes due to the ordering of vertices in $W_4$, which makes sure that each pair $(w,u)$ appears twice with a relative sign.
\end{proof}
We will soon see (Prop.(\ref{Shom})) that the sum over all connected graphs is in the kernel of $S$, the gauge-theory equivalent of $\tilde S$.

\begin{rem}
While the above operators $\tilde s,\tilde S$ were defined on scalar graphs, we also have variants $s$, $S$ for gauge theory graphs where we have internal gauge boson or ghost propagators or fermion propagators. There are obvious restrictions then to shrink only edges which connect two 3-gluon vertices. The detailed homology operations available in such circumstances are exhibited in Section \ref{gtgr}.
\end{rem}

We now make graphs into a Hopf algebra.

\subsection{Algebra of graphs}
As we said before, we consider graphs as generators of a free commutative $\mathbb{Q}$-algebra of graphs $H$. 
We write $\One$ for the unit represented by the empty set, with disjoint union of graphs furnishing the product.

\begin{defi}
The number of external edges $n_E:=|E^\Gamma_E|$ assigns the weight 
\begin{equation*}
\omega^\Gamma:=4-n_E=|V_3^\Gamma|-2|E^\Gamma_I|+4|\Gamma|,
\end{equation*}
 to a graph.
\end{defi}

A graph has positive valuation if $\omega^\Gamma\geq 0$. 
A graph has $\iota$-valuation if $\omega^\Gamma\geq\iota$.
Note that the valuation of a graph is invariant under shrinking edges.

The most obvious Hopf algebra structure is given by a co-product  based on subgraphs of non-negative weights:
\begin{equation*}
\Delta\Gamma=\Gamma\otimes\One+\One\otimes\Gamma+\sum_{\emptyset\not=\gamma=\prod\gamma_i,\omega^{\gamma_i}\geq 0}\gamma\otimes \Gamma/\gamma,
\end{equation*} 
is a coproduct, for a connected commutative Hopf algebra with unit $\One$ and the span of all non-trivial graphs as augmentation ideal, as usual \cite{Kr}.
\begin{rem}
Note that we assign to a graph $\gamma$ the powercounting weight $\omega^\Gamma$ of four dimensional gauge theory: a graph with four external half-edges is logarithmically divergent, a graph with three external half-edges is linear divergent, and a graph with two external half-edges is quadratically divergent.
\end{rem}

\begin{defi}\label{forest}
Let $f:=\{\gamma_i\}$ be a subset of proper positive valued 1PI subgraphs $\gamma_i\subset\Gamma$ such that any two elements $\gamma_i,\gamma_j$
of $f$ fulfill: $\gamma_i\cap\gamma_j=\emptyset$, or $\gamma_i\subset\gamma_j$ or $\gamma_j\subset\gamma_i$.Then, $f$ is called a forest. 
\end{defi}

\begin{defi}
It is maximal, if and only if $\Gamma/f$ contains no positive valued proper subgraph. It is complete, if it contains all positive valued proper subgraphs of all its elements.
\end{defi}

We denote by $|f|$  the number of elements of $f$ and by $\mathcal{F}^\Gamma$ the set of all forests of $\Gamma$.
For a union of graphs $\gamma=\bigcup_i\gamma_i$ we say it has $\iota$ valuation if all its components have.
There are Hopf algebras for any $\iota$-valuation:
\begin{equation*}
\Delta_\iota(\Gamma)=\Gamma\otimes\One+\One\otimes\Gamma+\sum_{\gamma,\omega(\gamma)\geq\iota}\gamma\otimes\Gamma/\gamma.
\end{equation*}
In particular the antipode for positive valued graphs can be written as
\begin{equation*}
S(\Gamma)=-\sum_{f\in\mathcal{F}^\Gamma}(-1)^{|f|}f\times \Gamma/f,
\end{equation*}
where the sum includes the empty set. 

If the number of external edges of a subgraph $\gamma$ is greater than two, $|\gamma^{[1]}_E|>2$, then $\gamma$ shrinks to a
vertex in $\Gamma/\gamma$. If it equals two, the two external edges are identified to a single edge in $\Gamma/\gamma$.\footnote{If we were to have massive particles,
we had to blow up notation slightly.}

These Hopf algebras on scalar graphs straightforwardly generalize to gauge theory graphs and in particular act on the sum of all graphs contributing to a given amplitude ---the combinatorial Green function--- filtered by the number of 4-valent vertices and the number of ghost loops. 

From now on, we demand that orientations of a graph $\Gamma$ are such that for any proper subgraph $\gamma\subset\Gamma$ with two external edges $e,f$, $|\gamma^{[1]}_E|=2$, the edges $e,f$ form a consistently oriented edge in $\Gamma/\gamma$. Then, the orientation of $\Gamma$ determines the orientations of 
$f$ and $\Gamma/f$, for all forests $f$. 
\begin{defi}
\label{epm}
We call $e^\gamma_+$ the edge pointing towards $\gamma$,  $e^\gamma_-$ the edge pointing away from it.
\end{defi}
\begin{lem}\label{quadrsub}
Let $\gamma_1,\gamma_2\subset\Gamma$ be two proper propagator subgraphs, $|\gamma^{[1]}_{i,E}|=2$.
Then, either: $\gamma_1\subset\gamma_2$ or $\gamma_2\subset\gamma_1$ or $\gamma_1\cap\gamma_2=\emptyset$. 
\end{lem}
\begin{proof}
This follows from \cite[Lemma 51]{BrKr}.
\end{proof}

\section{Feynman rules for scalar integrands}
First, some general remarks. For $n$ external edges and $l$ loops, an overall factor $c=(-i)^{n+3(l-1)}g^{n+2(l-1)}$ is not explicitly given below.
All momentum integrals are understood in Euclidean space, all parametric integrals over the real simplex $\sigma:\,\{A_e>0\}$, with boundary
$\prod_{e\in\Gamma^{[1]}}A_e=0$. A pairing between a parametric integrand and a simplex as in 
\[
\Phi(\Gamma) = \int_\sigma\mathrm d\underline A_{\Gamma^{[1]}}\ I_\Gamma,
\]
means just that: the pair of the simplex and a form, and is to be regarded as an honest integral only when the integrand is replaced by its suitably renormalized form
as defined below,
typically indicated by a superscript ${}^R$, 
so that the integral actually exists. We often simply write $\int$ for $\int_\sigma$.

Let $\Gamma$ be a 3-regular graph as defined in Definition \ref{defn:j-regular}.
In order to define the momenta, choose an orientation on $\Gamma$, which we represent by $\epsilon$ (the incidence matrix):
\begin{samepage}
\begin{equation*} \epsilon_{ve} = \begin{cases}
1 & \text{if the vertex $v$ is the endpoint of the edge $e$,} \\
-1 & \text{if the vertex $v$ is the starting point of the edge $e$,} \\
0 & \text{if $e$ is not incident on $v$.}
\end{cases}\end{equation*}
\end{samepage}

Furthermore, choose a basis of loops  $L\subset\mathcal{C}^\Gamma$ of $l=:|\Gamma|$ independent loops of $\Gamma$ and for each $\ell\in L$ an orientation given by $\epsilon^\ell$.\footnote{$\epsilon^\ell$ is such that $\epsilon^\ell_{ve_1}=-\epsilon^\ell_{ve2}$, where $e_1$ and $e_2$ are the two edges adjacent to $v$, which are inside $\ell$.}

The Feynman amplitude of $\Gamma$ is\footnote{$d$ is the dimension of spacetime, and can be safely set to four as later on we will renormalize the parametric integrand before integrating. If the reader wishes to integrate first, $d$ can be used as a regulator, which is mathematically questionable though \cite{KrP}
when combined with minimal subtraction despite its popularity in physics.}
\begin{equation*}
\Phi(\Gamma) =\mathbf{Q}^E_\Gamma \int\frac{\mathrm d\underline k_L}{(2\pi)^{dl}} \bigg(\prod_{e\in\Gamma^{[1]}}\frac1{\xi_e^{\prime2}}\bigg),
\end{equation*}
where
\begin{equation*} \int\frac{\mathrm d\underline k_L}{(2\pi)^{dl}} := \prod_{\ell\in L} \int_{\mathbb{R}^4}\frac{\mathrm d^dk_\ell}{(2\pi)^d}, \end{equation*}and
\begin{equation*} \epsilon_{ve}\xi'_e := \epsilon_{ve}\xi_e + \sum_{\substack{\ell\in L\\ \ell^{[1]}\ni e}} \epsilon^\ell_{ve} k_\ell. \end{equation*}
Note that we also include the external momenta, which is useful for our purposes. This gives just an overall factor above,
\[
\mathbf{Q}^E_\Gamma=\prod_{e\in\Gamma^{[1]}_{\rm ext}}\xi_e^2.
\]
We define 
\[
-\sum_{e\in\Gamma^{[1]},e\,\mathrm{adj.\, to}\,v}\epsilon_{ve}\xi_e=:\xi_v,
\] 
to be the external momentum for every vertex $v$ of our graph.
Note that if $v$ is a vertex to which an external edge $e$ is adjacent, $\xi_v$ does not equal $\xi_e$. 

In Schwinger parametric form, the Feynman amplitude for generic $\xi_e$  is
\begin{equation*}\label{param}
\Phi(\Gamma) = \int\mathrm d\underline A_{\Gamma^{[1]}}\ I_\Gamma,
\end{equation*}
where the integration is a short-hand notation for
\begin{equation*} \int\mathrm d\underline A_{\Gamma^{[1]}} = \prod_{e\in \Gamma^{[1]}} \int_0^\infty\mathrm d\!A_e
=\int_\sigma d\!A_e, \end{equation*}
and the integrand is obtained from integrating the universal quadric (\cite{BEK}) 
\begin{equation*}
Q_\Gamma:=\sum_{e\in\Gamma^{[1]}} A_e\xi_e^{\prime2},\label{uquadric}
\end{equation*}
so that analytically, we study 
the expression
\[
I_\Gamma:=\mathbf{Q}_\Gamma^E \int\frac{\mathrm d\underline k_L}{(2\pi)^{dl}}\prod_{e}e^{-Q_\Gamma}\prod_{v}\delta^4\left(\sum\epsilon_{ev}k(e)\right),
\]
which corresponds to a graph where at each vertex $v$, an external momentum $\xi_v$ is entering, and which can be written in the form of Eq.\eqref{scalint}, by use of the graph polynomials, to which we now turn.
\subsection{The first Kirchhoff polynomial $\psi_\Gamma$}
For the first Kirchhoff polynomial consider the short exact sequence
\begin{equation}
0\to H^1\to\mathbb{Q}^E\overbrace{\to}^{\partial}\mathbb{Q}^{V,0}\to0.
\end{equation}
Here, $H^1$ is provided by a chosen basis for the algebraically independent loops of a graph $\Gamma$.
 $E=|E^\Gamma|$ is the number of edges and 
$V=|V^\Gamma|$ the number of vertices, so $\mathbb{Q}^E$ is an $E$-dimensional $\mathbb{Q}$-vectorspace generated by the edges, similar $\mathbb{Q}^{V,0}$
for the vertices with a side constraint setting the sum of all vertices to zero.

Consider the matrix (see \cite{BEK,BlKr})
\begin{equation*}
N_0\equiv (N_0)_{ij}=\sum_{e\in l_i\cap l_j}A_e,
\end{equation*}
for $l_i,l_j\in H^1$.

Define the {\it first Kirchhoff polynomial} as the determinant
\begin{equation*}
\psi_\Gamma:=|N_0|.
\end{equation*}
\begin{prop}(\cite[Prop.2.2]{BEK})
The first Kirchhoff polynomial can be written as
$$\psi_\Gamma=\sum_{T}\prod_{e\not\in T}A_e$$
where the sum on the right is over spanning trees $T$ of $\Gamma$.
\end{prop}

\subsection{The second Kirchhoff polynomial $\phi_\Gamma$ and $|N|_{\mathrm{Pf}}$}
Let $\sigma^i$, $i\in 1,2,3$ be the three Pauli matrices, and $\sigma^0=\One_{2\times 2}$ the unit matrix.

For the second Kirchhoff polynomial,
augment the matrix $N_0$ to a new matrix $N$ in the following way:
\begin{enumerate}
\item Assign to each edge $e$ a quaternion \[
\mathbf{q}_e:=q_0\sigma^0-i\sum_{j=1}^3 q_j\sigma^j,\]
so that  $\xi_e^2\One_{2\times 2}=\mathbf{q}_e\bar{\mathbf{q}}_e$,
 and to the loop $l_i$,
the quaternion \[
u_i=\sum_{e\in l_i} A_e q_e.\]\\
\item Consider the column vector $u=(u_i)$ and the conjugated transposed row vector 
$\bar{u}$. Augment $u$ as the rightmost column vector to $M$, and $\bar{u}$ as the bottom row vector.
\item Add a new diagonal entry at the bottom right  $ \sum_e q_e\bar{q}_e A_e$.
\end{enumerate}

Note that by momentum conservation, to each vertex, we assign a momentum $\xi_v$.
and a corresponding quaternion $\mathbf{q}_v$.
\begin{rem}
Note that we use that we work in four dimensions of space-time, by rewriting the momentum four-vectors in a quaternionic basis.
\end{rem}
The matrix $N$ has a well-defined Pfaffian determinant (see \cite{BlKr})
with a remarkable form obtained for generic $\xi_e$ and hence generic $\xi_v$:
\begin{lem}(\cite[Eq.3.12]{BlKr})\label{gensym}
$$|N|_{\mathrm{Pf}}=-\sum_{T_1\cup T_2}\left(\sum_{e\not\in T_1\cup T_2} \tau(e)\xi_e\right)^2\prod_{e\not\in T_1\cup T_2}A_e,$$
where $\tau(e)$ is $+1$ if $e$ is oriented from $T_1$ to $T_2$ and $-1$ else.
\end{lem} 
\begin{proof}
See \cite{BlKr} .
\end{proof}
Note that $|N|_{\mathrm{Pf}}=|N|_{\mathrm{Pf}}(\{\xi_v\})$ is a function of all $\xi_v$, $v\in \Gamma^{[0]}$.
It gives the second Symanzik polynomial upon setting the $\xi_e$ in accordance with the external momenta:
\[
Q: \xi_e\to\xi_e+q_e,
\] 
and setting $\xi_e=0$ for all edges $e$ afterwards.
\begin{prop}
$$\phi_\Gamma\One_{2\times 2}=Q\left(|N|_{\mathrm{Pf}}\right)_{|\xi_e=0}=-\sum_{T_1\cup T_2}Q(T_1) \bar{Q}(T_1)\prod_{e\not\in T_1\cup T_2}A_2,$$
where the sum is over spanning two-forests $T_1\cup T_2$, so $T_1$, $T_2$ are two disjoint non-empty trees which together cover
all vertices of $\Gamma$. Here, $Q(T_i)=\sum_{v\in T_i^0}\mathbf{q}_v$, $v$ a vertex of $T_i$,  the momentum $q_v$
incoming at that vertex expressed in the quaternionic basis.
\end{prop}
\begin{rem}
For a Feynman graph $\Gamma$  contributing to a scattering amplitude $r$ with $k$ external edges adjacent to $m\leq k$ vertices of $\Gamma$, $\phi_\Gamma$ is the quantity of interest.   The quantity $|N_\Gamma|_{\mathrm{Pf}}$, which assigns  a momentum $\xi_v$ to every vertex (not only to those $m$ which have external edges attached) is more natural from a graph-theoretic viewpoint. For us, it has the added advantage that it assigns a four-momentum $\xi_e$ to every edge.
Derivatives with respect to such four-momenta will generate gauge theory Feynman rules for us below.
\end{rem}
\begin{cor}\label{linear}
$\frac{\partial}{\partial_{{\xi_e}_\mu}}|N|_{\mathrm{Pf}}$ is linear in $(A_e\xi_e^\mu)$.\\
For edges $e\not= f$, $\frac{\partial^2}{\partial_{{\xi_e}_\mu}\partial_{\xi(f)_\nu}}|N|_{\mathrm{Pf}}$ is linear in $(A_eA_fg^{\mu\nu})$ and  constant in 
all $\xi_e$.\\
For $e=f$,  $\frac{\partial^2}{\partial_{{\xi_e}_\mu}\partial_{{\xi_e}_\nu}}|N|_{\mathrm{Pf}}$ is linear in $g^{\mu\nu}$, constant in all $\xi_e$ and linear in $A_e$.
\end{cor}
\begin{proof}
This follows readily from Lemma \ref{gensym}.
\end{proof}
\begin{exa}
We let
$$\Gamma=\dunce
$$
with $\{1,2,3\},\{1,2,4\}$ a basis for the cycles of $\Gamma$. Then,
$$
N:=\left(\begin{array}{cc}
N_0:=\left( \sum_{e\in h_i\cap h_j}A_e\right)_{ij} & \sum_{e\in h_j}{\mu}_eA_e\\
\sum_{e\in h_j}\bar{\mu}_eA_e & \sum_{e\in\Gamma^{[1]}}\bar{\mu_e}\mu_e A_e\end{array}\right)$$
so
$$
N_\Gamma=\left(\begin{array}{ccc}
A_1+A_2+A_3 & A_1+A_2 & A_1\mu_1+A_2\mu_2+A_3\mu_3\\
A_1+A_2 & A_1+A_2+A_4 & A_1\mu_1+A_2\mu_2+A_4\mu_4\\
A_1\bar{\mu}_1+A_2\bar{\mu}_2+A_3\bar{\mu}_3 & A_1\bar{\mu}_1+A_2\bar{\mu}_2+A_4\bar{\mu}_4 &
\sum_{i:=1}^4 A_i \bar{\mu}_i\mu_i 
\end{array}
\right)$$
\[
\psi_\Gamma=(A_1+A_2)(A_3+A_4)+A_3A_4=\sum_{\mathrm{sp.Tr.}T}\prod_{e\not\in T}A_e\]
and
\begin{eqnarray*} \phi_\Gamma & = & -(A_3+A_4)A_1A_2 p_a^2+A_2A_3A_4p_b^2+A_1A_3A_4p_c^2\\
 & = & \sum_{\mathrm{sp.2-Tr.}T_1\cup T_2}Q(T_1)\cdot Q(T_2)\prod_{e\not\in T_1\cup T_2}A_e.
\end{eqnarray*}
\end{exa}
\begin{rem}
We use the two Symanzik polynomials to integrate out loop momenta in scalar integrands upon using
\[
\frac{1}{{\xi_e^\prime}\cdot {\xi_e^\prime}}=\int_0^\infty e^{-A{\xi_e^\prime}\cdot {\xi_e^\prime}}d\!A,\,
{\xi_e^\prime}\cdot {\xi_e^\prime}={\xi_e^\prime}^2,
\]
so that we get back Eq.(\ref{scalint}).
\end{rem}

\subsection{Correction for quadratic subdivergences}
The renormalized massless quadratically divergent two-point self-energy 
\[
\Sigma_R(q^2,\mu^2)=q^2\sigma(q^2/\mu^2),
\]
vanishes at $q^2=0$, and $\sigma=\Sigma_R/q^2$ vanishes at $q^2=\mu^2$:
\begin{equation}
\Sigma_R(q^2,\mu^2)_{|q^2=0}=0,\;\; \left(\frac{1}{q^2}\Sigma_R(q^2,\mu^2)\right)_{|q^2=\mu^2}=0.\label{qurencon}
\end{equation} 
This fixes the two renormalization conditions for any graph contributing to a massless quadratically divergent two-point functions which we employ.
Transversality of the gluon propagator self-energy 
\[\Pi_{\mu\nu}=q^2(g_{\mu\nu}-\hat{q}_\mu\hat{q}_\nu)\Pi(q^2/\mu^2),\]
renders this as a single condition $\Pi(1)=0$ on the sum of all contributing graphs at each loop order.

All graphs contributing to other Green functions are renormalized by simple subtractions at chosen kinematics (see \cite{BrKr} for a complete discussion).

Let us now discuss the correction factor for each quadratically divergent (sub)graph.
The scalar integrand in parametric variables is 
\[
I_\Gamma=\frac{e^{-\frac{|N_\Gamma|_{\mathrm{Pf}}}{\psi_\Gamma}}}{\psi_\Gamma^2}d\!A_1\cdots d\!A_{|\Gamma^{[1]}|}.
\] 
In fact, for gauge theory we need to study integrands which are slightly more general:
\[
\tilde I_{\Gamma,F=}:
F \frac{e^{-\frac{|N_\Gamma|_{\mathrm{Pf}}}{\psi_\Gamma}}}{\psi^2_\Gamma}d\!A_1\cdots d\!A_{|\Gamma^{[1]}|}. 
\]
Here, 
\[
F=\frac{F_N(\{A_e\})}{F_D(\{A_e\})}
\] 
is a rational function of the variables $A_e$, which is a quotient of homogeneous polynomials  $F_N,F_D$.
In fact, our integrand in Yang--Mills or gauge theory will give us a finite sum of such terms with different $F$.

We want to analyse the degree of the form 
\[
A_\Gamma^{F}= \frac{F}{\psi^2_\Gamma}d\!A_1\cdots d\!A_{|\Gamma^{[1]}|}.
\]
Consider a set of variables $\gamma^{[1]}_I\subseteq\Gamma^{[1]}_I$, for $\gamma\subset\Gamma$ a subgraph (the case $\gamma=\Gamma$ is allowed).
For a polynomial function $f=f(\{A_e\})$ we let $|f|_{\gamma^{[1]}_I}$ be twice\footnote{The mass dimension of an edge variable $A_e$ is $-2$.}  the degree of that polynomial in the variables indicated.
For a rational function which is a quotient $f_1/f_2$ of such functions, 
we have $|f_1/f_2|_{\gamma^{[1]}_I}=|f_1|_{\gamma^{[1]}_I}-|f_2|_{\gamma^{[1]}_I}$.
Also, $|d\!A_1\cdots d\!A_{|\Gamma^{[1]}|}|_{\gamma^{[1]}_I}=|\gamma^{[1]}_I|$.

\begin{prop}
For any 
$A_\Gamma^{F}$ appearing in the integrand of a graph $\Gamma$, we have $|A_\Gamma^{F}|_{\gamma^{[1]}_I}\leq \omega^\gamma$. 
\end{prop}
\begin{proof}
All short-distance singularities of a Feynman integrand correspond to forests  and their degree is bounded by the powercounting $\omega^\gamma$, for all $\gamma$ appearing as components of forests.
\end{proof}
Now assume $|A_\Gamma^{F}|_{\gamma^{[1]}_I}=2$ for some $\gamma$.
Such quadratically divergent short-distance singularities can only originate from self-energy subgraphs. Such graphs have two distinguished vertices at which external edges are adjacent. Moreover, it follows from Lemma 
\ref{quadrsub} that quadratically divergent subgraphs are either disjoint or nested. We have immediately from the definition of the second Kirchhoff polynomial $\phi_\gamma=:q^2\psi_{\gamma_\bullet}$, which defines $\gamma_\bullet$ to be the graph where those two vertices are identified \cite{BrKr}.

For each $|A_\Gamma^{F}|_{\gamma^{[1]}_I}$ above, let $2^{F}_\Gamma$ 
be the set of subgraphs $\gamma\subseteq\Gamma$ such that 
$|A_\Gamma^{F}|_{\gamma^{[1]}_I}=2$, for all $\gamma\in 2^{F}_\Gamma$.
Define, for all $A_\Gamma^F$,
\begin{equation}
\bar{A}_\Gamma^{F}:=A_\Gamma^{F}\prod_{\gamma\in 2^{F}_\Gamma}\frac{\psi_{\gamma^\bullet}}{\psi_\gamma A_{e^\gamma_+}}.\label{corrfac}
\end{equation}
\begin{lem}
The form $\bar{A}_\Gamma^{F}$ has only logarithmic poles along divergent subgraphs including self-energy subgraphs. 
Renormalizing these remaining logarithmic poles of  self-energy  subgraphs at a fixed $\mu^2$ imposes renormalization conditions Eq.(\ref{qurencon}).\footnote{We consider the case of massless propagators. For the incorporation of masses, see
\cite{BlKr}.}
\end{lem}
\begin{proof}
A partial integration with respect to the quadratic subgraph variables renders its overall divergence logarithmic. The boundary term is eliminated by our renormalization conditions, which adopt the BPHZ conditions of massive propagators to the renormalization conditions adopted here for massless gluons. See \cite[Section 3.5]{BrKr}. 
\end{proof}
Let us exhibit this in an example. Integrating the universal quadric, Eq.(\ref{uquadric}), only the propagators for the quadratically divergent subgraph $\gamma$ leaves us with a contribution
\[
\frac{1}{\xi^\prime_{e^\gamma_-}}F \frac{e^{-\frac{{\xi_e^\prime}^2\psi_{\gamma_\bullet}}{\psi_\gamma}}(d\!A)_\gamma}{\psi_\gamma^2} \frac{1}{\xi^\prime_{e^\gamma_+}},
\]
with $e^\gamma_\pm$ defined in Def.(\ref{epm}).

Setting $A_i=t_\gamma a_i$ isolates the quadratic divergence:
\[
\frac{1}{\xi^\prime_{e^\gamma_-}} \frac{Fe^{-t_\gamma\frac{{\xi_e^\prime}^2\psi_{\gamma_\bullet}}{\psi_\gamma}}\Omega_\gamma\wedge t_\gamma}{t_\gamma^2\psi_\gamma^2} \frac{1}{\xi^\prime_{e^\gamma_+}}.
\]
Consider the integral $\lim_{c_\gamma\to 0}\int_{c_\gamma}^\infty$ against the above, and partially integrate with respect to $t_\gamma$. 
This gives, modulo terms which vanish when $c_\gamma\to 0$,  a boundary term
\[
\frac{1}{\xi^\prime{e^\gamma_-}}\frac{F}{\psi_\gamma^2}\left(1-{\xi_e^\prime}^2\frac{\psi_{\gamma_\bullet}}{\psi_\Gamma}\right)\frac{1}{\xi^\prime_{e^\gamma_+}}
\] 
which is polynomial in ${\xi_e^\prime}$ and hence vanishes in our renormalization conditions.

What remains though is the logarithmically divergent
\[
\frac{1}{\xi^\prime_{e^\gamma_-}} \frac{F \psi_{\gamma_\bullet} {\xi_e^\prime}^2e^{-t_\gamma\frac{{\xi_e^\prime}^2\psi_{\gamma_\bullet}}{\psi_\gamma}}\Omega_\gamma\wedge t_\gamma}{t_\gamma\psi_\gamma^4} \frac{1}{\xi^\prime_{e^\gamma_+}}.
\]
which justifies the result above, upon using momentum conservation ${\xi_e^\prime}=\xi^\prime_{e^\gamma_+}$, and 
\[ 
\frac{\xi^\prime_e}{\xi^\prime_{e^\gamma_+}}=1=\frac{1}{2\pi i}\oint_{\gamma_e} \frac{e^{-{\xi_e^\prime}^2A_e}}{A_e},
\]
with $\gamma_e$ a curve which picks the residue at $A_e=0$. Collecting such residues will be automatic in our approach below.

\begin{rem}
For a self-energy graph $\gamma$ and an $F$ such that $|A_\gamma^F|_\gamma=2$, and $\tilde{I_\Gamma}^F=\tilde{I_\Gamma}^F(q^2)$,
the factor Eq.(\ref{corrfac}) is in accordance with subtractions $(\tilde{I_\Gamma}^F(q^2)-\tilde{I_\Gamma}^F(0))/q^2-
(\tilde{I_\Gamma}^F(\mu^2)-\tilde{I_\Gamma}^F(0))/\mu^2$, in accordance with the conditions (\ref{qurencon}).
\end{rem}
\subsection{Renormalization}
By summing over all $F$, we have finally constructed an integrand which we will call  $\bar{U}_\Gamma$ 
with log-poles only along any forest. We can hence render it finite by the usual forest formula. We define
\[
U^{\mathrm{R}}_\Gamma:=\sum_{f\in\mathcal{F}_\Gamma}(-1)^{|f|}Q(\bar{U}_{\Gamma/f})Q_0(\bar{U}_f),
\]
where for $f=\bigcup_i\gamma_i$, $U_f=\prod_i U_{\gamma_i}$ and $Q,Q_0$ are the maps 
\[
Q:\xi_e\to\xi_e+q_e,\; Q_0:\xi_e\to\xi_e+q_{e,0},
\]
where $q_e$ are the momenta as prescribed by the amplitude under consideration, $q_{e,0}$ those prescribed by the renormalization scheme
for this amplitude.

\begin{rem}
Note that a single $\tilde I_{\Gamma,F}$ renormalizes similarly,
\[
\tilde I_{\Gamma,F}^R=\bar{F}
\sum_{f\in\mathcal{F}_\Gamma}(-1)^{|f|}Q\left(\frac{e^{-\frac{|N_{\Gamma/f}|_{\mathrm{Pf}}}{\psi_{\Gamma/f}}}}{\psi_{\Gamma/f}^2}\right)
Q_0\left(\frac{e^{-\frac{|N_f|_{\mathrm{Pf}}}{\psi_{f}}}}{\psi_{f}^2}\right),
\]
with notation as in \cite{BrKr}. Note in particular how derivatives with respect to four-vectors $\xi_e$ still act on  $\tilde I_{\Gamma,F}^R$:
for any monomial in derivatives $X_E:=\prod_{e\in E}\partial_{{\xi_e}_{\mu(e)}}$, we have
\begin{equation}\label{xeeq}
X_E\tilde I_{\Gamma,F}^R=\bar{F}
\sum_{f\in\mathcal{F}_\Gamma}(-1)^{|f|}\left(\prod_{e\in E\cap {\Gamma/f}^{[1]}}\partial_{{\xi_e}_{\mu(e)}} \right) Q\left(\frac{e^{-\frac{|N_{\Gamma/f}|_{\mathrm{Pf}}}{\psi_{\Gamma/f}}}}{\psi_{\Gamma/f}^2}\right)
\left(\prod_{e\in E\cap {f}^{[1]}}\partial_{{\xi_e}_{\mu(e)}} \right)Q_0\left(\frac{e^{-\frac{|N_f|_{\mathrm{Pf}}}{\psi_{f}}}}{\psi_{f}^2}\right).
\end{equation}
This is particularly useful for future work when combined with the projective renormalized integrand, see \cite{BrKr}.

Also, even when $\Gamma$ is a graph without external ghost lines, some of the forests appearing above can correspond
to graphs with open ghost lines. Then, Eq.(\ref{xeeq}) reconfirms the formula Eq.(\ref{openlines}) below for amplitudes
with open ghost (or fermion, in an obvious modification) lines.   
\end{rem}
\begin{rem}
The regularized integrand in dimensional regularization is obtained by multiplying $\bar{U}_\Gamma$ by $1/\psi_\Gamma^{(4-d)/2}$,
and treating the Clifford algebra accordingly.
Evaluating the parametric integrals on the regularized integrand first  and renormalizing in accordance with our renormalization prescription
produces the same renormalized results as above. Using a minimal subtraction scheme is different though. See \cite{KrP} for a discussion of this point.
\end{rem}
\section{Gauge Theory Graphs}\label{gtgr}
We now turn to graphs in gauge theory, as contrasted to 3-regular graphs in scalar field theory. While the latter were graphs which can be regarded as corollas with three half-edges, connected by gluing two half-edges from different corollas to an internal edge $e$ which hence determine a pair of corollas $P_e$,
the former are graphs with 3- and 4-valent vertices. 

Again, we can consider them based on corollas, this time corollas which have either three or four half-edges 
of gauge boson type (indicated by wavy lines), or one gauge-boson half-edge with two half-edges of ghost type (indicated by consistently oriented straight dashed lines), or
one gauge-boson half-edge with two half-edges of fermion type (indicated by consistently oriented straight full lines).

We repeat our notational conventions.  We mark in such graphs edges and vertices in various ways and we let 
\[
\mathcal{G}^{r,l}_{n\fourvertex,m\ghostloop}
\]  
be the set of all graphs with external half-edges specifying the amplitude  $r$, with $l$ loops and $n$  4-gluon vertices, and  $m$ ghostloops.
Similarly, we will indicate the number of marked edges and other qualifiers as needed.

Still, if we want to leave a qualifier $l,n,m,\cdots$ unspecified (so that we consider the union of all sets with any number of such items), we replace it by $/$.
For sums or series of graphs we continue to  use
\[
X^{r,l}_{n\fourvertex,m\ghostloop}
\]  
which are sums (for fixed $l$) or series (for $l=/$) of graphs weighted by symmetry, colour and other such factors as defined below. 

We now start adopting graph homology to our purposes in gauge theory.
\subsection{Marking edges}
Recall that the Feynman rule for the 4-valent vertex is
\begin{equation*}\begin{split}
\Phi\Big(
\parbox[c][25pt]{25pt}{\centering\scriptsize
\begin{fmfgraph*}(15,15)
	\fmfleft{2,1}
	\fmfright{3,4}
	\fmfv{label=$1$,label.dist=2}{1}
	\fmfv{label=$2$,label.dist=2}{2}
	\fmfv{label=$3$,label.dist=2}{3}
	\fmfv{label=$4$,label.dist=2}{4}
	\fmf{boson}{1,v}
	\fmf{boson}{2,v}
	\fmf{boson}{3,v}
	\fmf{boson}{4,v}
\end{fmfgraph*}}
\Big) = &  + f^{a_1a_2b}f^{a_3a_4b}(g^{\mu_1\mu_3}g^{\mu_2\mu_4}-g^{\mu_1\mu_4}g^{\mu_2\mu_3}) \\
 & + f^{a_1a_3b}f^{a_2a_4b}(g^{\mu_1\mu_2}g^{\mu_3\mu_4}-g^{\mu_1\mu_4}g^{\mu_3\mu_2}) \\
 & + f^{a_1a_4b}f^{a_2a_3b}(g^{\mu_1\mu_2}g^{\mu_4\mu_3}-g^{\mu_1\mu_3}g^{\mu_4\mu_2}).
\end{split}\end{equation*}
We introduce a new edge type \markededge\ which has the following Feynman rule:
\begin{equation}\begin{split}
\Phi\Big(
\parbox[c][25pt]{25pt}{\centering\scriptsize
\begin{fmfgraph*}(15,15)
	\fmfleft{2,1}
	\fmfright{3,4}
	\fmfv{label=$1$,label.dist=2}{1}
	\fmfv{label=$2$,label.dist=2}{2}
	\fmfv{label=$3$,label.dist=2}{3}
	\fmfv{label=$4$,label.dist=2}{4}
	\fmf{boson,tension=2}{1,v1}
	\fmf{boson,tension=2}{2,v1}
	\fmf{boson,tension=2}{3,v2}
	\fmf{boson,tension=2}{4,v2}
	\fmf{marked,label=$e$,label.dist=4}{v1,v2}
\end{fmfgraph*}}
\Big)
&= f^{a_1a_2b}f^{a_3a_4b}(g^{\mu_1\mu_3}g^{\mu_2\mu_4}-g^{\mu_1\mu_4}g^{\mu_2\mu_3}) \\
&=: \colour\Big(
\parbox{15pt}{\centering\scriptsize
\begin{fmfgraph*}(15,15)
	\fmfleft{2,1}
	\fmfright{3,4}
	\fmf{boson,tension=2}{1,v1}
	\fmf{boson,tension=2}{2,v1}
	\fmf{boson,tension=2}{3,v2}
	\fmf{boson,tension=2}{4,v2}
	\fmf{boson,label.dist=3}{v1,v2}
\end{fmfgraph*}}
\Big) W_e,
\end{split}\label{eq:W}\end{equation}
so that  we can write the 4-point vertex as
\begin{equation}
\parbox{20pt}{\centering\scriptsize
\begin{fmfgraph*}(20,20)
	\fmfleft{2,1}
	\fmfright{3,4}
	\fmf{boson,tension=2}{1,v}
	\fmf{boson,tension=2}{2,v}
	\fmf{boson,tension=2}{3,v}
	\fmf{boson,tension=2}{4,v}
\end{fmfgraph*}}
\sim
\parbox{20pt}{\centering\scriptsize
\begin{fmfgraph*}(20,20)
	\fmfleft{2,1}
	\fmfright{3,4}
	\fmf{boson,tension=2}{1,v1}
	\fmf{boson,tension=2}{2,v1}
	\fmf{boson,tension=2}{3,v2}
	\fmf{boson,tension=2}{4,v2}
	\fmf{marked}{v1,v2}
\end{fmfgraph*}}
+
\parbox{20pt}{\centering\scriptsize
\begin{fmfgraph*}(20,20)
	\fmfleft{2,1}
	\fmfright{3,4}
	\fmf{boson,tension=2}{1,v1}
	\fmf{phantom,tension=2}{2,v1}
	\fmf{phantom,tension=2}{3,v2}
	\fmf{boson,tension=2}{4,v2}
	\fmf{marked}{v1,v2}
	\fmffreeze
	\fmf{boson,tension=2}{2,v2}
	\fmf{boson,tension=2,rubout=3}{3,v1}
\end{fmfgraph*}}
+
\parbox{20pt}{\centering\scriptsize
\begin{fmfgraph*}(20,20)
	\fmfleft{2,1}
	\fmfright{3,4}
	\fmf{boson,tension=2}{1,v1}
	\fmf{boson,tension=2}{2,v2}
	\fmf{boson,tension=2}{3,v2}
	\fmf{boson,tension=2}{4,v1}
	\fmf{marked}{v1,v2}
\end{fmfgraph*}}
.\label{eq:|}
\end{equation}
(The relation $\sim$ denotes that the left- and right-hand side have the same Feynman amplitude.) Note that because of this relation, the internal marked edge does not correspond to a propagator. It is just a graphical way of writing the three terms of the 4-valent vertex. 
\begin{rem}
The fact that the 4-valent vertex decomposes
in such a way into a product of two corollas is actually the starting point for recursion relations of amplitudes \cite{recursion,WalterCore}. 
\end{rem}

For any graph $\Gamma$ with marked edges, let $\bar{\Gamma}$ be the graph where the marked edges shrink to zero length.\footnote{If we have $k$ marked edges, here are $3^k$ different graphs $\Gamma$ which have the same $\bar{\Gamma}$.}
The Feynman--Schwinger integrand of a graph with marked edges is given for us by
\begin{equation}\begin{split}
I_\Gamma &= \Big(\prod_{e\in\Gamma^{[1]}_{\rm marked}}W_e\Big)\left( \prod_{v\in V^\Gamma, v\cap \Gamma^{[1]}_{\rm marked}=\emptyset} V_v\right) I_{\bar{\Gamma}}.\label{markedintegrand}
\end{split}\end{equation}
Here, $V_v$ is the colour-stripped part of the Feynman rule of a 3-gluon vertex, 
\begin{equation}
\label{3-vertex}
V_v=\sum_{\mathrm{cycl}(1,2,3)}(\xi_1-\xi_2)_{\mu_3}g_{\mu_1\mu_2}.
\end{equation}

Note that the scalar integrand $ I_{\bar{\Gamma}}$  does obviously not contain the edge variables of the $\markededge$-marked edges. 

For future use, we define
\begin{equation} \mathcal I_\Gamma = \Big(\prod_{e\in\Gamma^{[1]}_{\rm marked}}W_e\Big)\left( \prod_{v\in V^\Gamma, v\cap \Gamma^{[1]}_{\rm marked}=\emptyset} V_v\right) e^{-\sum\limits_{\substack{e\in\Gamma\\\text{unmarked}}}A_e\xi_e'^2}, \label{eq:mathcalI} \end{equation}
which is such that
\begin{equation*} \int\frac{\mathrm d\underline k_L}{(2\pi)^{dl}} \mathcal I_\Gamma = I_\Gamma. \end{equation*}

\begin{defn}\label{4-1}
Define a derivation $\chi_+: H \to H$ on generators by
\begin{equation*} \chi_+\Gamma = \sum_{e\in\Gamma^{[1]}_{\rm int}} \chi_+^{e}\Gamma, \end{equation*}
where
\begin{equation*} \chi_+^{e}\Gamma = \begin{cases}
0 & \text{if $e$ shares a vertex with a marked, fermion or ghost edge,} \\
\Gamma_{e\rightsquigarrow\markededge} & \text{otherwise}.
\end{cases}\end{equation*}
\end{defn}

The next lemma shows how symmetry factors relate upon exchanging 4-valent vertices for a pair of corollas with a marked edge in-between.
We consider  graphs with $l$ loops, $k$ 4-gluon vertices and $k'$ marked edges, for an amplitude $r$. Also, $\bar{\mathcal G}$ denotes unlabelled graphs, in contrast to labelled graphs in $\mathcal{G}$. 
\begin{lem}\label{lem:graph-4v-|}
For any graph $\Gamma\in\mathcal G^{r,l}_{k\fourvertex,k'\markededge}$ we have
\begin{equation} \frac1{\Sym(\Gamma)}\Gamma \sim \frac1k \sum_{\substack{\Gamma'\in\bar{\mathcal G}^{r,l}_{k-1\fourvertex,k'+1\markededge}\\ \exists e\in\Gamma^{\prime[1]}_{\rm marked}:\Gamma'/e=\Gamma }} \frac1{\Sym(\Gamma')}\Gamma'. \end{equation}
\end{lem}
\begin{proof}
\newcommand\blobgraph{{\includegraphics{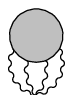}}}
\newcommand\blobgrapha{{\includegraphics{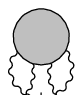}}}
\newcommand\blobgraphb{{\includegraphics{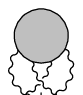}}}
\newcommand\blobgraphc{{\includegraphics{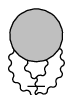}}}
\newcommand\blobgraphsmall{{\includegraphics{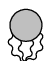}}}
\newcommand\blobgraphasmall{{\includegraphics{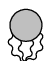}}}
\newcommand\blobgraphbsmall{{\includegraphics{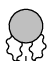}}}
\newcommand\blobgraphcsmall{{\includegraphics{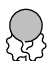}}}

Let $v\in\Gamma^{[0]}_4$ with adjacent edges $1,2,3,4$:
\begin{equation*} \Gamma = \parbox{30pt}{\centering
\begin{fmfgraph*}(20,40)
	\fmfbottom b
	\fmftop{t}
	\fmfpolyn{filled=25,smooth,tension=.5} p9
	\fmf{phantom,tension=99}{t,p7}
	\fmf{boson,tension=.25, left=.5 ,label=\scriptsize$1$,label.dist=1}{b,p1}
	\fmf{boson,tension=.25, left=.25,label=\scriptsize$2$,label.dist=1,label.pos=right}{b,p2}
	\fmf{boson,tension=.25,right=.25,label=\scriptsize$3$,label.dist=1,label.pos=right}{b,p3}
	\fmf{boson,tension=.25,right=.5 ,label=\scriptsize$4$,label.dist=1}{b,p4}
	\fmfv{label=\scriptsize$v$,label.dist=2}{b}
\end{fmfgraph*}}. \end{equation*}
(We do not show $\Gamma$'s external edges in the diagram.) Apply (\ref{eq:|}):
\begin{equation*}
\frac1{\Sym(\blobgraphsmall)} \blobgraph
=\frac1{\Sym(\blobgraphsmall)} \Big(\blobgrapha + \blobgraphb + \blobgraphc\Big) .
\end{equation*}
The following three cases can occur:
\begin{enumerate}
\item[1.] The four edges adjacent to $v$ are each un-interchangeable. In this case
\begin{equation*}
\Sym\Big(\blobgrapha\Big)
=\Sym\Big(\blobgraphb\Big)
=\Sym\Big(\blobgraphc\Big)
=\Sym\Big(\blobgraph\Big),
\end{equation*}
so that
\begin{equation*}\begin{split}
\frac1{\Sym(\blobgraphsmall)}\blobgraph
&\sim \frac1{\Sym(\blobgraphasmall)}\blobgrapha
+\frac1{\Sym(\blobgraphbsmall)}\blobgraphb 
+ \frac1{\Sym(\blobgraphcsmall)}\blobgraphc .
\end{split}\end{equation*}
Note that the three graphs at the right-hand side are all non-isomorphic.
\item[2.] Two of $v$'s adjacent edges are interchangeable, say $1$ and $2$. Then
\begin{equation*} \blobgraphb = \blobgraphc. \end{equation*}
The symmetry factors of the new graphs are
\begin{equation*} \Sym\Big(\blobgrapha\Big) = \Sym\Big(\blobgraph\Big), \end{equation*}
and
\begin{equation*} \Sym\Big(\blobgraphc\Big) = \tfrac12 \Sym\Big(\blobgraph\Big), \end{equation*}
so that
\begin{equation*}\begin{split}
\frac1{\Sym(\blobgraphsmall)} \blobgraph
&\sim \frac1{\Sym(\blobgraphsmall)} \Big(\blobgrapha+2\blobgraphc\Big) \\
&= \frac1{\Sym(\blobgraphasmall)} \blobgrapha + \frac1{\Sym(\blobgraphcsmall)} \blobgraphc .
\end{split}\end{equation*}
Note that the two graphs at the right-hand side are unequal.

\item[3.] Three of $v$'s adjacent edges are interchangeable, say $1$, $2$ and $3$. Then
\begin{equation*} \blobgrapha = \blobgraphb = \blobgraphc, \end{equation*}
and
\begin{equation*} \Sym\Big(\blobgrapha\Big) = \tfrac13 \Sym\Big(\blobgraph\Big). \end{equation*}
So:
\begin{equation*}
\frac1{\Sym(\blobgraphsmall)} \blobgraph
\sim\frac3{\Sym(\blobgraphsmall)} \blobgrapha
=\frac1{\Sym(\blobgraphasmall)}\blobgrapha.
\end{equation*}
\end{enumerate}
Thus we can conclude that
\begin{equation*} \frac1{\Sym(\Gamma)}\Gamma = \sum_{\substack{\Gamma'\in\bar{\mathcal G}^{n,l}_{k-1\fourvertex,k'+1\markededge}\\ \exists e\in\Gamma^{\prime[1]}_{\rm marked}:\Gamma'/e=\Gamma\\ \text{where the new vertex is $v$} }} \frac1{\Sym(\Gamma')}\Gamma'. \end{equation*}
The result follows up summing this over all 4-valent vertices in $\Gamma$, giving rise to the factor $\#\Gamma^{[0]}_4=k$:
\begin{equation*}
\Gamma \sim \frac1k \sum_{\substack{\Gamma'\in\bar{\mathcal G}^{n,l}_{k-1\fourvertex,k'+1\markededge}\\ \exists e\in\Gamma^{\prime[1]}_{\rm marked}:\Gamma'/e=\Gamma }} \frac1{\Sym(\Gamma')}\Phi(\Gamma') .
\end{equation*}
\end{proof}

\begin{exmp}
Take $\Gamma =
\parbox{25pt}{\begin{fmfgraph*}(25,25)
	\fmfsurroundn e3
	\fmf{boson,tension=2}{e1,v1}
	\fmf{boson,tension=2}{e2,v2}
	\fmf{boson,tension=2}{e3,v2}
	\fmf{boson,left,tension=.5}{v1,v2,v1}
\end{fmfgraph*}}
$ and apply equation (\ref{eq:|}) to the 4-valent vertex:
\begin{equation} \tfrac12
\parbox{25pt}{\begin{fmfgraph*}(25,25)
	\fmfsurroundn e3
	\fmf{boson,tension=2}{e1,v1}
	\fmf{boson,tension=2}{e2,v2}
	\fmf{boson,tension=2}{e3,v2}
	\fmf{boson,left,tension=.5}{v1,v2,v1}
\end{fmfgraph*}}
\sim \tfrac12 \Big(
\parbox{25pt}{\begin{fmfgraph*}(25,25)
	\fmfsurroundn e3
	\fmf{boson,tension=2}{e1,v1}
	\fmf{boson,tension=2}{e2,v2}
	\fmf{boson,tension=2}{e3,v3}
	\fmf{boson}{v3,v1,v2}
	\fmf{marked}{v2,v3}
\end{fmfgraph*}}
+ \parbox{25pt}{\begin{fmfgraph*}(25,25)
	\fmfsurroundn e3
	\fmf{boson,tension=2}{e1,v1}
	\fmf{boson,tension=2}{e2,v2}
	\fmf{boson,tension=2}{e3,v3}
	\fmf{phantom}{v3,v1,v2}
	\fmf{marked}{v2,v3}
	\fmffreeze
	\fmfi{boson}{vloc(__v1){dir -90} .. vloc(__v2){dir 120}}
	\fmfi{boson,rubout=3}{vloc(__v1){dir 90} .. vloc(__v3){dir -120}}
\end{fmfgraph*}}
+ \parbox{25pt}{\begin{fmfgraph*}(25,25)
	\fmfsurroundn e3
	\fmf{boson,tension=2}{e1,v1}
	\fmf{boson,tension=2}{e2,v3}
	\fmf{boson,tension=2}{e3,v3}
	\fmf{boson,left,tension=.5}{v1,v2,v1}
	\fmf{marked}{v2,v3}
\end{fmfgraph*}}
\Big) 
=
\parbox{25pt}{\begin{fmfgraph*}(25,25)
	\fmfsurroundn e3
	\fmf{boson,tension=2}{e1,v1}
	\fmf{boson,tension=2}{e2,v2}
	\fmf{boson,tension=2}{e3,v3}
	\fmf{boson}{v3,v1,v2}
	\fmf{marked}{v2,v3}
\end{fmfgraph*}}
+ \tfrac12 \parbox{25pt}{\begin{fmfgraph*}(25,25)
	\fmfsurroundn e3
	\fmf{boson,tension=2}{e1,v1}
	\fmf{boson,tension=2}{e2,v3}
	\fmf{boson,tension=2}{e3,v3}
	\fmf{boson,left,tension=.5}{v1,v2,v1}
	\fmf{marked}{v2,v3}
\end{fmfgraph*}}
. \label{eq:exmp:graph-4v-|-n=3}\end{equation}
\end{exmp}

\begin{exmp}
Take $\Gamma =
\parbox{25pt}{\begin{fmfgraph*}(25,12.5)
	\fmfsurroundn e2
	\fmf{boson,tension=2}{e1,v1}
	\fmf{boson,tension=2}{e2,v2}
	\fmf{boson}{v1,v2}
	\fmffreeze
	\fmf{boson,left}{v1,v2,v1}
\end{fmfgraph*}}
$ and apply (\ref{eq:|}) to one of the two vertices:
\begin{equation}
\tfrac16
\parbox{25pt}{\begin{fmfgraph*}(25,12.5)
	\fmfsurroundn e2
	\fmf{boson,tension=2}{e1,v1}
	\fmf{boson,tension=2}{e2,v2}
	\fmf{boson}{v1,v2}
	\fmffreeze
	\fmf{boson,left}{v1,v2,v1}
\end{fmfgraph*}}
\sim \tfrac16 \Big(
\parbox{25pt}{\begin{fmfgraph*}(25,12.5)
	\fmfsurroundn e2
	\fmf{boson,tension=2}{e1,v1}
	\fmf{boson,tension=2}{e2,v2}
	\fmf{boson,left,tension=.5}{v1,v2}
	\fmf{phantom,right,tension=.5,tag=1}{v1,v2}
	\fmffreeze
	\fmfi{marked}{subpath (1,2) of vpath1(__v1,__v2)}
	\fmfi{boson}{subpath (0,1) of vpath1(__v1,__v2)}
	\fmfi{boson}{point 1 of vpath1(__v1,__v2){down} .. point 0 of vpath1(__v1,__v2){right}}
\end{fmfgraph*}}
+ \parbox{25pt}{\begin{fmfgraph*}(25,12.5)
	\fmfsurroundn e2
	\fmf{boson,tension=2}{e1,v1}
	\fmf{boson,tension=2}{e2,v2}
	\fmf{boson,right,tension=.5}{v1,v2}
	\fmf{phantom,left,tension=.5,tag=1}{v1,v2}
	\fmffreeze
	\fmfi{marked}{subpath (1,2) of vpath1(__v1,__v2)}
	\fmfi{boson}{subpath (0,1) of vpath1(__v1,__v2)}
	\fmfi{boson}{point 1 of vpath1(__v1,__v2) .. point 0 of vpath1(__v1,__v2){right}}
\end{fmfgraph*}}
+\parbox{25pt}{\begin{fmfgraph*}(25,12.5)
	\fmfsurroundn e2
	\fmf{boson,tension=2}{e1,v1}
	\fmf{boson,tension=2}{e2,v2}
	\fmf{phantom,right,tag=1}{v1,v2}
	\fmffreeze
	\fmfi{marked}{subpath (1,2) of vpath1(__v1,__v2)}
	\fmfi{boson}{subpath (0,1) of vpath1(__v1,__v2)}
	\fmfi{boson}{vloc(__v1){dir 180} .. vloc(__v2){dir 180}}
	\fmfi{boson,rubout=3}{point 1 of vpath1(__v1,__v2){dir 240} .. point 0 of vpath1(__v1,__v2){dir 90}}
\end{fmfgraph*}}
\Big)
= \tfrac12
\parbox{25pt}{\begin{fmfgraph*}(25,12.5)
	\fmfsurroundn e2
	\fmf{boson,tension=2}{e1,v1}
	\fmf{boson,tension=2}{e2,v2}
	\fmf{boson,left,tension=.5}{v1,v2}
	\fmf{phantom,right,tension=.5,tag=1}{v1,v2}
	\fmffreeze
	\fmfi{marked}{subpath (1,2) of vpath1(__v1,__v2)}
	\fmfi{boson}{subpath (0,1) of vpath1(__v1,__v2)}
	\fmfi{boson}{point 1 of vpath1(__v1,__v2) .. point 0 of vpath1(__v1,__v2){right}}
\end{fmfgraph*}}
. \label{eq:exmp:graph-4v-|-l=2-1}\end{equation}
Analogously, we get
\begin{align}
\tfrac16
\parbox{25pt}{\begin{fmfgraph*}(25,12.5)
	\fmfsurroundn e2
	\fmf{boson,tension=2}{e1,v1}
	\fmf{boson,tension=2}{e2,v2}
	\fmf{boson}{v1,v2}
	\fmffreeze
	\fmf{boson,left}{v1,v2,v1}
\end{fmfgraph*}}
&\sim \tfrac12
\parbox{25pt}{\begin{fmfgraph*}(25,12.5)
	\fmfsurround{e2,e1}
	\fmf{boson,tension=2}{e1,v1}
	\fmf{boson,tension=2}{e2,v2}
	\fmf{boson,right,tension=.5}{v1,v2}
	\fmf{phantom,left,tension=.5,tag=1}{v1,v2}
	\fmffreeze
	\fmfi{marked}{subpath (1,2) of vpath1(__v1,__v2)}
	\fmfi{boson}{subpath (0,1) of vpath1(__v1,__v2)}
	\fmfi{boson}{point 1 of vpath1(__v1,__v2) .. point 0 of vpath1(__v1,__v2){left}}
\end{fmfgraph*}},\nonumber
\intertext{so that we can write}
\tfrac16
\parbox{25pt}{\begin{fmfgraph*}(25,12.5)
	\fmfsurroundn e2
	\fmf{boson,tension=2}{e1,v1}
	\fmf{boson,tension=2}{e2,v2}
	\fmf{boson}{v1,v2}
	\fmffreeze
	\fmf{boson,left}{v1,v2,v1}
\end{fmfgraph*}}
&\sim \tfrac12 \Big( \tfrac12
\parbox{25pt}{\begin{fmfgraph*}(25,12.5)
	\fmfsurroundn e2
	\fmf{boson,tension=2}{e1,v1}
	\fmf{boson,tension=2}{e2,v2}
	\fmf{boson,left,tension=.5}{v1,v2}
	\fmf{phantom,right,tension=.5,tag=1}{v1,v2}
	\fmffreeze
	\fmfi{marked}{subpath (1,2) of vpath1(__v1,__v2)}
	\fmfi{boson}{subpath (0,1) of vpath1(__v1,__v2)}
	\fmfi{boson}{point 1 of vpath1(__v1,__v2) .. point 0 of vpath1(__v1,__v2){right}}
\end{fmfgraph*}}
+ \tfrac12
\parbox{25pt}{\begin{fmfgraph*}(25,12.5)
	\fmfsurround{e2,e1}
	\fmf{boson,tension=2}{e1,v1}
	\fmf{boson,tension=2}{e2,v2}
	\fmf{boson,right,tension=.5}{v1,v2}
	\fmf{phantom,left,tension=.5,tag=1}{v1,v2}
	\fmffreeze
	\fmfi{marked}{subpath (1,2) of vpath1(__v1,__v2)}
	\fmfi{boson}{subpath (0,1) of vpath1(__v1,__v2)}
	\fmfi{boson}{point 1 of vpath1(__v1,__v2) .. point 0 of vpath1(__v1,__v2){left}}
\end{fmfgraph*}}
\Big). \label{eq:exmp:graph-4v-|-l=2-2}
\end{align}
\end{exmp}

\begin{exmp}
\begin{equation*}\begin{split}
\tfrac12
\parbox{25pt}{\begin{fmfgraph*}(25,12.5)
	\fmfsurroundn e2
	\fmf{boson,tension=2}{e1,v1}
	\fmf{boson,tension=2}{e2,v2}
	\fmf{boson,left,tension=.5}{v1,v2}
	\fmf{phantom,right,tension=.5,tag=1}{v1,v2}
	\fmffreeze
	\fmfi{marked}{subpath (1,2) of vpath1(__v1,__v2)}
	\fmfi{boson}{subpath (0,1) of vpath1(__v1,__v2)}
	\fmfi{boson}{point 1 of vpath1(__v1,__v2) .. point 0 of vpath1(__v1,__v2){right}}
\end{fmfgraph*}}
&\sim \tfrac12 \Big(
\parbox{25pt}{\begin{fmfgraph*}(25,12.5)
	\fmfsurroundn e2
	\fmf{boson,tension=2}{e1,v1}
	\fmf{boson,tension=2}{e2,v2}
	\fmf{boson,left,tension=.5}{v1,v2}
	\fmf{phantom,right,tension=.5,tag=1}{v1,v2}
	\fmffreeze
	\fmfi{marked}{subpath (0,2/3) of vpath1(__v1,__v2)}
	\fmfi{marked}{subpath (4/3,2) of vpath1(__v1,__v2)}
	\fmfi{boson}{point 2/3 of vpath1(__v1,__v2){up} .. point 4/3 of vpath1(__v1,__v2){down}}
	\fmfi{boson}{point 2/3 of vpath1(__v1,__v2){down} .. point 4/3 of vpath1(__v1,__v2){up}}
\end{fmfgraph*}}
+
\parbox{25pt}{
\begin{fmfgraph*}(25,12.5)
	\fmfleft{l}
	\fmfright{r}
	\fmf{phantom}{l,v1,v2,r}
	\fmf{boson}{l,v1}
	\fmf{boson}{v2,r}
	\fmf{phantom,left,tension=0,tag=1}{v1,v2}
	\fmf{phantom,right,tension=0,tag=2}{v1,v2}
	\fmffreeze
	\fmfi{marked}{subpath (0,1) of vpath1(__v1,__v2)}
	\fmfi{boson}{subpath (0,1) of vpath2(__v1,__v2)}
	\fmfi{boson}{point 1 of vpath1(__v1,__v2) .. point 1 of vpath2(__v1,__v2)}
	\fmfi{boson}{subpath (1,2) of vpath1(__v1,__v2)}
	\fmfi{marked}{subpath (1,2) of vpath2(__v1,__v2)}
\end{fmfgraph*}}
+
\parbox{25pt}{
\begin{fmfgraph*}(25,12.5)
	\fmfleft{l}
	\fmfright{r}
	\fmf{phantom}{l,v1,v2,r}
	\fmf{boson}{l,v1}
	\fmf{boson}{v2,r}
	\fmf{phantom,left,tension=0,tag=1}{v1,v2}
	\fmf{phantom,right,tension=0,tag=2}{v1,v2}
	\fmffreeze
	\fmfi{marked}{subpath (0,1) of vpath1(__v1,__v2)}
	\fmfi{boson}{subpath (0,1) of vpath2(__v1,__v2)}
	\fmfi{marked}{subpath (1,2) of vpath2(__v1,__v2)}
	\fmfi{boson}{point 1 of vpath1(__v1,__v2){dir -120} .. vloc(__v2){dir 0}}
	\fmfi{boson,rubout=3}{point 1 of vpath1(__v1,__v2){dir 0}
		.. point 1 of vpath2(__v1,__v2){dir -60}}
\end{fmfgraph*}}
\Big)
= \tfrac12
\parbox{25pt}{\begin{fmfgraph*}(25,12.5)
	\fmfsurroundn e2
	\fmf{boson,tension=2}{e1,v1}
	\fmf{boson,tension=2}{e2,v2}
	\fmf{boson,left,tension=.5}{v1,v2}
	\fmf{phantom,right,tension=.5,tag=1}{v1,v2}
	\fmffreeze
	\fmfi{marked}{subpath (0,2/3) of vpath1(__v1,__v2)}
	\fmfi{marked}{subpath (4/3,2) of vpath1(__v1,__v2)}
	\fmfi{boson}{point 2/3 of vpath1(__v1,__v2){up} .. point 4/3 of vpath1(__v1,__v2){down}}
	\fmfi{boson}{point 2/3 of vpath1(__v1,__v2){down} .. point 4/3 of vpath1(__v1,__v2){up}}
\end{fmfgraph*}}
+
\parbox{25pt}{
\begin{fmfgraph*}(25,12.5)
	\fmfleft{l}
	\fmfright{r}
	\fmf{phantom}{l,v1,v2,r}
	\fmf{boson}{l,v1}
	\fmf{boson}{v2,r}
	\fmf{phantom,left,tension=0,tag=1}{v1,v2}
	\fmf{phantom,right,tension=0,tag=2}{v1,v2}
	\fmffreeze
	\fmfi{marked}{subpath (0,1) of vpath1(__v1,__v2)}
	\fmfi{boson}{subpath (0,1) of vpath2(__v1,__v2)}
	\fmfi{boson}{point 1 of vpath1(__v1,__v2) .. point 1 of vpath2(__v1,__v2)}
	\fmfi{boson}{subpath (1,2) of vpath1(__v1,__v2)}
	\fmfi{marked}{subpath (1,2) of vpath2(__v1,__v2)}
\end{fmfgraph*}}
.
\end{split}\end{equation*}
\end{exmp}

The next lemma is crucial, as it shows that the fundamental relation between a 4-gluon vertex and a pair of 3-gluon vertices, in all three channels,
gives a relation between combinatorial Green functions. We would have no chance at getting a well-defined gauge theory without such a relation.
\begin{lem}\label{lem:||&4-vertices}\begin{enumerate}
\item For any $k$ and $k'$, $k'<k$:
\begin{equation*}
\frac1{(\fracx k{k'})} X^{n.l}_{k-k'\fourvertex,k'\markededge} \sim \frac1{(\fracx k{k'+1})} X^{n,l}_{k-k'-1\fourvertex,k'+1\markededge}.
\end{equation*}
\item For any $k$:
\begin{equation*}
X^{n,l}_{k\fourvertex} \sim X^{n,l}_{k\markededge}.
\end{equation*}
\end{enumerate}\end{lem}

\begin{proof}
For i. we have from lemma \ref{lem:graph-4v-|} for the Green's function $X^{n,l}_{k-k'\fourvertex,k'\markededge}$:
\begin{equation*}\begin{split}
 X^{n,l}_{k-k'\fourvertex,k'\markededge}
&= \sum_{\Gamma\in\bar{\mathcal G}^{n,l}_{k-k'\fourvertex,k'\markededge}} \frac1{\Sym(\Gamma)} \Gamma \\
&\sim \frac1{k-k'} \sum_{\Gamma\in\bar{\mathcal G}^{n,l}_{k-k'\fourvertex\,k'\markededge}} \sum_{\substack{\Gamma'\in\bar{\mathcal G}^{n,l}_{k-k'-1\fourvertex,k'+1\markededge}\\ \exists e\in\Gamma^{\prime[1]}_{\rm marked}:\Gamma'/e=\Gamma }} \frac1{\Sym(\Gamma')} \Gamma' \\
&= \frac{k'+1}{k-k'} \sum_{\Gamma'\in\bar{\mathcal G}^{n,l}_{k-k'-1\fourvertex,k'+1\markededge}} \frac1{\Sym(\Gamma')} \Gamma' \\
&= \frac{k'+1}{k-k'} X^{n,l}_{k-k'-1\fourvertex,k'+1\markededge} .
\end{split}\end{equation*}
The factor $k'+1$ appears because every graph $\Gamma'\in\bar{\mathcal G}^{n,l}_{k-k'-1\fourvertex,k+1'\markededge}$ can be obtained from $\#\Gamma'^{[1]}_{\rm marked}=k'+1$ graphs $\Gamma\in\bar{\mathcal G}^{n,l}_{k-k'\fourvertex,k'\markededge}$ by applying (\ref{eq:|}). 
Using the identity
\begin{equation*} \Big(\fracx k{k'+1}\Big)
= \frac{k-k'}{k'+1}\Big(\fracx k{k'}\Big), \end{equation*}
it follows that
\begin{equation*} \frac1{(\fracx k{k'})} X^{n,l}_{k-k'\fourvertex,k'\markededge} \sim \frac1{(\fracx k{k'+1})} X^{n,l}_{k-k'-1\fourvertex,k'+1\markededge} . \end{equation*}

For ii. we have
\begin{equation*} X^{n,l}_{k\fourvertex} \sim \frac1{(\fracx k{k'})} X^{n,l}_{k-k'\fourvertex,k'\markededge}. \end{equation*}
This is true by induction: it is an equality for $k'=0$ and the inductive step is true by i.. Taking $k'=k$ gives:
\begin{equation*} X^{n,l}_{k\fourvertex} \sim X^{n,l}_{k\markededge}. \qedhere \end{equation*}
\end{proof}

In the following, if $r$ is an $n$-gluon amplitude, we often replace the subscript $r$ by $n$, as in this example:
\begin{exmp}
Take $n=3$, $l=1$ and $k=1$.
\begin{equation*}\begin{split}
X^{3,1}_{1\fourvertex} &= \tfrac12
\parbox{25pt}{\begin{fmfgraph*}(25,25)
	\fmfsurroundn e3
	\fmf{boson,tension=2}{e1,v1}
	\fmf{boson,tension=2}{e2,v2}
	\fmf{boson,tension=2}{e3,v2}
	\fmf{boson,left,tension=.5}{v1,v2,v1}
\end{fmfgraph*}}
+ \tfrac12
\parbox{25pt}{\begin{fmfgraph*}(25,25)
	\fmfsurround{e3,e1,e2}
	\fmf{boson,tension=2}{e1,v1}
	\fmf{boson,tension=2}{e2,v2}
	\fmf{boson,tension=2}{e3,v2}
	\fmf{boson,left,tension=.5}{v1,v2,v1}
\end{fmfgraph*}}
+ \tfrac12
\parbox{25pt}{\begin{fmfgraph*}(25,25)
	\fmfsurround{e2,e3,e1}
	\fmf{boson,tension=2}{e1,v1}
	\fmf{boson,tension=2}{e2,v2}
	\fmf{boson,tension=2}{e3,v2}
	\fmf{boson,left,tension=.5}{v1,v2,v1}
\end{fmfgraph*}}
\\
&\sim 
\parbox{25pt}{\begin{fmfgraph*}(25,25)
	\fmfsurroundn e3
	\fmf{boson,tension=2}{e1,v1}
	\fmf{boson,tension=2}{e2,v2}
	\fmf{boson,tension=2}{e3,v3}
	\fmf{boson}{v3,v1,v2}
	\fmf{marked}{v2,v3}
\end{fmfgraph*}}
+ \parbox{25pt}{\begin{fmfgraph*}(25,25)
	\fmfsurround{e3,e1,e2}
	\fmf{boson,tension=2}{e1,v1}
	\fmf{boson,tension=2}{e2,v2}
	\fmf{boson,tension=2}{e3,v3}
	\fmf{boson}{v3,v1,v2}
	\fmf{marked}{v2,v3}
\end{fmfgraph*}}
+ \parbox{25pt}{\begin{fmfgraph*}(25,25)
	\fmfsurround{e2,e3,e1}
	\fmf{boson,tension=2}{e1,v1}
	\fmf{boson,tension=2}{e2,v2}
	\fmf{boson,tension=2}{e3,v3}
	\fmf{boson}{v3,v1,v2}
	\fmf{marked}{v2,v3}
\end{fmfgraph*}}
+ \tfrac12 \parbox{25pt}{\begin{fmfgraph*}(25,25)
	\fmfsurroundn e3
	\fmf{boson,tension=2}{e1,v1}
	\fmf{boson,tension=2}{e2,v3}
	\fmf{boson,tension=2}{e3,v3}
	\fmf{boson,left,tension=.5}{v1,v2,v1}
	\fmf{marked}{v2,v3}
\end{fmfgraph*}}
+ \tfrac12 \parbox{25pt}{\begin{fmfgraph*}(25,25)
	\fmfsurround{e3,e1,e2}
	\fmf{boson,tension=2}{e1,v1}
	\fmf{boson,tension=2}{e2,v3}
	\fmf{boson,tension=2}{e3,v3}
	\fmf{boson,left,tension=.5}{v1,v2,v1}
	\fmf{marked}{v2,v3}
\end{fmfgraph*}}
+ \tfrac12 \parbox{25pt}{\begin{fmfgraph*}(25,25)
	\fmfsurround{e2,e3,e1}
	\fmf{boson,tension=2}{e1,v1}
	\fmf{boson,tension=2}{e2,v3}
	\fmf{boson,tension=2}{e3,v3}
	\fmf{boson,left,tension=.5}{v1,v2,v1}
	\fmf{marked}{v2,v3}
\end{fmfgraph*}}
= X^{3,1}_{1\markededge}
\end{split}\end{equation*}
where we have used Eq.(\ref{eq:exmp:graph-4v-|-n=3}).
\end{exmp}

\begin{exmp}
Take $n=2$, $l=2$ and $k=2$. Note that $\bar{\mathcal G}^{2,2}_{2\fourvertex}$ contains just one graph and use (\ref{eq:exmp:graph-4v-|-l=2-1}) and (\ref{eq:exmp:graph-4v-|-l=2-2}):
\begin{equation*}
X^{2,2}_{2\fourvertex}
= \tfrac16
\parbox{25pt}{\begin{fmfgraph*}(25,12.5)
	\fmfsurroundn e2
	\fmf{boson,tension=2}{e1,v1}
	\fmf{boson,tension=2}{e2,v2}
	\fmf{boson}{v1,v2}
	\fmffreeze
	\fmf{boson,left}{v1,v2,v1}
\end{fmfgraph*}}
\sim \tfrac12
\parbox{25pt}{\begin{fmfgraph*}(25,12.5)
	\fmfsurroundn e2
	\fmf{boson,tension=2}{e1,v1}
	\fmf{boson,tension=2}{e2,v2}
	\fmf{boson,left,tension=.5}{v1,v2}
	\fmf{phantom,right,tension=.5,tag=1}{v1,v2}
	\fmffreeze
	\fmfi{marked}{subpath (1,2) of vpath1(__v1,__v2)}
	\fmfi{boson}{subpath (0,1) of vpath1(__v1,__v2)}
	\fmfi{boson}{point 1 of vpath1(__v1,__v2) .. point 0 of vpath1(__v1,__v2){right}}
\end{fmfgraph*}}
\sim \tfrac12
\parbox{25pt}{\begin{fmfgraph*}(25,12.5)
	\fmfsurroundn e2
	\fmf{boson,tension=2}{e1,v1}
	\fmf{boson,tension=2}{e2,v2}
	\fmf{boson,left,tension=.5}{v1,v2}
	\fmf{phantom,right,tension=.5,tag=1}{v1,v2}
	\fmffreeze
	\fmfi{marked}{subpath (0,2/3) of vpath1(__v1,__v2)}
	\fmfi{marked}{subpath (4/3,2) of vpath1(__v1,__v2)}
	\fmfi{boson}{point 2/3 of vpath1(__v1,__v2){up} .. point 4/3 of vpath1(__v1,__v2){down}}
	\fmfi{boson}{point 2/3 of vpath1(__v1,__v2){down} .. point 4/3 of vpath1(__v1,__v2){up}}
\end{fmfgraph*}}
+
\parbox{25pt}{
\begin{fmfgraph*}(25,12.5)
	\fmfleft{l}
	\fmfright{r}
	\fmf{phantom}{l,v1,v2,r}
	\fmf{boson}{l,v1}
	\fmf{boson}{v2,r}
	\fmf{phantom,left,tension=0,tag=1}{v1,v2}
	\fmf{phantom,right,tension=0,tag=2}{v1,v2}
	\fmffreeze
	\fmfi{marked}{subpath (0,1) of vpath1(__v1,__v2)}
	\fmfi{boson}{subpath (0,1) of vpath2(__v1,__v2)}
	\fmfi{boson}{point 1 of vpath1(__v1,__v2) .. point 1 of vpath2(__v1,__v2)}
	\fmfi{boson}{subpath (1,2) of vpath1(__v1,__v2)}
	\fmfi{marked}{subpath (1,2) of vpath2(__v1,__v2)}
\end{fmfgraph*}}
= X^{2,2}_{2\markededge} .
\end{equation*}
\end{exmp}
\begin{lem}\label{lem:skeletongraphs}
\begin{enumerate}
\item For any graph $\Gamma\in\bar{\mathcal G}^{n,l}$:
\begin{samepage}\begin{equation*} \frac1{\Sym(\Gamma)} \frac{(\chi_+)^k}{k!} \Gamma
= \sum_{\substack{\Gamma'\in\bar{\mathcal G}^{n,l}_{k\markededge}\\ \skeleton(\Gamma')=\Gamma}} \frac1{\Sym(\Gamma')} \Gamma'. 
\end{equation*}
where the skeleton of a marked graph is the graph with its markings removed.
\end{samepage}
\item For any $k \geq 0$, $\frac{(\chi_+)^k}{k!} X^{n,l} = X^{n,l}_{k\markededge}.$
\item We have $e^{\chi_+} X^{n,l} = X^{n,l}_{/\markededge}= X^{n,l}_{/\fourvertex}.$
\end{enumerate}
\end{lem}
\begin{proof}
For i. we have
\begin{equation*}\begin{split}
 \frac1{\Sym(\Gamma)} \frac{(\chi_+)^k}{k!} [\Gamma]
&= \frac1{\Sym(\Gamma)} \sum_{\{e_1,\ldots,e_k\}\subset\Gamma^{[1]}_{\rm int}} [\chi_+^{e_1}\cdots\chi_+^{k_k}\Gamma] \\
&= \frac1{\Sym(\Gamma)} \sum_{\substack{\Gamma'\in\bar{\mathcal G}^{n,l}_{k\markededge}\\ \skeleton(\Gamma')=[\Gamma]}} \#\big\{\{e_1,\ldots,e_k\}\subset\Gamma^{[1]}_{\rm int} \ \big|\ [\chi_+^{e_1}\cdots\chi_+^{k_k}\Gamma]=\Gamma'\big\}\ \Gamma'\\
& = \sum_{\substack{\Gamma'\in\bar{\mathcal G}^{n,l}_{k\markededge}\\ \skeleton(\Gamma')=[\Gamma]}} \frac1{\Sym(\Gamma')}\Gamma'.
\end{split}\end{equation*}
For ii. we apply i. to the combinatorial Green's function $X^{n,l}$, instead of a single graph. Summing 
over all graphs in $\bar{\mathcal G}^{n,l}$ yields
\begin{equation*}\begin{split}
\frac{(\chi_+)^k}{k!} X^{n,l}
&= \sum_{\Gamma\in\bar{\mathcal G}^{n,l}} \frac1{\Sym(\Gamma)} \frac{(\chi_+)^k}{k!} \Gamma \\
&= \sum_{\Gamma'\in\bar{\mathcal G}^{n,l}_{k\markededge}} \frac1{\Sym(\Gamma')} \Gamma'
= X^{n,l}_{k\markededge}.
\end{split}\end{equation*}
Finally, iii. follows directly from ii.\ by taking the sum over $k$.

\end{proof}

\begin{exmp}
Take $
\parbox{25pt}{\begin{fmfgraph*}(25,12.5)
	\fmfleft{l}
	\fmfright{r}
	\fmf{phantom}{l,v1,v2,r}
	\fmf{boson}{l,v1}
	\fmf{boson}{v2,r}
	\fmf{phantom,left,tension=0,tag=1}{v1,v2}
	\fmf{phantom,right,tension=0,tag=2}{v1,v2}
	\fmffreeze
	\fmfi{boson}{subpath (0,1) of vpath1(__v1,__v2)}
	\fmfi{boson}{subpath (0,1) of vpath2(__v1,__v2)}
	\fmfi{boson}{point 1 of vpath1(__v1,__v2) .. point 1 of vpath2(__v1,__v2)}
	\fmfi{boson}{subpath (1,2) of vpath1(__v1,__v2)}
	\fmfi{boson}{subpath (1,2) of vpath2(__v1,__v2)}
\end{fmfgraph*}}$
and $k=2$; then lemma \ref{lem:skeletongraphs}.i reads
\begin{equation*}
\tfrac12 \frac{(\chi_+)^2}2
\parbox{25pt}{\begin{fmfgraph*}(25,12.5)
	\fmfleft{l}
	\fmfright{r}
	\fmf{phantom}{l,v1,v2,r}
	\fmf{boson}{l,v1}
	\fmf{boson}{v2,r}
	\fmf{phantom,left,tension=0,tag=1}{v1,v2}
	\fmf{phantom,right,tension=0,tag=2}{v1,v2}
	\fmffreeze
	\fmfi{boson}{subpath (0,1) of vpath1(__v1,__v2)}
	\fmfi{boson}{subpath (0,1) of vpath2(__v1,__v2)}
	\fmfi{boson}{point 1 of vpath1(__v1,__v2) .. point 1 of vpath2(__v1,__v2)}
	\fmfi{boson}{subpath (1,2) of vpath1(__v1,__v2)}
	\fmfi{boson}{subpath (1,2) of vpath2(__v1,__v2)}
\end{fmfgraph*}}
= \tfrac12 \Big( \parbox{25pt}{
\begin{fmfgraph*}(25,12.5)
	\fmfleft{l}
	\fmfright{r}
	\fmf{phantom}{l,v1,v2,r}
	\fmf{boson}{l,v1}
	\fmf{boson}{v2,r}
	\fmf{phantom,left,tension=0,tag=1}{v1,v2}
	\fmf{phantom,right,tension=0,tag=2}{v1,v2}
	\fmffreeze
	\fmfi{marked}{subpath (0,1) of vpath1(__v1,__v2)}
	\fmfi{boson}{subpath (0,1) of vpath2(__v1,__v2)}
	\fmfi{boson}{point 1 of vpath1(__v1,__v2) .. point 1 of vpath2(__v1,__v2)}
	\fmfi{boson}{subpath (1,2) of vpath1(__v1,__v2)}
	\fmfi{marked}{subpath (1,2) of vpath2(__v1,__v2)}
\end{fmfgraph*}}
+ \parbox{25pt}{
\begin{fmfgraph*}(25,12.5)
	\fmfleft{l}
	\fmfright{r}
	\fmf{phantom}{l,v1,v2,r}
	\fmf{boson}{l,v1}
	\fmf{boson}{v2,r}
	\fmf{phantom,left,tension=0,tag=1}{v1,v2}
	\fmf{phantom,right,tension=0,tag=2}{v1,v2}
	\fmffreeze
	\fmfi{boson}{subpath (0,1) of vpath1(__v1,__v2)}
	\fmfi{marked}{subpath (0,1) of vpath2(__v1,__v2)}
	\fmfi{boson}{point 1 of vpath1(__v1,__v2) .. point 1 of vpath2(__v1,__v2)}
	\fmfi{marked}{subpath (1,2) of vpath1(__v1,__v2)}
	\fmfi{boson}{subpath (1,2) of vpath2(__v1,__v2)}
\end{fmfgraph*}} \Big)
=
\parbox{25pt}{\begin{fmfgraph*}(25,12.5)
	\fmfleft{l}
	\fmfright{r}
	\fmf{phantom}{l,v1,v2,r}
	\fmf{boson}{l,v1}
	\fmf{boson}{v2,r}
	\fmf{phantom,left,tension=0,tag=1}{v1,v2}
	\fmf{phantom,right,tension=0,tag=2}{v1,v2}
	\fmffreeze
	\fmfi{marked}{subpath (0,1) of vpath1(__v1,__v2)}
	\fmfi{boson}{subpath (0,1) of vpath2(__v1,__v2)}
	\fmfi{boson}{point 1 of vpath1(__v1,__v2) .. point 1 of vpath2(__v1,__v2)}
	\fmfi{boson}{subpath (1,2) of vpath1(__v1,__v2)}
	\fmfi{marked}{subpath (1,2) of vpath2(__v1,__v2)}
\end{fmfgraph*}}
.
\end{equation*}
\end{exmp}

\begin{rem}
Note that $\chi_+$ has a non-trivial kernel as it can create self-loops graphs, for example:
\begin{equation*}
\tfrac12 \chi_+
\parbox{25pt}{\begin{fmfgraph*}(25,12.5)
	\fmfleft 1
	\fmfright 2
	\fmf{boson,tension=2}{1,a}
	\fmf{boson,tension=2}{2,b}
	\fmf{boson,tension=.5,right}{a,b,a}
\end{fmfgraph*}}
=
\parbox{25pt}{\begin{fmfgraph*}(25,12.5)
	\fmfleft 1
	\fmfright 2
	\fmf{boson,tension=2}{1,a}
	\fmf{boson,tension=2}{2,b}
	\fmf{boson,tension=.5,right}{b,a}
	\fmf{marked,tension=.5,right}{a,b}
\end{fmfgraph*}}
\sim \tfrac12
\parbox{25pt}{\begin{fmfgraph*}(25,12.5)
	\fmfleft 1
	\fmfright 2
	\fmf{boson}{1,a,2}
	\fmf{boson}{a,a}
\end{fmfgraph*}}
\end{equation*}
Here we have used that
\begin{equation*}
\parbox{25pt}{\begin{fmfgraph*}(25,12.5)
	\fmfleft 1
	\fmfright 2
	\fmf{boson}{1,a,2}
	\fmf{boson}{a,a}
\end{fmfgraph*}}
\sim
\parbox{25pt}{\begin{fmfgraph*}(25,12.5)
	\fmfleft 1
	\fmfright 2
	\fmf{boson,tension=2}{1,a}
	\fmf{boson,tension=2}{2,b}
	\fmf{boson,tension=.5,right}{b,a}
	\fmf{marked,tension=.5,right}{a,b}
\end{fmfgraph*}}
+
\parbox{25pt}{\begin{fmfgraph*}(25,12.5)
	\fmfleft 1
	\fmfright 2
	\fmftop t
	\fmf{boson,tension=2}{1,a}
	\fmf{boson,tension=2}{2,b}
	\fmf{marked,right}{a,b}
	\fmffreeze
	\fmfi{boson}{vloc(__a){dir -30} .. vloc(__t){dir 180}}
	\fmfi{boson,rubout=5}{vloc(__t){dir 180} .. vloc(__b){dir 30}}
\end{fmfgraph*}}
+
\parbox{25pt}{\begin{fmfgraph*}(25,25)
	\fmfleft 1
	\fmfright 2
	\fmftop t
	\fmf{boson}{1,a,2}
	\fmffreeze
	\fmf{marked}{a,b}
	\fmf{boson,tension=.5,right}{b,t,b}
\end{fmfgraph*}}
\sim 2
\parbox{25pt}{\begin{fmfgraph*}(25,12.5)
	\fmfleft 1
	\fmfright 2
	\fmf{boson,tension=2}{1,a}
	\fmf{boson,tension=2}{2,b}
	\fmf{boson,tension=.5,right}{b,a}
	\fmf{marked,tension=.5,right}{a,b}
\end{fmfgraph*}}
.
\end{equation*}
This does not influence our results, since self-loops have amplitude zero.
\end{rem}
It is now time to study graph homology, again by studying marked edges, but now the labelling plays a crucial role.
\subsection{The graph differential $s$ for gauge theory graphs} 
\begin{defn}\label{4-12}
The derivation $s: H \to H$ is given on generators by
\begin{equation*} s\Gamma = \sum_{e\in\Gamma^{[1]}_{\rm int}} (-)^{\#\{e'\in\Gamma^{[1]}_{\rm marked} \ |\ e'<e\}} s_e\Gamma  \end{equation*}
where
\begin{equation*} s_e\Gamma = \begin{cases}
0 & \text{if $e$ shares a vertex with a marked or ghost edge} \\
\Gamma_{e\rightsquigarrow\mmarkededge} & \text{otherwise}
\end{cases}\end{equation*}
In the above, $<$ is a (strict) total ordering on $\Gamma^{[1]}$.

Next, we want to distinguish the markings created by $\chi_+$ and $s$. Therefore we draw the latter with two lines instead of one. So two lines indicate the action of $s$, and we denote:
$\Gamma^{[1]}_\markededge\cup\Gamma^{[1]}_\mmarkededge=\Gamma^{[1]}_{\rm marked} \subset \Gamma^{[1]}_{\rm int}$.
\end{defn}

\begin{prop}\label{ssquare}
$s$ is a differential: $s^2\Gamma = 0$.
\end{prop}
\begin{proof}
We compute
\begin{equation*}\begin{split}
s^2\Gamma
&= \sum_{e_1,e_2\in\Gamma^{[1]}_{\rm int}}
(-)^{\#\{e_1'\in\Gamma^{[1]}_{\rm marked} \ |\ e_1'<e_1\}+\#\{e_2'\in\Gamma^{[1]}_{\rm marked}\cup\{e_1\} \ |\ e_2'<e_2\}} s_{e_1}s_{e_2} \Gamma \\
&= \sum_{\substack{e_1,e_2\in\Gamma^{[1]}_{\rm int}\\ e_1<e_2 }}
(-)^{\#\{e\in\Gamma^{[1]}_{\rm marked}\ |\ e_1<e<e_2\}+1} s_{e_1}s_{e_2} \Gamma \\
&\qquad +\sum_{\substack{e_1,e_2\in\Gamma^{[1]}_{\rm int}\\ e_1>e_2 }}
(-)^{\#\{e\in\Gamma^{[1]}_{\rm marked}\ |\ e_2<e<e_1\}} s_{e_1}s_{e_2} \Gamma \\
&=0.
\end{split}\end{equation*}
\end{proof}

\begin{exmp} We work with labelled graphs, eg.  
\begin{equation*}
\Gamma=
\parbox[c][35pt]{55pt}{\centering\scriptsize
\begin{fmfgraph*}(40,20)
	\fmfleft{1}
	\fmfright{2}
	\fmfv{label=$1$,label.dist=2}{1}
	\fmfv{label=$2$,label.dist=2}{2}
	\fmf{phantom}{1,a,b,2}
	\fmf{boson}{1,a}
	\fmf{boson}{b,2}
	\fmf{phantom,left,tension=0,tag=1}{a,b}
	\fmf{phantom,right,tension=0,tag=2}{a,b}
	\fmffreeze
	\fmfi{boson,label=$3$,label.dist=2}{subpath (0,1) of vpath1(__a,__b)}
	\fmfi{boson,label=$4$,label.dist=2}{subpath (0,1) of vpath2(__a,__b)}
	\fmfi{boson,label=$5$,label.dist=2}{point 1 of vpath1(__a,__b) .. point 1 of vpath2(__a,__b)}
	\fmfi{boson,label=$6$,label.dist=2}{subpath (1,2) of vpath1(__a,__b)}
	\fmfi{boson,label=$7$,label.dist=2}{subpath (1,2) of vpath2(__a,__b)}
\end{fmfgraph*}},
\end{equation*}
for which
\begin{equation*}\begin{split}
s^2
\parbox{20pt}{
\begin{fmfgraph*}(20,10)
	\fmfleft{1}
	\fmfright{2}
	\fmf{phantom}{1,a,b,2}
	\fmf{boson}{1,a}
	\fmf{boson}{b,2}
	\fmf{phantom,left,tension=0,tag=1}{a,b}
	\fmf{phantom,right,tension=0,tag=2}{a,b}
	\fmffreeze
	\fmfi{boson}{subpath (0,1) of vpath1(__a,__b)}
	\fmfi{boson}{subpath (0,1) of vpath2(__a,__b)}
	\fmfi{boson}{point 1 of vpath1(__a,__b) .. point 1 of vpath2(__a,__b)}
	\fmfi{boson}{subpath (1,2) of vpath1(__a,__b)}
	\fmfi{boson}{subpath (1,2) of vpath2(__a,__b)}
\end{fmfgraph*}}
&= s
\parbox{20pt}{
\begin{fmfgraph*}(20,10)
	\fmfleft{1}
	\fmfright{2}
	\fmf{phantom}{1,a,b,2}
	\fmf{boson}{1,a}
	\fmf{boson}{b,2}
	\fmf{phantom,left,tension=0,tag=1}{a,b}
	\fmf{phantom,right,tension=0,tag=2}{a,b}
	\fmffreeze
	\fmfi{mmarked}{subpath (0,1) of vpath1(__a,__b)}
	\fmfi{boson}{subpath (0,1) of vpath2(__a,__b)}
	\fmfi{boson}{point 1 of vpath1(__a,__b) .. point 1 of vpath2(__a,__b)}
	\fmfi{boson}{subpath (1,2) of vpath1(__a,__b)}
	\fmfi{boson}{subpath (1,2) of vpath2(__a,__b)}
\end{fmfgraph*}}
+ s
\parbox{20pt}{
\begin{fmfgraph*}(20,10)
	\fmfleft{1}
	\fmfright{2}
	\fmf{phantom}{1,a,b,2}
	\fmf{boson}{1,a}
	\fmf{boson}{b,2}
	\fmf{phantom,left,tension=0,tag=1}{a,b}
	\fmf{phantom,right,tension=0,tag=2}{a,b}
	\fmffreeze
	\fmfi{boson}{subpath (0,1) of vpath1(__a,__b)}
	\fmfi{mmarked}{subpath (0,1) of vpath2(__a,__b)}
	\fmfi{boson}{point 1 of vpath1(__a,__b) .. point 1 of vpath2(__a,__b)}
	\fmfi{boson}{subpath (1,2) of vpath1(__a,__b)}
	\fmfi{boson}{subpath (1,2) of vpath2(__a,__b)}
\end{fmfgraph*}}
+ s
\parbox{20pt}{
\begin{fmfgraph*}(20,10)
	\fmfleft{1}
	\fmfright{2}
	\fmf{phantom}{1,a,b,2}
	\fmf{boson}{1,a}
	\fmf{boson}{b,2}
	\fmf{phantom,left,tension=0,tag=1}{a,b}
	\fmf{phantom,right,tension=0,tag=2}{a,b}
	\fmffreeze
	\fmfi{boson}{subpath (0,1) of vpath1(__a,__b)}
	\fmfi{boson}{subpath (0,1) of vpath2(__a,__b)}
	\fmfi{mmarked}{point 1 of vpath1(__a,__b) .. point 1 of vpath2(__a,__b)}
	\fmfi{boson}{subpath (1,2) of vpath1(__a,__b)}
	\fmfi{boson}{subpath (1,2) of vpath2(__a,__b)}
\end{fmfgraph*}}
+ s
\parbox{20pt}{
\begin{fmfgraph*}(20,10)
	\fmfleft{1}
	\fmfright{2}
	\fmf{phantom}{1,a,b,2}
	\fmf{boson}{1,a}
	\fmf{boson}{b,2}
	\fmf{phantom,left,tension=0,tag=1}{a,b}
	\fmf{phantom,right,tension=0,tag=2}{a,b}
	\fmffreeze
	\fmfi{boson}{subpath (0,1) of vpath1(__a,__b)}
	\fmfi{boson}{subpath (0,1) of vpath2(__a,__b)}
	\fmfi{boson}{point 1 of vpath1(__a,__b) .. point 1 of vpath2(__a,__b)}
	\fmfi{mmarked}{subpath (1,2) of vpath1(__a,__b)}
	\fmfi{boson}{subpath (1,2) of vpath2(__a,__b)}
\end{fmfgraph*}}
+ s
\parbox{20pt}{
\begin{fmfgraph*}(20,10)
	\fmfleft{1}
	\fmfright{2}
	\fmf{phantom}{1,a,b,2}
	\fmf{boson}{1,a}
	\fmf{boson}{b,2}
	\fmf{phantom,left,tension=0,tag=1}{a,b}
	\fmf{phantom,right,tension=0,tag=2}{a,b}
	\fmffreeze
	\fmfi{boson}{subpath (0,1) of vpath1(__a,__b)}
	\fmfi{boson}{subpath (0,1) of vpath2(__a,__b)}
	\fmfi{boson}{point 1 of vpath1(__a,__b) .. point 1 of vpath2(__a,__b)}
	\fmfi{boson}{subpath (1,2) of vpath1(__a,__b)}
	\fmfi{mmarked}{subpath (1,2) of vpath2(__a,__b)}
\end{fmfgraph*}}
\\
&= -
\parbox{20pt}{
\begin{fmfgraph*}(20,10)
	\fmfleft{1}
	\fmfright{2}
	\fmf{phantom}{1,a,b,2}
	\fmf{boson}{1,a}
	\fmf{boson}{b,2}
	\fmf{phantom,left,tension=0,tag=1}{a,b}
	\fmf{phantom,right,tension=0,tag=2}{a,b}
	\fmffreeze
	\fmfi{mmarked}{subpath (0,1) of vpath1(__a,__b)}
	\fmfi{boson}{subpath (0,1) of vpath2(__a,__b)}
	\fmfi{boson}{point 1 of vpath1(__a,__b) .. point 1 of vpath2(__a,__b)}
	\fmfi{boson}{subpath (1,2) of vpath1(__a,__b)}
	\fmfi{mmarked}{subpath (1,2) of vpath2(__a,__b)}
\end{fmfgraph*}}
-
\parbox{20pt}{
\begin{fmfgraph*}(20,10)
	\fmfleft{1}
	\fmfright{2}
	\fmf{phantom}{1,a,b,2}
	\fmf{boson}{1,a}
	\fmf{boson}{b,2}
	\fmf{phantom,left,tension=0,tag=1}{a,b}
	\fmf{phantom,right,tension=0,tag=2}{a,b}
	\fmffreeze
	\fmfi{boson}{subpath (0,1) of vpath1(__a,__b)}
	\fmfi{mmarked}{subpath (0,1) of vpath2(__a,__b)}
	\fmfi{boson}{point 1 of vpath1(__a,__b) .. point 1 of vpath2(__a,__b)}
	\fmfi{mmarked}{subpath (1,2) of vpath1(__a,__b)}
	\fmfi{boson}{subpath (1,2) of vpath2(__a,__b)}
\end{fmfgraph*}}
+
\parbox{20pt}{
\begin{fmfgraph*}(20,10)
	\fmfleft{1}
	\fmfright{2}
	\fmf{phantom}{1,a,b,2}
	\fmf{boson}{1,a}
	\fmf{boson}{b,2}
	\fmf{phantom,left,tension=0,tag=1}{a,b}
	\fmf{phantom,right,tension=0,tag=2}{a,b}
	\fmffreeze
	\fmfi{boson}{subpath (0,1) of vpath1(__a,__b)}
	\fmfi{mmarked}{subpath (0,1) of vpath2(__a,__b)}
	\fmfi{boson}{point 1 of vpath1(__a,__b) .. point 1 of vpath2(__a,__b)}
	\fmfi{mmarked}{subpath (1,2) of vpath1(__a,__b)}
	\fmfi{boson}{subpath (1,2) of vpath2(__a,__b)}
\end{fmfgraph*}}
+
\parbox{20pt}{
\begin{fmfgraph*}(20,10)
	\fmfleft{1}
	\fmfright{2}
	\fmf{phantom}{1,a,b,2}
	\fmf{boson}{1,a}
	\fmf{boson}{b,2}
	\fmf{phantom,left,tension=0,tag=1}{a,b}
	\fmf{phantom,right,tension=0,tag=2}{a,b}
	\fmffreeze
	\fmfi{mmarked}{subpath (0,1) of vpath1(__a,__b)}
	\fmfi{boson}{subpath (0,1) of vpath2(__a,__b)}
	\fmfi{boson}{point 1 of vpath1(__a,__b) .. point 1 of vpath2(__a,__b)}
	\fmfi{boson}{subpath (1,2) of vpath1(__a,__b)}
	\fmfi{mmarked}{subpath (1,2) of vpath2(__a,__b)}
\end{fmfgraph*}}
=0.
\end{split}\end{equation*}
\end{exmp}
 
\subsection{The differential $S$ for gauge theories}
Marking edges, which corresponds upon summation of connected diagrams to shrinking  pairs of two 3-gluon vertices to 4-gluon vertices, should match with the graphs with 4-gluon vertices present in the theory. This can be phrased homologically. 
\begin{defn}\label{4-15}
A derivation $S: H \to H$ is given by $S = s + \sigma$ where
\begin{equation*} \sigma\Gamma = (-)^{\#\Gamma^{[1]}_{\rm marked}} \sum_{e\in\Gamma^{[1]}\markededge} (-)^{\#\{e'\in\Gamma^{[1]}\markededge \ |\ e'>e\}} \sigma_e\Gamma, \end{equation*}
and
\begin{equation*} \sigma_e\Gamma = \Gamma_{e\rightsquigarrow\mmarkededge}. \end{equation*}
\end{defn}

\begin{prop}
\label{Ssquare} 
$S$ is a differential: $S^2\Gamma = 0$.
\end{prop}
\begin{proof}
A calculation shows that
\begin{align*}
s\sigma\Gamma
&= (-)^{\Gamma^{[1]}_{\rm marked}} \sum_{e_1\in\Gamma^{[1]}_\markededge}\sum_{e_2\in\Gamma^{[1]}_{\rm int}} (-)^{\#\{e_1'\in\Gamma^{[1]}\markededge \ |\ e_1'>e_1\}+\#\{e_2'\in\Gamma^{[1]}_{\rm marked} \ |\ e_2'<e_2\}} s_{e_2}\sigma_{e_1}\Gamma \\
\sigma s\Gamma
&= (-)^{\Gamma^{[1]}_{\rm marked}+1} \sum_{e_2\in\Gamma^{[1]}_{\rm int}} \sum_{e_1\in\Gamma^{[1]}\markededge} (-)^{\#\{e_2'\in\Gamma_{\rm marked} \ |\ e_2'<e_2\} + \#\{e_1'\in\Gamma^{[1]}\markededge \ |\ e_1'>e_1\}} s_{e_2}\sigma_{e_1}\Gamma \\
&= - s\sigma\Gamma,
\intertext{and also}
\sigma^2\Gamma
&= \sum_{e_1\in\Gamma^{[1]}\markededge} \sum_{e_1\in\Gamma^{[1]}\markededge\smallsetminus\{e_1\}} (-)^{\#\{e_1'\in\Gamma^{[1]}\markededge \ |\ e_1'>e_1\}+\#\{e_2'\in\Gamma^{[1]}\markededge\smallsetminus\{e_1\} \ |\ e_2'>e_2\}} \sigma_{e_1}\sigma_{e_2}\Gamma \\
&= \sum_{\substack{e_1,e_2\in\Gamma^{[1]}\markededge\\ e_1<e_2}} (-)^{\#\{e\in\Gamma^{[1]}\markededge \ |\ e_1<e<e_2\}+1} \sigma_{e_1}\sigma_{e_2}\Gamma \\
&\qquad+ \sum_{\substack{e_1,e_2\in\Gamma^{[1]}\markededge\\ e_1>e_2}} (-)^{\#\{e\in\Gamma^{[1]}\markededge \ |\ e_1<e<e_2\}} \sigma_{e_1}\sigma_{e_2}\Gamma \\
&= 0.
\end{align*}
so that
\begin{equation*} S^2\Gamma
= s^2\Gamma + s\sigma\Gamma + \sigma s\Gamma + \sigma^2\Gamma
= 0. \end{equation*}
\end{proof}
\begin{rem}
Note that upon summing the markings in a 3-valent corolla, and identifying such a sum with a 4-valent vertex, the operators $s,S$  here reduce to the operators 
$\tilde s$ and $\tilde S$  we had before in Eq.(\ref{S-homology}).
\end{rem}
\begin{exmp}
\begin{align*}
s\sigma
\parbox{20pt}{\begin{fmfgraph*}(20,10)
	\fmfleft{1}
	\fmfright{2}
	\fmf{phantom}{1,a,b,2}
	\fmf{boson}{1,a}
	\fmf{boson}{b,2}
	\fmf{phantom,left,tension=0,tag=1}{a,b}
	\fmf{phantom,right,tension=0,tag=2}{a,b}
	\fmffreeze
	\fmfi{marked}{subpath (0,1) of vpath1(__a,__b)}
	\fmfi{boson}{subpath (0,1) of vpath2(__a,__b)}
	\fmfi{boson}{point 1 of vpath1(__a,__b) .. point 1 of vpath2(__a,__b)}
	\fmfi{boson}{subpath (1,2) of vpath1(__a,__b)}
	\fmfi{boson}{subpath (1,2) of vpath2(__a,__b)}
\end{fmfgraph*}}
&= -s
\parbox{20pt}{\begin{fmfgraph*}(20,10)
	\fmfleft{1}
	\fmfright{2}
	\fmf{phantom}{1,a,b,2}
	\fmf{boson}{1,a}
	\fmf{boson}{b,2}
	\fmf{phantom,left,tension=0,tag=1}{a,b}
	\fmf{phantom,right,tension=0,tag=2}{a,b}
	\fmffreeze
	\fmfi{mmarked}{subpath (0,1) of vpath1(__a,__b)}
	\fmfi{boson}{subpath (0,1) of vpath2(__a,__b)}
	\fmfi{boson}{point 1 of vpath1(__a,__b) .. point 1 of vpath2(__a,__b)}
	\fmfi{boson}{subpath (1,2) of vpath1(__a,__b)}
	\fmfi{boson}{subpath (1,2) of vpath2(__a,__b)}
\end{fmfgraph*}}
=
\parbox{20pt}{\begin{fmfgraph*}(20,10)
	\fmfleft{1}
	\fmfright{2}
	\fmf{phantom}{1,a,b,2}
	\fmf{boson}{1,a}
	\fmf{boson}{b,2}
	\fmf{phantom,left,tension=0,tag=1}{a,b}
	\fmf{phantom,right,tension=0,tag=2}{a,b}
	\fmffreeze
	\fmfi{mmarked}{subpath (0,1) of vpath1(__a,__b)}
	\fmfi{boson}{subpath (0,1) of vpath2(__a,__b)}
	\fmfi{boson}{point 1 of vpath1(__a,__b) .. point 1 of vpath2(__a,__b)}
	\fmfi{boson}{subpath (1,2) of vpath1(__a,__b)}
	\fmfi{mmarked}{subpath (1,2) of vpath2(__a,__b)}
\end{fmfgraph*}}
\\
\sigma s
\parbox{20pt}{\begin{fmfgraph*}(20,10)
	\fmfleft{1}
	\fmfright{2}
	\fmf{phantom}{1,a,b,2}
	\fmf{boson}{1,a}
	\fmf{boson}{b,2}
	\fmf{phantom,left,tension=0,tag=1}{a,b}
	\fmf{phantom,right,tension=0,tag=2}{a,b}
	\fmffreeze
	\fmfi{marked}{subpath (0,1) of vpath1(__a,__b)}
	\fmfi{boson}{subpath (0,1) of vpath2(__a,__b)}
	\fmfi{boson}{point 1 of vpath1(__a,__b) .. point 1 of vpath2(__a,__b)}
	\fmfi{boson}{subpath (1,2) of vpath1(__a,__b)}
	\fmfi{boson}{subpath (1,2) of vpath2(__a,__b)}
\end{fmfgraph*}}
&= -\sigma 
\parbox{20pt}{\begin{fmfgraph*}(20,10)
	\fmfleft{1}
	\fmfright{2}
	\fmf{phantom}{1,a,b,2}
	\fmf{boson}{1,a}
	\fmf{boson}{b,2}
	\fmf{phantom,left,tension=0,tag=1}{a,b}
	\fmf{phantom,right,tension=0,tag=2}{a,b}
	\fmffreeze
	\fmfi{marked}{subpath (0,1) of vpath1(__a,__b)}
	\fmfi{boson}{subpath (0,1) of vpath2(__a,__b)}
	\fmfi{boson}{point 1 of vpath1(__a,__b) .. point 1 of vpath2(__a,__b)}
	\fmfi{boson}{subpath (1,2) of vpath1(__a,__b)}
	\fmfi{mmarked}{subpath (1,2) of vpath2(__a,__b)}
\end{fmfgraph*}}
= -
\parbox{20pt}{\begin{fmfgraph*}(20,10)
	\fmfleft{1}
	\fmfright{2}
	\fmf{phantom}{1,a,b,2}
	\fmf{boson}{1,a}
	\fmf{boson}{b,2}
	\fmf{phantom,left,tension=0,tag=1}{a,b}
	\fmf{phantom,right,tension=0,tag=2}{a,b}
	\fmffreeze
	\fmfi{mmarked}{subpath (0,1) of vpath1(__a,__b)}
	\fmfi{boson}{subpath (0,1) of vpath2(__a,__b)}
	\fmfi{boson}{point 1 of vpath1(__a,__b) .. point 1 of vpath2(__a,__b)}
	\fmfi{boson}{subpath (1,2) of vpath1(__a,__b)}
	\fmfi{mmarked}{subpath (1,2) of vpath2(__a,__b)}
\end{fmfgraph*}}
\\
\sigma^2
\parbox{20pt}{\begin{fmfgraph*}(20,10)
	\fmfleft{1}
	\fmfright{2}
	\fmf{phantom}{1,a,b,2}
	\fmf{boson}{1,a}
	\fmf{boson}{b,2}
	\fmf{phantom,left,tension=0,tag=1}{a,b}
	\fmf{phantom,right,tension=0,tag=2}{a,b}
	\fmffreeze
	\fmfi{marked}{subpath (0,1) of vpath1(__a,__b)}
	\fmfi{boson}{subpath (0,1) of vpath2(__a,__b)}
	\fmfi{boson}{point 1 of vpath1(__a,__b) .. point 1 of vpath2(__a,__b)}
	\fmfi{boson}{subpath (1,2) of vpath1(__a,__b)}
	\fmfi{boson}{subpath (1,2) of vpath2(__a,__b)}
\end{fmfgraph*}}
&= -\sigma
\parbox{20pt}{\begin{fmfgraph*}(20,10)
	\fmfleft{1}
	\fmfright{2}
	\fmf{phantom}{1,a,b,2}
	\fmf{boson}{1,a}
	\fmf{boson}{b,2}
	\fmf{phantom,left,tension=0,tag=1}{a,b}
	\fmf{phantom,right,tension=0,tag=2}{a,b}
	\fmffreeze
	\fmfi{mmarked}{subpath (0,1) of vpath1(__a,__b)}
	\fmfi{boson}{subpath (0,1) of vpath2(__a,__b)}
	\fmfi{boson}{point 1 of vpath1(__a,__b) .. point 1 of vpath2(__a,__b)}
	\fmfi{boson}{subpath (1,2) of vpath1(__a,__b)}
	\fmfi{boson}{subpath (1,2) of vpath2(__a,__b)}
\end{fmfgraph*}}
=0
\\
S^2
\parbox{20pt}{\begin{fmfgraph*}(20,10)
	\fmfleft{1}
	\fmfright{2}
	\fmf{phantom}{1,a,b,2}
	\fmf{boson}{1,a}
	\fmf{boson}{b,2}
	\fmf{phantom,left,tension=0,tag=1}{a,b}
	\fmf{phantom,right,tension=0,tag=2}{a,b}
	\fmffreeze
	\fmfi{marked}{subpath (0,1) of vpath1(__a,__b)}
	\fmfi{boson}{subpath (0,1) of vpath2(__a,__b)}
	\fmfi{boson}{point 1 of vpath1(__a,__b) .. point 1 of vpath2(__a,__b)}
	\fmfi{boson}{subpath (1,2) of vpath1(__a,__b)}
	\fmfi{boson}{subpath (1,2) of vpath2(__a,__b)}
\end{fmfgraph*}}
&= 0 \end{align*}
\end{exmp}

\begin{exmp}
\begin{align*} 
s\sigma
\parbox{20pt}{\begin{fmfgraph*}(20,10)
	\fmfleft{1}
	\fmfright{2}
	\fmf{phantom}{1,a,b,2}
	\fmf{boson}{1,a}
	\fmf{boson}{b,2}
	\fmf{phantom,left,tension=0,tag=1}{a,b}
	\fmf{phantom,right,tension=0,tag=2}{a,b}
	\fmffreeze
	\fmfi{marked}{subpath (0,1) of vpath1(__a,__b)}
	\fmfi{boson}{subpath (0,1) of vpath2(__a,__b)}
	\fmfi{boson}{point 1 of vpath1(__a,__b) .. point 1 of vpath2(__a,__b)}
	\fmfi{boson}{subpath (1,2) of vpath1(__a,__b)}
	\fmfi{marked}{subpath (1,2) of vpath2(__a,__b)}
\end{fmfgraph*}}
&= - s
\parbox{20pt}{\begin{fmfgraph*}(20,10)
	\fmfleft{1}
	\fmfright{2}
	\fmf{phantom}{1,a,b,2}
	\fmf{boson}{1,a}
	\fmf{boson}{b,2}
	\fmf{phantom,left,tension=0,tag=1}{a,b}
	\fmf{phantom,right,tension=0,tag=2}{a,b}
	\fmffreeze
	\fmfi{mmarked}{subpath (0,1) of vpath1(__a,__b)}
	\fmfi{boson}{subpath (0,1) of vpath2(__a,__b)}
	\fmfi{boson}{point 1 of vpath1(__a,__b) .. point 1 of vpath2(__a,__b)}
	\fmfi{boson}{subpath (1,2) of vpath1(__a,__b)}
	\fmfi{marked}{subpath (1,2) of vpath2(__a,__b)}
\end{fmfgraph*}}
+ s
\parbox{20pt}{\begin{fmfgraph*}(20,10)
	\fmfleft{1}
	\fmfright{2}
	\fmf{phantom}{1,a,b,2}
	\fmf{boson}{1,a}
	\fmf{boson}{b,2}
	\fmf{phantom,left,tension=0,tag=1}{a,b}
	\fmf{phantom,right,tension=0,tag=2}{a,b}
	\fmffreeze
	\fmfi{marked}{subpath (0,1) of vpath1(__a,__b)}
	\fmfi{boson}{subpath (0,1) of vpath2(__a,__b)}
	\fmfi{boson}{point 1 of vpath1(__a,__b) .. point 1 of vpath2(__a,__b)}
	\fmfi{boson}{subpath (1,2) of vpath1(__a,__b)}
	\fmfi{mmarked}{subpath (1,2) of vpath2(__a,__b)}
\end{fmfgraph*}}
=0 
\\
 \sigma s
\parbox{20pt}{\begin{fmfgraph*}(20,10)
	\fmfleft{1}
	\fmfright{2}
	\fmf{phantom}{1,a,b,2}
	\fmf{boson}{1,a}
	\fmf{boson}{b,2}
	\fmf{phantom,left,tension=0,tag=1}{a,b}
	\fmf{phantom,right,tension=0,tag=2}{a,b}
	\fmffreeze
	\fmfi{marked}{subpath (0,1) of vpath1(__a,__b)}
	\fmfi{boson}{subpath (0,1) of vpath2(__a,__b)}
	\fmfi{boson}{point 1 of vpath1(__a,__b) .. point 1 of vpath2(__a,__b)}
	\fmfi{boson}{subpath (1,2) of vpath1(__a,__b)}
	\fmfi{marked}{subpath (1,2) of vpath2(__a,__b)}
\end{fmfgraph*}}
&= 0 
\\
\sigma^2
\parbox{20pt}{\begin{fmfgraph*}(20,10)
	\fmfleft{1}
	\fmfright{2}
	\fmf{phantom}{1,a,b,2}
	\fmf{boson}{1,a}
	\fmf{boson}{b,2}
	\fmf{phantom,left,tension=0,tag=1}{a,b}
	\fmf{phantom,right,tension=0,tag=2}{a,b}
	\fmffreeze
	\fmfi{marked}{subpath (0,1) of vpath1(__a,__b)}
	\fmfi{boson}{subpath (0,1) of vpath2(__a,__b)}
	\fmfi{boson}{point 1 of vpath1(__a,__b) .. point 1 of vpath2(__a,__b)}
	\fmfi{boson}{subpath (1,2) of vpath1(__a,__b)}
	\fmfi{marked}{subpath (1,2) of vpath2(__a,__b)}
\end{fmfgraph*}}
&= -\sigma
\parbox{20pt}{\begin{fmfgraph*}(20,10)
	\fmfleft{1}
	\fmfright{2}
	\fmf{phantom}{1,a,b,2}
	\fmf{boson}{1,a}
	\fmf{boson}{b,2}
	\fmf{phantom,left,tension=0,tag=1}{a,b}
	\fmf{phantom,right,tension=0,tag=2}{a,b}
	\fmffreeze
	\fmfi{mmarked}{subpath (0,1) of vpath1(__a,__b)}
	\fmfi{boson}{subpath (0,1) of vpath2(__a,__b)}
	\fmfi{boson}{point 1 of vpath1(__a,__b) .. point 1 of vpath2(__a,__b)}
	\fmfi{boson}{subpath (1,2) of vpath1(__a,__b)}
	\fmfi{marked}{subpath (1,2) of vpath2(__a,__b)}
\end{fmfgraph*}}
+ \sigma
\parbox{20pt}{\begin{fmfgraph*}(20,10)
	\fmfleft{1}
	\fmfright{2}
	\fmf{phantom}{1,a,b,2}
	\fmf{boson}{1,a}
	\fmf{boson}{b,2}
	\fmf{phantom,left,tension=0,tag=1}{a,b}
	\fmf{phantom,right,tension=0,tag=2}{a,b}
	\fmffreeze
	\fmfi{marked}{subpath (0,1) of vpath1(__a,__b)}
	\fmfi{boson}{subpath (0,1) of vpath2(__a,__b)}
	\fmfi{boson}{point 1 of vpath1(__a,__b) .. point 1 of vpath2(__a,__b)}
	\fmfi{boson}{subpath (1,2) of vpath1(__a,__b)}
	\fmfi{mmarked}{subpath (1,2) of vpath2(__a,__b)}
\end{fmfgraph*}}
= -
\parbox{20pt}{\begin{fmfgraph*}(20,10)
	\fmfleft{1}
	\fmfright{2}
	\fmf{phantom}{1,a,b,2}
	\fmf{boson}{1,a}
	\fmf{boson}{b,2}
	\fmf{phantom,left,tension=0,tag=1}{a,b}
	\fmf{phantom,right,tension=0,tag=2}{a,b}
	\fmffreeze
	\fmfi{mmarked}{subpath (0,1) of vpath1(__a,__b)}
	\fmfi{boson}{subpath (0,1) of vpath2(__a,__b)}
	\fmfi{boson}{point 1 of vpath1(__a,__b) .. point 1 of vpath2(__a,__b)}
	\fmfi{boson}{subpath (1,2) of vpath1(__a,__b)}
	\fmfi{mmarked}{subpath (1,2) of vpath2(__a,__b)}
\end{fmfgraph*}}
+
\parbox{20pt}{\begin{fmfgraph*}(20,10)
	\fmfleft{1}
	\fmfright{2}
	\fmf{phantom}{1,a,b,2}
	\fmf{boson}{1,a}
	\fmf{boson}{b,2}
	\fmf{phantom,left,tension=0,tag=1}{a,b}
	\fmf{phantom,right,tension=0,tag=2}{a,b}
	\fmffreeze
	\fmfi{mmarked}{subpath (0,1) of vpath1(__a,__b)}
	\fmfi{boson}{subpath (0,1) of vpath2(__a,__b)}
	\fmfi{boson}{point 1 of vpath1(__a,__b) .. point 1 of vpath2(__a,__b)}
	\fmfi{boson}{subpath (1,2) of vpath1(__a,__b)}
	\fmfi{mmarked}{subpath (1,2) of vpath2(__a,__b)}
\end{fmfgraph*}}
= 0
\\
 S^2
\parbox{20pt}{\begin{fmfgraph*}(20,10)
	\fmfleft{1}
	\fmfright{2}
	\fmf{phantom}{1,a,b,2}
	\fmf{boson}{1,a}
	\fmf{boson}{b,2}
	\fmf{phantom,left,tension=0,tag=1}{a,b}
	\fmf{phantom,right,tension=0,tag=2}{a,b}
	\fmffreeze
	\fmfi{mmarked}{subpath (0,1) of vpath1(__a,__b)}
	\fmfi{boson}{subpath (0,1) of vpath2(__a,__b)}
	\fmfi{boson}{point 1 of vpath1(__a,__b) .. point 1 of vpath2(__a,__b)}
	\fmfi{boson}{subpath (1,2) of vpath1(__a,__b)}
	\fmfi{mmarked}{subpath (1,2) of vpath2(__a,__b)}
\end{fmfgraph*}}
&=0 \end{align*}
\end{exmp}
The cancellations between 3-gluon and 4-gluon vertices necessary to obtain a unitary and covariant gauge theory
demand that shrinking internal edges in graphs with $k_3$ 3-gluon vertices and $k_4$ 4-gluon vertices matches with the graphs 
having $(k_3-2)$ 3-gluon vertices, and $(k_4+1)$ 4-gluon vertices. Rephrased in terms of our marked edges and using our sign conventions,
that precisely is captured by 
\begin{prop}\label{Shom}
Let $\Gamma$ be a graph without marked edges. Then:
\begin{equation} Se^{\chi_+}\Gamma = 0. \end{equation}
\end{prop}
\begin{proof}
By definition:
\begin{equation*} 
e^{\chi_+}\Gamma
= \sum_{k\geq0} \sum_{\substack{ e_1,\ldots,e_k\in\Gamma^{[1]}_{\rm int}\\ e_1<\cdots<e_k}} \chi_+^{e_1}\cdots\chi_+^{e_k} \Gamma \end{equation*}
on which
\begin{equation*}\begin{split}
se^{\chi_+}\Gamma
&= \sum_{k\geq0} \sum_{\substack{ e_1,\ldots,e_{k+1}\in\Gamma^{[1]}_{\rm int}\\ e_1<\cdots<e_{k+1}}} \sum_{l=1}^{k+1} (-)^{l-1} \chi_+^{e_1}\cdots s_{e_l}\cdots\chi_+^{e_{k+1}} \Gamma \\
&= \sum_{k\geq1} \sum_{\substack{ e_1,\ldots,e_k\in\Gamma^{[1]}_{\rm int}\\ e_1<\cdots<e_k}} \sum_{l=1}^k (-)^{l-1} \chi_+^{e_1}\cdots s_{e_l}\cdots\chi_+^{e_k} \Gamma \\
\sigma e^{\chi_+}\Gamma
&= \sum_{k\geq1} \sum_{\substack{ e_1,\ldots,e_k\in\Gamma^{[1]}_{\rm int}\\ e_1<\cdots<e_k}} (-)^k \sum_{l=1}^k (-)^{k-l} \chi_+^{e_1}\cdots s_{e_l}\cdots\chi_+^{e_k} \Gamma \\
&= -se^{\chi_+}\Gamma.
\end{split}\end{equation*}
We conclude that
\begin{equation*} Se^{\chi_+}\Gamma = (s+\sigma)e^{\chi_+}\Gamma = 0. \end{equation*}
\end{proof}

\begin{exmp}
\begin{align*} e^{\chi_+}
\parbox{20pt}{\begin{fmfgraph*}(20,10)
	\fmfleft{1}
	\fmfright{2}
	\fmf{phantom}{1,a,b,2}
	\fmf{boson}{1,a}
	\fmf{boson}{b,2}
	\fmf{phantom,left,tension=0,tag=1}{a,b}
	\fmf{phantom,right,tension=0,tag=2}{a,b}
	\fmffreeze
	\fmfi{boson}{subpath (0,1) of vpath1(__a,__b)}
	\fmfi{boson}{subpath (0,1) of vpath2(__a,__b)}
	\fmfi{boson}{point 1 of vpath1(__a,__b) .. point 1 of vpath2(__a,__b)}
	\fmfi{boson}{subpath (1,2) of vpath1(__a,__b)}
	\fmfi{boson}{subpath (1,2) of vpath2(__a,__b)}
\end{fmfgraph*}}
&=
\parbox{20pt}{\begin{fmfgraph*}(20,10)
	\fmfleft{1}
	\fmfright{2}
	\fmf{phantom}{1,a,b,2}
	\fmf{boson}{1,a}
	\fmf{boson}{b,2}
	\fmf{phantom,left,tension=0,tag=1}{a,b}
	\fmf{phantom,right,tension=0,tag=2}{a,b}
	\fmffreeze
	\fmfi{boson}{subpath (0,1) of vpath1(__a,__b)}
	\fmfi{boson}{subpath (0,1) of vpath2(__a,__b)}
	\fmfi{boson}{point 1 of vpath1(__a,__b) .. point 1 of vpath2(__a,__b)}
	\fmfi{boson}{subpath (1,2) of vpath1(__a,__b)}
	\fmfi{boson}{subpath (1,2) of vpath2(__a,__b)}
\end{fmfgraph*}}
+
\parbox{20pt}{\begin{fmfgraph*}(20,10)
	\fmfleft{1}
	\fmfright{2}
	\fmf{phantom}{1,a,b,2}
	\fmf{boson}{1,a}
	\fmf{boson}{b,2}
	\fmf{phantom,left,tension=0,tag=1}{a,b}
	\fmf{phantom,right,tension=0,tag=2}{a,b}
	\fmffreeze
	\fmfi{marked}{subpath (0,1) of vpath1(__a,__b)}
	\fmfi{boson}{subpath (0,1) of vpath2(__a,__b)}
	\fmfi{boson}{point 1 of vpath1(__a,__b) .. point 1 of vpath2(__a,__b)}
	\fmfi{boson}{subpath (1,2) of vpath1(__a,__b)}
	\fmfi{boson}{subpath (1,2) of vpath2(__a,__b)}
\end{fmfgraph*}}
+
\parbox{20pt}{\begin{fmfgraph*}(20,10)
	\fmfleft{1}
	\fmfright{2}
	\fmf{phantom}{1,a,b,2}
	\fmf{boson}{1,a}
	\fmf{boson}{b,2}
	\fmf{phantom,left,tension=0,tag=1}{a,b}
	\fmf{phantom,right,tension=0,tag=2}{a,b}
	\fmffreeze
	\fmfi{boson}{subpath (0,1) of vpath1(__a,__b)}
	\fmfi{marked}{subpath (0,1) of vpath2(__a,__b)}
	\fmfi{boson}{point 1 of vpath1(__a,__b) .. point 1 of vpath2(__a,__b)}
	\fmfi{boson}{subpath (1,2) of vpath1(__a,__b)}
	\fmfi{boson}{subpath (1,2) of vpath2(__a,__b)}
\end{fmfgraph*}}
+
\parbox{20pt}{\begin{fmfgraph*}(20,10)
	\fmfleft{1}
	\fmfright{2}
	\fmf{phantom}{1,a,b,2}
	\fmf{boson}{1,a}
	\fmf{boson}{b,2}
	\fmf{phantom,left,tension=0,tag=1}{a,b}
	\fmf{phantom,right,tension=0,tag=2}{a,b}
	\fmffreeze
	\fmfi{boson}{subpath (0,1) of vpath1(__a,__b)}
	\fmfi{boson}{subpath (0,1) of vpath2(__a,__b)}
	\fmfi{marked}{point 1 of vpath1(__a,__b) .. point 1 of vpath2(__a,__b)}
	\fmfi{boson}{subpath (1,2) of vpath1(__a,__b)}
	\fmfi{boson}{subpath (1,2) of vpath2(__a,__b)}
\end{fmfgraph*}}
+
\parbox{20pt}{\begin{fmfgraph*}(20,10)
	\fmfleft{1}
	\fmfright{2}
	\fmf{phantom}{1,a,b,2}
	\fmf{boson}{1,a}
	\fmf{boson}{b,2}
	\fmf{phantom,left,tension=0,tag=1}{a,b}
	\fmf{phantom,right,tension=0,tag=2}{a,b}
	\fmffreeze
	\fmfi{boson}{subpath (0,1) of vpath1(__a,__b)}
	\fmfi{boson}{subpath (0,1) of vpath2(__a,__b)}
	\fmfi{boson}{point 1 of vpath1(__a,__b) .. point 1 of vpath2(__a,__b)}
	\fmfi{marked}{subpath (1,2) of vpath1(__a,__b)}
	\fmfi{boson}{subpath (1,2) of vpath2(__a,__b)}
\end{fmfgraph*}}
+
\parbox{20pt}{\begin{fmfgraph*}(20,10)
	\fmfleft{1}
	\fmfright{2}
	\fmf{phantom}{1,a,b,2}
	\fmf{boson}{1,a}
	\fmf{boson}{b,2}
	\fmf{phantom,left,tension=0,tag=1}{a,b}
	\fmf{phantom,right,tension=0,tag=2}{a,b}
	\fmffreeze
	\fmfi{boson}{subpath (0,1) of vpath1(__a,__b)}
	\fmfi{boson}{subpath (0,1) of vpath2(__a,__b)}
	\fmfi{boson}{point 1 of vpath1(__a,__b) .. point 1 of vpath2(__a,__b)}
	\fmfi{boson}{subpath (1,2) of vpath1(__a,__b)}
	\fmfi{marked}{subpath (1,2) of vpath2(__a,__b)}
\end{fmfgraph*}}
+
\parbox{20pt}{\begin{fmfgraph*}(20,10)
	\fmfleft{1}
	\fmfright{2}
	\fmf{phantom}{1,a,b,2}
	\fmf{boson}{1,a}
	\fmf{boson}{b,2}
	\fmf{phantom,left,tension=0,tag=1}{a,b}
	\fmf{phantom,right,tension=0,tag=2}{a,b}
	\fmffreeze
	\fmfi{marked}{subpath (0,1) of vpath1(__a,__b)}
	\fmfi{boson}{subpath (0,1) of vpath2(__a,__b)}
	\fmfi{boson}{point 1 of vpath1(__a,__b) .. point 1 of vpath2(__a,__b)}
	\fmfi{boson}{subpath (1,2) of vpath1(__a,__b)}
	\fmfi{marked}{subpath (1,2) of vpath2(__a,__b)}
\end{fmfgraph*}}
+
\parbox{20pt}{\begin{fmfgraph*}(20,10)
	\fmfleft{1}
	\fmfright{2}
	\fmf{phantom}{1,a,b,2}
	\fmf{boson}{1,a}
	\fmf{boson}{b,2}
	\fmf{phantom,left,tension=0,tag=1}{a,b}
	\fmf{phantom,right,tension=0,tag=2}{a,b}
	\fmffreeze
	\fmfi{boson}{subpath (0,1) of vpath1(__a,__b)}
	\fmfi{marked}{subpath (0,1) of vpath2(__a,__b)}
	\fmfi{boson}{point 1 of vpath1(__a,__b) .. point 1 of vpath2(__a,__b)}
	\fmfi{marked}{subpath (1,2) of vpath1(__a,__b)}
	\fmfi{boson}{subpath (1,2) of vpath2(__a,__b)}
\end{fmfgraph*}}
\\
se^{\chi_+}
\parbox{20pt}{\begin{fmfgraph*}(20,10)
	\fmfleft{1}
	\fmfright{2}
	\fmf{phantom}{1,a,b,2}
	\fmf{boson}{1,a}
	\fmf{boson}{b,2}
	\fmf{phantom,left,tension=0,tag=1}{a,b}
	\fmf{phantom,right,tension=0,tag=2}{a,b}
	\fmffreeze
	\fmfi{boson}{subpath (0,1) of vpath1(__a,__b)}
	\fmfi{boson}{subpath (0,1) of vpath2(__a,__b)}
	\fmfi{boson}{point 1 of vpath1(__a,__b) .. point 1 of vpath2(__a,__b)}
	\fmfi{boson}{subpath (1,2) of vpath1(__a,__b)}
	\fmfi{boson}{subpath (1,2) of vpath2(__a,__b)}
\end{fmfgraph*}}
&=
\parbox{20pt}{\begin{fmfgraph*}(20,10)
	\fmfleft{1}
	\fmfright{2}
	\fmf{phantom}{1,a,b,2}
	\fmf{boson}{1,a}
	\fmf{boson}{b,2}
	\fmf{phantom,left,tension=0,tag=1}{a,b}
	\fmf{phantom,right,tension=0,tag=2}{a,b}
	\fmffreeze
	\fmfi{mmarked}{subpath (0,1) of vpath1(__a,__b)}
	\fmfi{boson}{subpath (0,1) of vpath2(__a,__b)}
	\fmfi{boson}{point 1 of vpath1(__a,__b) .. point 1 of vpath2(__a,__b)}
	\fmfi{boson}{subpath (1,2) of vpath1(__a,__b)}
	\fmfi{boson}{subpath (1,2) of vpath2(__a,__b)}
\end{fmfgraph*}}
+
\parbox{20pt}{\begin{fmfgraph*}(20,10)
	\fmfleft{1}
	\fmfright{2}
	\fmf{phantom}{1,a,b,2}
	\fmf{boson}{1,a}
	\fmf{boson}{b,2}
	\fmf{phantom,left,tension=0,tag=1}{a,b}
	\fmf{phantom,right,tension=0,tag=2}{a,b}
	\fmffreeze
	\fmfi{boson}{subpath (0,1) of vpath1(__a,__b)}
	\fmfi{mmarked}{subpath (0,1) of vpath2(__a,__b)}
	\fmfi{boson}{point 1 of vpath1(__a,__b) .. point 1 of vpath2(__a,__b)}
	\fmfi{boson}{subpath (1,2) of vpath1(__a,__b)}
	\fmfi{boson}{subpath (1,2) of vpath2(__a,__b)}
\end{fmfgraph*}}
+
\parbox{20pt}{\begin{fmfgraph*}(20,10)
	\fmfleft{1}
	\fmfright{2}
	\fmf{phantom}{1,a,b,2}
	\fmf{boson}{1,a}
	\fmf{boson}{b,2}
	\fmf{phantom,left,tension=0,tag=1}{a,b}
	\fmf{phantom,right,tension=0,tag=2}{a,b}
	\fmffreeze
	\fmfi{boson}{subpath (0,1) of vpath1(__a,__b)}
	\fmfi{boson}{subpath (0,1) of vpath2(__a,__b)}
	\fmfi{mmarked}{point 1 of vpath1(__a,__b) .. point 1 of vpath2(__a,__b)}
	\fmfi{boson}{subpath (1,2) of vpath1(__a,__b)}
	\fmfi{boson}{subpath (1,2) of vpath2(__a,__b)}
\end{fmfgraph*}}
+
\parbox{20pt}{\begin{fmfgraph*}(20,10)
	\fmfleft{1}
	\fmfright{2}
	\fmf{phantom}{1,a,b,2}
	\fmf{boson}{1,a}
	\fmf{boson}{b,2}
	\fmf{phantom,left,tension=0,tag=1}{a,b}
	\fmf{phantom,right,tension=0,tag=2}{a,b}
	\fmffreeze
	\fmfi{boson}{subpath (0,1) of vpath1(__a,__b)}
	\fmfi{boson}{subpath (0,1) of vpath2(__a,__b)}
	\fmfi{boson}{point 1 of vpath1(__a,__b) .. point 1 of vpath2(__a,__b)}
	\fmfi{mmarked}{subpath (1,2) of vpath1(__a,__b)}
	\fmfi{boson}{subpath (1,2) of vpath2(__a,__b)}
\end{fmfgraph*}}
+
\parbox{20pt}{\begin{fmfgraph*}(20,10)
	\fmfleft{1}
	\fmfright{2}
	\fmf{phantom}{1,a,b,2}
	\fmf{boson}{1,a}
	\fmf{boson}{b,2}
	\fmf{phantom,left,tension=0,tag=1}{a,b}
	\fmf{phantom,right,tension=0,tag=2}{a,b}
	\fmffreeze
	\fmfi{boson}{subpath (0,1) of vpath1(__a,__b)}
	\fmfi{boson}{subpath (0,1) of vpath2(__a,__b)}
	\fmfi{boson}{point 1 of vpath1(__a,__b) .. point 1 of vpath2(__a,__b)}
	\fmfi{boson}{subpath (1,2) of vpath1(__a,__b)}
	\fmfi{mmarked}{subpath (1,2) of vpath2(__a,__b)}
\end{fmfgraph*}}
\\&\qquad-
\parbox{20pt}{\begin{fmfgraph*}(20,10)
	\fmfleft{1}
	\fmfright{2}
	\fmf{phantom}{1,a,b,2}
	\fmf{boson}{1,a}
	\fmf{boson}{b,2}
	\fmf{phantom,left,tension=0,tag=1}{a,b}
	\fmf{phantom,right,tension=0,tag=2}{a,b}
	\fmffreeze
	\fmfi{marked}{subpath (0,1) of vpath1(__a,__b)}
	\fmfi{boson}{subpath (0,1) of vpath2(__a,__b)}
	\fmfi{boson}{point 1 of vpath1(__a,__b) .. point 1 of vpath2(__a,__b)}
	\fmfi{boson}{subpath (1,2) of vpath1(__a,__b)}
	\fmfi{mmarked}{subpath (1,2) of vpath2(__a,__b)}
\end{fmfgraph*}}
-
\parbox{20pt}{\begin{fmfgraph*}(20,10)
	\fmfleft{1}
	\fmfright{2}
	\fmf{phantom}{1,a,b,2}
	\fmf{boson}{1,a}
	\fmf{boson}{b,2}
	\fmf{phantom,left,tension=0,tag=1}{a,b}
	\fmf{phantom,right,tension=0,tag=2}{a,b}
	\fmffreeze
	\fmfi{boson}{subpath (0,1) of vpath1(__a,__b)}
	\fmfi{marked}{subpath (0,1) of vpath2(__a,__b)}
	\fmfi{boson}{point 1 of vpath1(__a,__b) .. point 1 of vpath2(__a,__b)}
	\fmfi{mmarked}{subpath (1,2) of vpath1(__a,__b)}
	\fmfi{boson}{subpath (1,2) of vpath2(__a,__b)}
\end{fmfgraph*}}
+
\parbox{20pt}{\begin{fmfgraph*}(20,10)
	\fmfleft{1}
	\fmfright{2}
	\fmf{phantom}{1,a,b,2}
	\fmf{boson}{1,a}
	\fmf{boson}{b,2}
	\fmf{phantom,left,tension=0,tag=1}{a,b}
	\fmf{phantom,right,tension=0,tag=2}{a,b}
	\fmffreeze
	\fmfi{boson}{subpath (0,1) of vpath1(__a,__b)}
	\fmfi{mmarked}{subpath (0,1) of vpath2(__a,__b)}
	\fmfi{boson}{point 1 of vpath1(__a,__b) .. point 1 of vpath2(__a,__b)}
	\fmfi{marked}{subpath (1,2) of vpath1(__a,__b)}
	\fmfi{boson}{subpath (1,2) of vpath2(__a,__b)}
\end{fmfgraph*}}
+
\parbox{20pt}{\begin{fmfgraph*}(20,10)
	\fmfleft{1}
	\fmfright{2}
	\fmf{phantom}{1,a,b,2}
	\fmf{boson}{1,a}
	\fmf{boson}{b,2}
	\fmf{phantom,left,tension=0,tag=1}{a,b}
	\fmf{phantom,right,tension=0,tag=2}{a,b}
	\fmffreeze
	\fmfi{mmarked}{subpath (0,1) of vpath1(__a,__b)}
	\fmfi{boson}{subpath (0,1) of vpath2(__a,__b)}
	\fmfi{boson}{point 1 of vpath1(__a,__b) .. point 1 of vpath2(__a,__b)}
	\fmfi{boson}{subpath (1,2) of vpath1(__a,__b)}
	\fmfi{marked}{subpath (1,2) of vpath2(__a,__b)}
\end{fmfgraph*}}
\\
\sigma e^{\chi_+}
\parbox{20pt}{\begin{fmfgraph*}(20,10)
	\fmfleft{1}
	\fmfright{2}
	\fmf{phantom}{1,a,b,2}
	\fmf{boson}{1,a}
	\fmf{boson}{b,2}
	\fmf{phantom,left,tension=0,tag=1}{a,b}
	\fmf{phantom,right,tension=0,tag=2}{a,b}
	\fmffreeze
	\fmfi{boson}{subpath (0,1) of vpath1(__a,__b)}
	\fmfi{boson}{subpath (0,1) of vpath2(__a,__b)}
	\fmfi{boson}{point 1 of vpath1(__a,__b) .. point 1 of vpath2(__a,__b)}
	\fmfi{boson}{subpath (1,2) of vpath1(__a,__b)}
	\fmfi{boson}{subpath (1,2) of vpath2(__a,__b)}
\end{fmfgraph*}}
&= -
\parbox{20pt}{\begin{fmfgraph*}(20,10)
	\fmfleft{1}
	\fmfright{2}
	\fmf{phantom}{1,a,b,2}
	\fmf{boson}{1,a}
	\fmf{boson}{b,2}
	\fmf{phantom,left,tension=0,tag=1}{a,b}
	\fmf{phantom,right,tension=0,tag=2}{a,b}
	\fmffreeze
	\fmfi{mmarked}{subpath (0,1) of vpath1(__a,__b)}
	\fmfi{boson}{subpath (0,1) of vpath2(__a,__b)}
	\fmfi{boson}{point 1 of vpath1(__a,__b) .. point 1 of vpath2(__a,__b)}
	\fmfi{boson}{subpath (1,2) of vpath1(__a,__b)}
	\fmfi{boson}{subpath (1,2) of vpath2(__a,__b)}
\end{fmfgraph*}}
-
\parbox{20pt}{\begin{fmfgraph*}(20,10)
	\fmfleft{1}
	\fmfright{2}
	\fmf{phantom}{1,a,b,2}
	\fmf{boson}{1,a}
	\fmf{boson}{b,2}
	\fmf{phantom,left,tension=0,tag=1}{a,b}
	\fmf{phantom,right,tension=0,tag=2}{a,b}
	\fmffreeze
	\fmfi{boson}{subpath (0,1) of vpath1(__a,__b)}
	\fmfi{mmarked}{subpath (0,1) of vpath2(__a,__b)}
	\fmfi{boson}{point 1 of vpath1(__a,__b) .. point 1 of vpath2(__a,__b)}
	\fmfi{boson}{subpath (1,2) of vpath1(__a,__b)}
	\fmfi{boson}{subpath (1,2) of vpath2(__a,__b)}
\end{fmfgraph*}}
-
\parbox{20pt}{\begin{fmfgraph*}(20,10)
	\fmfleft{1}
	\fmfright{2}
	\fmf{phantom}{1,a,b,2}
	\fmf{boson}{1,a}
	\fmf{boson}{b,2}
	\fmf{phantom,left,tension=0,tag=1}{a,b}
	\fmf{phantom,right,tension=0,tag=2}{a,b}
	\fmffreeze
	\fmfi{boson}{subpath (0,1) of vpath1(__a,__b)}
	\fmfi{boson}{subpath (0,1) of vpath2(__a,__b)}
	\fmfi{mmarked}{point 1 of vpath1(__a,__b) .. point 1 of vpath2(__a,__b)}
	\fmfi{boson}{subpath (1,2) of vpath1(__a,__b)}
	\fmfi{boson}{subpath (1,2) of vpath2(__a,__b)}
\end{fmfgraph*}}
-
\parbox{20pt}{\begin{fmfgraph*}(20,10)
	\fmfleft{1}
	\fmfright{2}
	\fmf{phantom}{1,a,b,2}
	\fmf{boson}{1,a}
	\fmf{boson}{b,2}
	\fmf{phantom,left,tension=0,tag=1}{a,b}
	\fmf{phantom,right,tension=0,tag=2}{a,b}
	\fmffreeze
	\fmfi{boson}{subpath (0,1) of vpath1(__a,__b)}
	\fmfi{boson}{subpath (0,1) of vpath2(__a,__b)}
	\fmfi{boson}{point 1 of vpath1(__a,__b) .. point 1 of vpath2(__a,__b)}
	\fmfi{mmarked}{subpath (1,2) of vpath1(__a,__b)}
	\fmfi{boson}{subpath (1,2) of vpath2(__a,__b)}
\end{fmfgraph*}}
-
\parbox{20pt}{\begin{fmfgraph*}(20,10)
	\fmfleft{1}
	\fmfright{2}
	\fmf{phantom}{1,a,b,2}
	\fmf{boson}{1,a}
	\fmf{boson}{b,2}
	\fmf{phantom,left,tension=0,tag=1}{a,b}
	\fmf{phantom,right,tension=0,tag=2}{a,b}
	\fmffreeze
	\fmfi{boson}{subpath (0,1) of vpath1(__a,__b)}
	\fmfi{boson}{subpath (0,1) of vpath2(__a,__b)}
	\fmfi{boson}{point 1 of vpath1(__a,__b) .. point 1 of vpath2(__a,__b)}
	\fmfi{boson}{subpath (1,2) of vpath1(__a,__b)}
	\fmfi{mmarked}{subpath (1,2) of vpath2(__a,__b)}
\end{fmfgraph*}}
\\&\qquad-
\parbox{20pt}{\begin{fmfgraph*}(20,10)
	\fmfleft{1}
	\fmfright{2}
	\fmf{phantom}{1,a,b,2}
	\fmf{boson}{1,a}
	\fmf{boson}{b,2}
	\fmf{phantom,left,tension=0,tag=1}{a,b}
	\fmf{phantom,right,tension=0,tag=2}{a,b}
	\fmffreeze
	\fmfi{mmarked}{subpath (0,1) of vpath1(__a,__b)}
	\fmfi{boson}{subpath (0,1) of vpath2(__a,__b)}
	\fmfi{boson}{point 1 of vpath1(__a,__b) .. point 1 of vpath2(__a,__b)}
	\fmfi{boson}{subpath (1,2) of vpath1(__a,__b)}
	\fmfi{marked}{subpath (1,2) of vpath2(__a,__b)}
\end{fmfgraph*}}
+
\parbox{20pt}{\begin{fmfgraph*}(20,10)
	\fmfleft{1}
	\fmfright{2}
	\fmf{phantom}{1,a,b,2}
	\fmf{boson}{1,a}
	\fmf{boson}{b,2}
	\fmf{phantom,left,tension=0,tag=1}{a,b}
	\fmf{phantom,right,tension=0,tag=2}{a,b}
	\fmffreeze
	\fmfi{marked}{subpath (0,1) of vpath1(__a,__b)}
	\fmfi{boson}{subpath (0,1) of vpath2(__a,__b)}
	\fmfi{boson}{point 1 of vpath1(__a,__b) .. point 1 of vpath2(__a,__b)}
	\fmfi{boson}{subpath (1,2) of vpath1(__a,__b)}
	\fmfi{mmarked}{subpath (1,2) of vpath2(__a,__b)}
\end{fmfgraph*}}
-
\parbox{20pt}{\begin{fmfgraph*}(20,10)
	\fmfleft{1}
	\fmfright{2}
	\fmf{phantom}{1,a,b,2}
	\fmf{boson}{1,a}
	\fmf{boson}{b,2}
	\fmf{phantom,left,tension=0,tag=1}{a,b}
	\fmf{phantom,right,tension=0,tag=2}{a,b}
	\fmffreeze
	\fmfi{boson}{subpath (0,1) of vpath1(__a,__b)}
	\fmfi{mmarked}{subpath (0,1) of vpath2(__a,__b)}
	\fmfi{boson}{point 1 of vpath1(__a,__b) .. point 1 of vpath2(__a,__b)}
	\fmfi{marked}{subpath (1,2) of vpath1(__a,__b)}
	\fmfi{boson}{subpath (1,2) of vpath2(__a,__b)}
\end{fmfgraph*}}
+
\parbox{20pt}{\begin{fmfgraph*}(20,10)
	\fmfleft{1}
	\fmfright{2}
	\fmf{phantom}{1,a,b,2}
	\fmf{boson}{1,a}
	\fmf{boson}{b,2}
	\fmf{phantom,left,tension=0,tag=1}{a,b}
	\fmf{phantom,right,tension=0,tag=2}{a,b}
	\fmffreeze
	\fmfi{boson}{subpath (0,1) of vpath1(__a,__b)}
	\fmfi{marked}{subpath (0,1) of vpath2(__a,__b)}
	\fmfi{boson}{point 1 of vpath1(__a,__b) .. point 1 of vpath2(__a,__b)}
	\fmfi{mmarked}{subpath (1,2) of vpath1(__a,__b)}
	\fmfi{boson}{subpath (1,2) of vpath2(__a,__b)}
\end{fmfgraph*}}\\
Se^{\chi_+}
\parbox{20pt}{\begin{fmfgraph*}(20,10)
	\fmfleft{1}
	\fmfright{2}
	\fmf{phantom}{1,a,b,2}
	\fmf{boson}{1,a}
	\fmf{boson}{b,2}
	\fmf{phantom,left,tension=0,tag=1}{a,b}
	\fmf{phantom,right,tension=0,tag=2}{a,b}
	\fmffreeze
	\fmfi{boson}{subpath (0,1) of vpath1(__a,__b)}
	\fmfi{boson}{subpath (0,1) of vpath2(__a,__b)}
	\fmfi{boson}{point 1 of vpath1(__a,__b) .. point 1 of vpath2(__a,__b)}
	\fmfi{boson}{subpath (1,2) of vpath1(__a,__b)}
	\fmfi{boson}{subpath (1,2) of vpath2(__a,__b)}
\end{fmfgraph*}}
&=0 \end{align*}
\end{exmp}

This finishes our considerations of graph homology; we have proved Theorem \ref{chithm}.

\subsection{The ghost cycle generator $\delta_+^C$}
We now turn to an investigation of the ghost sector through cycle homology.
\begin{defn}\label{4-22}
Let $\mathcal C_\Gamma$ be the set of cycles in $\Gamma$. We write $\mathcal C_\Gamma = \{C_1,C_2,\ldots\}$.
\begin{equation*} \delta_+\Gamma
= \sum_{C\in\mathcal C_\Gamma}\delta_+^{C}\Gamma, \end{equation*}
where
\begin{samepage}\begin{equation*} \delta_+^{C}\Gamma = \begin{cases}
0 & \text{if $C$ has a vertex which has an adjacent marked or ghost edge} \\
\Gamma_{C\rightsquigarrow\ \dotsedge} & \text{otherwise.}
\end{cases}\end{equation*}
Note that an (unoriented) ghost cycle is the short-hand notation for the sum of the two orientations:
\begin{equation*}
\parbox{20pt}{\begin{fmfgraph*}(20,20)
	\fmfsurroundn{e}{6}
	\fmf{boson,tension=2}{e1,v1}
	\fmf{boson,tension=2}{e2,v2}
	\fmf{boson,tension=2}{e3,v3}
	\fmf{boson,tension=2}{e4,v4}
	\fmf{boson,tension=2}{e5,v5}
	\fmf{phantom,tension=2}{e6,v6}
	\fmf{phantom}{v1,v2,v3,v4,v5,v6,v1}
	\fmffreeze
	\fmf{dots,right}{v1,v4,v1}
\end{fmfgraph*}}
=
\parbox{20pt}{\begin{fmfgraph*}(20,20)
	\fmfsurroundn{e}{6}
	\fmf{boson,tension=2}{e1,v1}
	\fmf{boson,tension=2}{e2,v2}
	\fmf{boson,tension=2}{e3,v3}
	\fmf{boson,tension=2}{e4,v4}
	\fmf{boson,tension=2}{e5,v5}
	\fmf{phantom,tension=2}{e6,v6}
	\fmf{phantom}{v1,v2,v3,v4,v5,v6,v1}
	\fmffreeze
	\fmf{phantom,right,tag=1}{v1,v4}
	\fmf{phantom,right,tag=2}{v4,v1}
	\fmffreeze
	\fmfi{ghost}{subpath (0,2/3) of vpath1(__v1,__v4)}
	\fmfi{ghost}{subpath (2/3,4/3) of vpath1(__v1,__v4)}
	\fmfi{ghost}{subpath (4/3,2) of vpath1(__v1,__v4)}
	\fmfi{ghost}{subpath (0,2/3) of vpath2(__v1,__v4)}
	\fmfi{dots}{subpath (2/3,4/3) of vpath2(__v1,__v4)}
	\fmfi{dots}{subpath (4/3,2) of vpath2(__v1,__v4)}
\end{fmfgraph*}}
+
\parbox{20pt}{\begin{fmfgraph*}(20,20)
	\fmfsurroundn{e}{6}
	\fmf{boson,tension=2}{e1,v1}
	\fmf{boson,tension=2}{e2,v2}
	\fmf{boson,tension=2}{e3,v3}
	\fmf{boson,tension=2}{e4,v4}
	\fmf{boson,tension=2}{e5,v5}
	\fmf{phantom,tension=2}{e6,v6}
	\fmf{phantom}{v1,v2,v3,v4,v5,v6,v1}
	\fmffreeze
	\fmf{phantom,left,tag=1}{v1,v4}
	\fmf{phantom,left,tag=2}{v4,v1}
	\fmffreeze
	\fmfi{dots}{subpath (0,2/3) of vpath1(__v1,__v4)}
	\fmfi{dots}{subpath (2/3,4/3) of vpath1(__v1,__v4)}
	\fmfi{ghost}{subpath (4/3,2) of vpath1(__v1,__v4)}
	\fmfi{ghost}{subpath (0,2/3) of vpath2(__v1,__v4)}
	\fmfi{ghost}{subpath (2/3,4/3) of vpath2(__v1,__v4)}
	\fmfi{ghost}{subpath (4/3,2) of vpath2(__v1,__v4)}
\end{fmfgraph*}}.
\end{equation*}
\end{samepage}
\end{defn}

\begin{lem}\label{lem:skeletongraphs-delta}\begin{enumerate}
\item For any graph $\Gamma\in\bar{\mathcal G}^{n,l}$:
\begin{equation*} \frac1{\Sym(\Gamma)} \frac{(\delta_+)^k}{k!} \Gamma
= \sum_{\substack{\Gamma'\in\bar{\mathcal G}^{n,l}_{k\ghostloop}\\ \skeleton(\Gamma')=\Gamma}} \frac1{\Sym(\Gamma')} \Gamma'.  \end{equation*}
\item For any $k \geq 0$, $\frac{(\delta_+)^k}{k!} X^{n,l} = X^{n,l}_{k\ghostloop}$.
\item We have $e^{\delta_+} X^{n,l} = X^{n,l}_{/\ghostloop}$.
\end{enumerate}
\end{lem}
\begin{proof} Analogous to Lemma \ref{lem:||&4-vertices}. \end{proof}

\begin{exmp}An example for Lemma \ref{lem:skeletongraphs-delta}.i is
\begin{equation*}
\tfrac12\delta_+
\parbox{20pt}{\begin{fmfgraph*}(20,10)
	\fmfleft{1}
	\fmfright{2}
	\fmf{phantom}{1,a,d,2}
	\fmf{boson}{1,a}
	\fmf{boson}{d,2}
	\fmf{phantom,left,tension=0,tag=1}{a,d}
	\fmf{phantom,right,tension=0,tag=2}{a,d}
	\fmffreeze
	\fmfi{boson}{subpath (0,2/3) of vpath1(__a,__d)}
	\fmfi{boson}{point 2/3 of vpath1(__a,__d){dir 90} .. point 4/3 of vpath1(__a,__d){dir -90}}
	\fmfi{boson}{point 2/3 of vpath1(__a,__d){dir -90} .. point 4/3 of vpath1(__a,__d){dir 90}}
	\fmfi{boson}{subpath (4/3,2) of vpath1(__a,__d)}
	\fmfi{boson}{subpath (0,2) of vpath2(__a,__d)}
\end{fmfgraph*}}
= \tfrac12 \Big(
\parbox{20pt}{\begin{fmfgraph*}(20,10)
	\fmfleft{1}
	\fmfright{2}
	\fmf{phantom}{1,a,d,2}
	\fmf{boson}{1,a}
	\fmf{boson}{d,2}
	\fmf{phantom,left,tension=0,tag=1}{a,d}
	\fmf{phantom,right,tension=0,tag=2}{a,d}
	\fmffreeze
	\fmfi{dots}{subpath (0,2/3) of vpath1(__a,__d)}
	\fmfi{dots}{point 2/3 of vpath1(__a,__d){dir 90} .. point 4/3 of vpath1(__a,__d){dir -90}}
	\fmfi{boson}{point 2/3 of vpath1(__a,__d){dir -90} .. point 4/3 of vpath1(__a,__d){dir 90}}
	\fmfi{dots}{subpath (4/3,2) of vpath1(__a,__d)}
	\fmfi{dots}{subpath (0,2) of vpath2(__a,__d)}
\end{fmfgraph*}}
+
\parbox{20pt}{\begin{fmfgraph*}(20,10)
	\fmfleft{1}
	\fmfright{2}
	\fmf{phantom}{1,a,d,2}
	\fmf{boson}{1,a}
	\fmf{boson}{d,2}
	\fmf{phantom,left,tension=0,tag=1}{a,d}
	\fmf{phantom,right,tension=0,tag=2}{a,d}
	\fmffreeze
	\fmfi{dots}{subpath (0,2/3) of vpath1(__a,__d)}
	\fmfi{boson}{point 2/3 of vpath1(__a,__d){dir 90} .. point 4/3 of vpath1(__a,__d){dir -90}}
	\fmfi{dots}{point 2/3 of vpath1(__a,__d){dir -90} .. point 4/3 of vpath1(__a,__d){dir 90}}
	\fmfi{dots}{subpath (4/3,2) of vpath1(__a,__d)}
	\fmfi{dots}{subpath (0,2) of vpath2(__a,__d)}
\end{fmfgraph*}}
+
\parbox{20pt}{\begin{fmfgraph*}(20,10)
	\fmfleft{1}
	\fmfright{2}
	\fmf{phantom}{1,a,d,2}
	\fmf{boson}{1,a}
	\fmf{boson}{d,2}
	\fmf{phantom,left,tension=0,tag=1}{a,d}
	\fmf{phantom,right,tension=0,tag=2}{a,d}
	\fmffreeze
	\fmfi{boson}{subpath (0,2/3) of vpath1(__a,__d)}
	\fmfi{dots}{point 2/3 of vpath1(__a,__d){dir 90} .. point 4/3 of vpath1(__a,__d){dir -90}}
	\fmfi{dots}{point 2/3 of vpath1(__a,__d){dir -90} .. point 4/3 of vpath1(__a,__d){dir 90}}
	\fmfi{boson}{subpath (4/3,2) of vpath1(__a,__d)}
	\fmfi{boson}{subpath (0,2) of vpath2(__a,__d)}
\end{fmfgraph*}}
\Big) =
\parbox{20pt}{\begin{fmfgraph*}(20,10)
	\fmfleft{1}
	\fmfright{2}
	\fmf{phantom}{1,a,d,2}
	\fmf{boson}{1,a}
	\fmf{boson}{d,2}
	\fmf{phantom,left,tension=0,tag=1}{a,d}
	\fmf{phantom,right,tension=0,tag=2}{a,d}
	\fmffreeze
	\fmfi{dots}{subpath (0,2/3) of vpath1(__a,__d)}
	\fmfi{dots}{point 2/3 of vpath1(__a,__d){dir 90} .. point 4/3 of vpath1(__a,__d){dir -90}}
	\fmfi{boson}{point 2/3 of vpath1(__a,__d){dir -90} .. point 4/3 of vpath1(__a,__d){dir 90}}
	\fmfi{dots}{subpath (4/3,2) of vpath1(__a,__d)}
	\fmfi{dots}{subpath (0,2) of vpath2(__a,__d)}
\end{fmfgraph*}}
+ \tfrac12
\parbox{20pt}{\begin{fmfgraph*}(20,10)
	\fmfleft{1}
	\fmfright{2}
	\fmf{phantom}{1,a,d,2}
	\fmf{boson}{1,a}
	\fmf{boson}{d,2}
	\fmf{phantom,left,tension=0,tag=1}{a,d}
	\fmf{phantom,right,tension=0,tag=2}{a,d}
	\fmffreeze
	\fmfi{boson}{subpath (0,2/3) of vpath1(__a,__d)}
	\fmfi{dots}{point 2/3 of vpath1(__a,__d){dir 90} .. point 4/3 of vpath1(__a,__d){dir -90}}
	\fmfi{dots}{point 2/3 of vpath1(__a,__d){dir -90} .. point 4/3 of vpath1(__a,__d){dir 90}}
	\fmfi{boson}{subpath (4/3,2) of vpath1(__a,__d)}
	\fmfi{boson}{subpath (0,2) of vpath2(__a,__d)}
\end{fmfgraph*}}
\end{equation*}
\end{exmp}

\begin{rem} The operators $\chi_+^e$ and $\delta_+^C$ commute, hence so do $\chi_+$ and $\delta_+$. \end{rem}

\begin{cor}\begin{enumerate}
\item The combinatorial Green's functions $X^{n,l}_{k\markededge,\tilde l\ghostloop}$ can be written as
\begin{equation*} X^{n,l}_{k\markededge,\tilde l\ghostloop} = \frac{\chi_+^k\delta_+^{\tilde l}}{k!\tilde l!} X^{n,l}_{0\markededge,0\ghostloop}. \end{equation*}
\item The full combinatorial Green's function can be written as
\begin{equation*} X^{n,l}_{/\markededge,/\ghostloop} = e^{\chi_+}e^{\delta_+} X^{n,l}_{0\markededge,0\ghostloop}. \end{equation*}
\end{enumerate}\end{cor}

\subsection{The cycle differential $t$}
\begin{defn}\label{4-27} For a graph $\Gamma$ choose a labelling of the cycles $C_1,C_2,\ldots\in\mathcal C_\Gamma$. We define a derivation $t:H \to H$ acting on graphs as:
\begin{equation*} t\Gamma = \sum_{C_i\in\mathcal C_\Gamma} (-)^{\#\{C_{i'}\in\mathcal C_{\Gamma\rm gh}\ |\ i'<i\}} t_{C_i}\Gamma \end{equation*}
where
$$ t_C\Gamma = \begin{cases}
0 & \text{if $C$ has a vertex which has an adjacent marked or ghost edge} \\
\Gamma_{C\rightsquigarrow\ \ddotsedge} & \text{otherwise.}
\end{cases}$$

Next, we want to distinguish the markings created by $\delta_+$ and $t$. Therefore we draw the former with little circles instead of dots. We denote: $\mathcal C_{\Gamma\dotsedge}\cup\mathcal C_{\Gamma\ddotsedge} = \mathcal C_{\Gamma\rm gh}\subset\mathcal C_\Gamma$.
\end{defn}

\begin{prop}
$t$ is a differential: $t^2\Gamma = 0$.
\end{prop}
\begin{proof} Analogous to Proposition \ref{ssquare}. 
\end{proof}

\begin{exmp}
Consider the graph
\begin{equation*}
\Gamma=
\parbox[c][35pt]{95pt}{\centering\scriptsize
\begin{fmfgraph*}(80,20)
	\fmfleft{1}
	\fmfright{2}
	\fmfv{label=$1$,label.dist=2}{1}
	\fmfv{label=$2$,label.dist=2}{2}
	\fmf{boson,tension=2}{1,a}
	\fmf{boson,tension=.5,left,label=$3$,label.dist=2}{a,b}
	\fmf{boson,tension=.5,left,label=$4$,label.dist=2}{b,a}
	\fmf{boson,label=$5$,label.dist=2}{b,c}
	\fmf{boson,tension=.5,left,label=$6$,label.dist=2}{c,d}
	\fmf{boson,tension=.5,left,label=$7$,label.dist=2}{d,c}
	\fmf{boson,tension=2}{d,2}
	\fmfv{label=a,label.dist=2,label.angle=0}{a}
	\fmfv{label=b,label.dist=2,label.angle=180}{b}
	\fmfv{label=c,label.dist=2,label.angle=0}{c}
	\fmfv{label=d,label.dist=2,label.angle=180}{d}
\end{fmfgraph*}}
\end{equation*}
and label the two cycles:
\begin{align*}
C_1 = 
\parbox[c][35pt]{40pt}{\centering\scriptsize
\begin{fmfgraph*}(40,20)
	\fmfleft{1}
	\fmfright{5}
	\fmf{boson,tension=2}{1,a}
	\fmf{boson,tension=.5,left,label=$3$,label.dist=2}{a,b}
	\fmf{boson,tension=.5,left,label=$4$,label.dist=2}{b,a}
	\fmf{boson,tension=2}{b,5}
	\fmfv{label=a,label.dist=2,label.angle=0}{a}
	\fmfv{label=b,label.dist=2,label.angle=180}{b}
\end{fmfgraph*}}
,\qquad C_2 = 
\parbox[c][35pt]{40pt}{\centering\scriptsize
\begin{fmfgraph*}(40,20)
	\fmfleft{5}
	\fmfright{2}
	\fmf{boson,tension=2}{5,c}
	\fmf{boson,tension=.5,left,label=$6$,label.dist=2}{c,d}
	\fmf{boson,tension=.5,left,label=$7$,label.dist=2}{d,c}
	\fmf{boson,tension=2}{d,2}
	\fmfv{label=c,label.dist=2,label.angle=0}{c}
	\fmfv{label=d,label.dist=2,label.angle=180}{d}
\end{fmfgraph*}}
.
\end{align*}
Then:
\begin{equation*}\begin{split}
t^2
\parbox{40pt}{\begin{fmfgraph*}(40,10)
	\fmfleft{1}
	\fmfright{2}
	\fmf{boson,tension=2}{1,a}
	\fmf{boson,tension=.5,left}{a,b,a}
	\fmf{boson}{b,c}
	\fmf{boson,tension=.5,left}{c,d,c}
	\fmf{boson,tension=2}{d,2}
\end{fmfgraph*}}
&= t
\parbox{40pt}{\begin{fmfgraph*}(40,10)
	\fmfleft{1}
	\fmfright{2}
	\fmf{boson,tension=2}{1,a}
	\fmf{dbl_dots,right,tension=.5}{a,b,a}
	\fmf{boson}{b,c}
	\fmf{boson,tension=.5,left}{c,d,c}
	\fmf{boson,tension=2}{d,2}
\end{fmfgraph*}}
+ t
\parbox{40pt}{\begin{fmfgraph*}(40,10)
	\fmfleft{1}
	\fmfright{2}
	\fmf{boson,tension=2}{1,a}
	\fmf{boson,right,tension=.5}{a,b,a}
	\fmf{boson}{b,c}
	\fmf{dbl_dots,tension=.5,left}{c,d,c}
	\fmf{boson,tension=2}{d,2}
\end{fmfgraph*}}
\\
&= -
\parbox{40pt}{\begin{fmfgraph*}(40,10)
	\fmfleft{1}
	\fmfright{2}
	\fmf{boson,tension=2}{1,a}
	\fmf{dbl_dots,right,tension=.5}{a,b,a}
	\fmf{boson}{b,c}
	\fmf{dbl_dots,tension=.5,left}{c,d,c}
	\fmf{boson,tension=2}{d,2}
\end{fmfgraph*}}
+
\parbox{40pt}{\begin{fmfgraph*}(40,10)
	\fmfleft{1}
	\fmfright{2}
	\fmf{boson,tension=2}{1,a}
	\fmf{dbl_dots,right,tension=.5}{a,b,a}
	\fmf{boson}{b,c}
	\fmf{dbl_dots,tension=.5,left}{c,d,c}
	\fmf{boson,tension=2}{d,2}
\end{fmfgraph*}}
=0.
\end{split}\end{equation*}
\end{exmp}

\subsection{The differential $T$}
The $T$-homology checks that the longitudinal degrees of freedom in a loop through 3-gluon vertices are appropriately matched by ghost loops,
so that physical amplitudes are in the kernel of $T$. 
Hence, we define
\begin{defn}\label{4-30}
A derivation $T:H \to H$ is given by $T = t + \tau$ where
\begin{equation*} \tau\Gamma = (-)^{\#\mathcal C_{\Gamma\rm gh}} \sum_{C_i\in\mathcal C_{\Gamma\dotsedge}} (-)^{\#\{C_{i}'\in\mathcal C_{\Gamma\ \dotsedge} \ |\ i'>i\}} \tau_{C_i}\Gamma \end{equation*}
and
\begin{equation*} \tau_{C_i}\Gamma = \Gamma_{C_i\rightsquigarrow\ \ddotsedge}. \end{equation*}
\end{defn}

\begin{prop}
$T$ is a differential: $T^2\Gamma = 0$.
\end{prop}
\begin{proof}
As in Proposition \ref{Ssquare} this follows from
$t\tau\Gamma=-\tau t\Gamma$ and $\tau^2\Gamma=0$. 
\end{proof}

\begin{exmp}
\begin{align*}
 t\tau
\parbox{40pt}{\begin{fmfgraph*}(40,10)
	\fmfleft{1}
	\fmfright{2}
	\fmf{boson,tension=2}{1,a}
	\fmf{dots,tension=.5,left}{a,b,a}
	\fmf{boson}{b,c}
	\fmf{boson,tension=.5,left}{c,d,c}
	\fmf{boson,tension=2}{d,2}
\end{fmfgraph*}}
&= -t
\parbox{40pt}{\begin{fmfgraph*}(40,10)
	\fmfleft{1}
	\fmfright{2}
	\fmf{boson,tension=2}{1,a}
	\fmf{dbl_dots,tension=.5,left}{a,b,a}
	\fmf{boson}{b,c}
	\fmf{boson,tension=.5,left}{c,d,c}
	\fmf{boson,tension=2}{d,2}
\end{fmfgraph*}}
=
\parbox{40pt}{\begin{fmfgraph*}(40,10)
	\fmfleft{1}
	\fmfright{2}
	\fmf{boson,tension=2}{1,a}
	\fmf{dbl_dots,tension=.5,left}{a,b,a}
	\fmf{boson}{b,c}
	\fmf{dbl_dots,tension=.5,left}{c,d,c}
	\fmf{boson,tension=2}{d,2}
\end{fmfgraph*}}
\\
\tau t
\parbox{40pt}{\begin{fmfgraph*}(40,10)
	\fmfleft{1}
	\fmfright{2}
	\fmf{boson,tension=2}{1,a}
	\fmf{dots,tension=.5,left}{a,b,a}
	\fmf{boson}{b,c}
	\fmf{boson,tension=.5,left}{c,d,c}
	\fmf{boson,tension=2}{d,2}
\end{fmfgraph*}}
&= -\tau
\parbox{40pt}{\begin{fmfgraph*}(40,10)
	\fmfleft{1}
	\fmfright{2}
	\fmf{boson,tension=2}{1,a}
	\fmf{dots,tension=.5,left}{a,b,a}
	\fmf{boson}{b,c}
	\fmf{dbl_dots,tension=.5,left}{c,d,c}
	\fmf{boson,tension=2}{d,2}
\end{fmfgraph*}}
= -
\parbox{40pt}{\begin{fmfgraph*}(40,10)
	\fmfleft{1}
	\fmfright{2}
	\fmf{boson,tension=2}{1,a}
	\fmf{dbl_dots,tension=.5,left}{a,b,a}
	\fmf{boson}{b,c}
	\fmf{dbl_dots,tension=.5,left}{c,d,c}
	\fmf{boson,tension=2}{d,2}
\end{fmfgraph*}}
\\
 \tau^2
\parbox{40pt}{\begin{fmfgraph*}(40,10)
	\fmfleft{1}
	\fmfright{2}
	\fmf{boson,tension=2}{1,a}
	\fmf{dots,tension=.5,left}{a,b,a}
	\fmf{boson}{b,c}
	\fmf{boson,tension=.5,left}{c,d,c}
	\fmf{boson,tension=2}{d,2}
\end{fmfgraph*}}
&= -\tau
\parbox{40pt}{\begin{fmfgraph*}(40,10)
	\fmfleft{1}
	\fmfright{2}
	\fmf{boson,tension=2}{1,a}
	\fmf{dbl_dots,tension=.5,left}{a,b,a}
	\fmf{boson}{b,c}
	\fmf{boson,tension=.5,left}{c,d,c}
	\fmf{boson,tension=2}{d,2}
\end{fmfgraph*}}
= 0 
\\
 T^2
\parbox{40pt}{\begin{fmfgraph*}(40,10)
	\fmfleft{1}
	\fmfright{2}
	\fmf{boson,tension=2}{1,a}
	\fmf{dots,tension=.5,left}{a,b,a}
	\fmf{boson}{b,c}
	\fmf{boson,tension=.5,left}{c,d,c}
	\fmf{boson,tension=2}{d,2}
\end{fmfgraph*}}
&= 0 \end{align*}
\end{exmp}

\begin{exmp}
\begin{align*} 
t\tau
\parbox{40pt}{\begin{fmfgraph*}(40,10)
	\fmfleft{1}
	\fmfright{2}
	\fmf{boson,tension=2}{1,a}
	\fmf{dots,tension=.5,left}{a,b,a}
	\fmf{boson}{b,c}
	\fmf{dots,tension=.5,left}{c,d,c}
	\fmf{boson,tension=2}{d,2}
\end{fmfgraph*}}
&= -t
\parbox{40pt}{\begin{fmfgraph*}(40,10)
	\fmfleft{1}
	\fmfright{2}
	\fmf{boson,tension=2}{1,a}
	\fmf{dbl_dots,tension=.5,left}{a,b,a}
	\fmf{boson}{b,c}
	\fmf{dots,tension=.5,left}{c,d,c}
	\fmf{boson,tension=2}{d,2}
\end{fmfgraph*}}
+ t
\parbox{40pt}{\begin{fmfgraph*}(40,10)
	\fmfleft{1}
	\fmfright{2}
	\fmf{boson,tension=2}{1,a}
	\fmf{dots,tension=.5,left}{a,b,a}
	\fmf{boson}{b,c}
	\fmf{dbl_dots,tension=.5,left}{c,d,c}
	\fmf{boson,tension=2}{d,2}
\end{fmfgraph*}}
=0 
\\
\tau t
\parbox{40pt}{\begin{fmfgraph*}(40,10)
	\fmfleft{1}
	\fmfright{2}
	\fmf{boson,tension=2}{1,a}
	\fmf{dots,tension=.5,left}{a,b,a}
	\fmf{boson}{b,c}
	\fmf{dots,tension=.5,left}{c,d,c}
	\fmf{boson,tension=2}{d,2}
\end{fmfgraph*}}
&=0 
\\
 \tau^2
\parbox{40pt}{\begin{fmfgraph*}(40,10)
	\fmfleft{1}
	\fmfright{2}
	\fmf{boson,tension=2}{1,a}
	\fmf{dots,tension=.5,left}{a,b,a}
	\fmf{boson}{b,c}
	\fmf{dots,tension=.5,left}{c,d,c}
	\fmf{boson,tension=2}{d,2}
\end{fmfgraph*}}
&= -\tau
\parbox{40pt}{\begin{fmfgraph*}(40,10)
	\fmfleft{1}
	\fmfright{2}
	\fmf{boson,tension=2}{1,a}
	\fmf{dbl_dots,tension=.5,left}{a,b,a}
	\fmf{boson}{b,c}
	\fmf{dots,tension=.5,left}{c,d,c}
	\fmf{boson,tension=2}{d,2}
\end{fmfgraph*}}
+ \tau
\parbox{40pt}{\begin{fmfgraph*}(40,10)
	\fmfleft{1}
	\fmfright{2}
	\fmf{boson,tension=2}{1,a}
	\fmf{dots,tension=.5,left}{a,b,a}
	\fmf{boson}{b,c}
	\fmf{dbl_dots,tension=.5,left}{c,d,c}
	\fmf{boson,tension=2}{d,2}
\end{fmfgraph*}}
= -
\parbox{40pt}{\begin{fmfgraph*}(40,10)
	\fmfleft{1}
	\fmfright{2}
	\fmf{boson,tension=2}{1,a}
	\fmf{dbl_dots,tension=.5,left}{a,b,a}
	\fmf{boson}{b,c}
	\fmf{dbl_dots,tension=.5,left}{c,d,c}
	\fmf{boson,tension=2}{d,2}
\end{fmfgraph*}}
+
\parbox{40pt}{\begin{fmfgraph*}(40,10)
	\fmfleft{1}
	\fmfright{2}
	\fmf{boson,tension=2}{1,a}
	\fmf{dbl_dots,tension=.5,left}{a,b,a}
	\fmf{boson}{b,c}
	\fmf{dbl_dots,tension=.5,left}{c,d,c}
	\fmf{boson,tension=2}{d,2}
\end{fmfgraph*}}
=0 
\\
 T^2
\parbox{40pt}{\begin{fmfgraph*}(40,10)
	\fmfleft{1}
	\fmfright{2}
	\fmf{boson,tension=2}{1,a}
	\fmf{dots,tension=.5,left}{a,b,a}
	\fmf{boson}{b,c}
	\fmf{dots,tension=.5,left}{c,d,c}
	\fmf{boson,tension=2}{d,2}
\end{fmfgraph*}}
&=0 \end{align*}
\end{exmp}

\begin{prop}
Let $\Gamma$ be a graph without ghost edges. Then $Te^{\delta_+}\Gamma = 0$.
\end{prop}
\begin{proof} Analogous to Proposition \ref{Shom}. 
\end{proof}
Symmetry factors are no issue in the following example as we sum over both orientations for the two ghost lines.

\begin{exmp}
\begin{align*}
e^{\delta_+}
\parbox{40pt}{\begin{fmfgraph*}(40,10)
	\fmfleft{1}
	\fmfright{2}
	\fmf{boson,tension=2}{1,a}
	\fmf{boson,tension=.5,left}{a,b,a}
	\fmf{boson}{b,c}
	\fmf{boson,tension=.5,left}{c,d,c}
	\fmf{boson,tension=2}{d,2}
\end{fmfgraph*}}
&=
\parbox{40pt}{\begin{fmfgraph*}(40,10)
	\fmfleft{1}
	\fmfright{2}
	\fmf{boson,tension=2}{1,a}
	\fmf{boson,tension=.5,left}{a,b,a}
	\fmf{boson}{b,c}
	\fmf{boson,tension=.5,left}{c,d,c}
	\fmf{boson,tension=2}{d,2}
\end{fmfgraph*}}
+
\parbox{40pt}{\begin{fmfgraph*}(40,10)
	\fmfleft{1}
	\fmfright{2}
	\fmf{boson,tension=2}{1,a}
	\fmf{dots,tension=.5,left}{a,b,a}
	\fmf{boson}{b,c}
	\fmf{boson,tension=.5,left}{c,d,c}
	\fmf{boson,tension=2}{d,2}
\end{fmfgraph*}}
+
\parbox{40pt}{\begin{fmfgraph*}(40,10)
	\fmfleft{1}
	\fmfright{2}
	\fmf{boson,tension=2}{1,a}
	\fmf{boson,tension=.5,left}{a,b,a}
	\fmf{boson}{b,c}
	\fmf{dots,tension=.5,left}{c,d,c}
	\fmf{boson,tension=2}{d,2}
\end{fmfgraph*}}
+
\parbox{40pt}{\begin{fmfgraph*}(40,10)
	\fmfleft{1}
	\fmfright{2}
	\fmf{boson,tension=2}{1,a}
	\fmf{dots,tension=.5,left}{a,b,a}
	\fmf{boson}{b,c}
	\fmf{dots,tension=.5,left}{c,d,c}
	\fmf{boson,tension=2}{d,2}
\end{fmfgraph*}}
\\
te^{\delta_+}
\parbox{40pt}{\begin{fmfgraph*}(40,10)
	\fmfleft{1}
	\fmfright{2}
	\fmf{boson,tension=2}{1,a}
	\fmf{boson,tension=.5,left}{a,b,a}
	\fmf{boson}{b,c}
	\fmf{boson,tension=.5,left}{c,d,c}
	\fmf{boson,tension=2}{d,2}
\end{fmfgraph*}}
&=
\parbox{40pt}{\begin{fmfgraph*}(40,10)
	\fmfleft{1}
	\fmfright{2}
	\fmf{boson,tension=2}{1,a}
	\fmf{dbl_dots,tension=.5,left}{a,b,a}
	\fmf{boson}{b,c}
	\fmf{boson,tension=.5,left}{c,d,c}
	\fmf{boson,tension=2}{d,2}
\end{fmfgraph*}}
+
\parbox{40pt}{\begin{fmfgraph*}(40,10)
	\fmfleft{1}
	\fmfright{2}
	\fmf{boson,tension=2}{1,a}
	\fmf{boson,tension=.5,left}{a,b,a}
	\fmf{boson}{b,c}
	\fmf{dbl_dots,tension=.5,left}{c,d,c}
	\fmf{boson,tension=2}{d,2}
\end{fmfgraph*}}
-
\parbox{40pt}{\begin{fmfgraph*}(40,10)
	\fmfleft{1}
	\fmfright{2}
	\fmf{boson,tension=2}{1,a}
	\fmf{dbl_dots,tension=.5,left}{a,b,a}
	\fmf{boson}{b,c}
	\fmf{dots,tension=.5,left}{c,d,c}
	\fmf{boson,tension=2}{d,2}
\end{fmfgraph*}}
+
\parbox{40pt}{\begin{fmfgraph*}(40,10)
	\fmfleft{1}
	\fmfright{2}
	\fmf{boson,tension=2}{1,a}
	\fmf{dots,tension=.5,left}{a,b,a}
	\fmf{boson}{b,c}
	\fmf{dbl_dots,tension=.5,left}{c,d,c}
	\fmf{boson,tension=2}{d,2}
\end{fmfgraph*}}
\\
\tau e^{\delta_+}
\parbox{40pt}{\begin{fmfgraph*}(40,10)
	\fmfleft{1}
	\fmfright{2}
	\fmf{boson,tension=2}{1,a}
	\fmf{boson,tension=.5,left}{a,b,a}
	\fmf{boson}{b,c}
	\fmf{boson,tension=.5,left}{c,d,c}
	\fmf{boson,tension=2}{d,2}
\end{fmfgraph*}}
&= -
\parbox{40pt}{\begin{fmfgraph*}(40,10)
	\fmfleft{1}
	\fmfright{2}
	\fmf{boson,tension=2}{1,a}
	\fmf{dbl_dots,tension=.5,left}{a,b,a}
	\fmf{boson}{b,c}
	\fmf{boson,tension=.5,left}{c,d,c}
	\fmf{boson,tension=2}{d,2}
\end{fmfgraph*}}
-
\parbox{40pt}{\begin{fmfgraph*}(40,10)
	\fmfleft{1}
	\fmfright{2}
	\fmf{boson,tension=2}{1,a}
	\fmf{boson,tension=.5,left}{a,b,a}
	\fmf{boson}{b,c}
	\fmf{dbl_dots,tension=.5,left}{c,d,c}
	\fmf{boson,tension=2}{d,2}
\end{fmfgraph*}}
-
\parbox{40pt}{\begin{fmfgraph*}(40,10)
	\fmfleft{1}
	\fmfright{2}
	\fmf{boson,tension=2}{1,a}
	\fmf{dots,tension=.5,left}{a,b,a}
	\fmf{boson}{b,c}
	\fmf{dbl_dots,tension=.5,left}{c,d,c}
	\fmf{boson,tension=2}{d,2}
\end{fmfgraph*}}
+
\parbox{40pt}{\begin{fmfgraph*}(40,10)
	\fmfleft{1}
	\fmfright{2}
	\fmf{boson,tension=2}{1,a}
	\fmf{dbl_dots,tension=.5,left}{a,b,a}
	\fmf{boson}{b,c}
	\fmf{dots,tension=.5,left}{c,d,c}
	\fmf{boson,tension=2}{d,2}
\end{fmfgraph*}}
\\
T e^{\delta_+}
\parbox{40pt}{\begin{fmfgraph*}(40,10)
	\fmfleft{1}
	\fmfright{2}
	\fmf{boson,tension=2}{1,a}
	\fmf{boson,tension=.5,left}{a,b,a}
	\fmf{boson}{b,c}
	\fmf{boson,tension=.5,left}{c,d,c}
	\fmf{boson,tension=2}{d,2}
\end{fmfgraph*}}
&=0 \end{align*}
\end{exmp}

This homology ensures that longitudinal degrees of freedom propagating in loops cancel.
We summarize:
\begin{thm} 
Let $\Gamma$ be a graph without marked and ghost edges. Then
\begin{equation*} Se^{\delta_+}e^{\chi_+}\Gamma = 0,\quad  \text{ and } \quad Te^{\delta_+}e^{\chi_+}\Gamma = 0. \end{equation*}
\end{thm}

\subsection{The bicomplex}
As $[s,t]=[S,T] = 0$, we get a double complex:
$$
\xymatrix{ 
& & \vdots \ar[d]_s & \vdots \ar[d]_s &\\
& \cdots \ar[r]_t  &H_{k,\tilde l} \ar[d]_s \ar[r]_t  &H_{k,\tilde l+1} \ar[d]_s \ar[r]_t & \cdots\\
& \cdots \ar[r]_t  &H_{k+1,\tilde l} \ar[d]_s \ar[r]_{t}  &H_{k+1,\tilde l+1} \ar[d]_s \ar[r]_t & \cdots\\
&&\vdots &\vdots&
}
$$
Here, $H_{\ldots,\ldots}$ are to be regarded as reflecting the relevant vector space structure only of these spaces. The corresponding Hopf algebras and combinatorial Green
functions are discussed now.  
This bicomplex above and its relation to gauge symmetry and BRST cohomology will be the study of future work.

\section{Combinatorial Green functions}
The Hopf algebras on scalar graphs straightforwardly generalize to gauge theory graphs. In particular, the coproduct acts on the sum of all graphs contributing to a given amplitude ---the combinatorial Green function--- filtered by the number of 4-valent vertices and the number of ghost loops. Let us make this more precise.
\subsection{Gradings on the Hopf algebra}
Recall that the Hopf algebra $H$ is graded by the loop number, since the number of loops in a subgraph $\gamma \subset \Gamma$ and in the graph $\Gamma/\gamma$ add up to $|\Gamma| \equiv n(\Gamma)$. Another (multi)grading is given by the number of vertices. In order for this to be compatible with the coproduct ---creating an extra vertex in the quotient $\Gamma/\gamma$--- we say a graph $\Gamma$ with $E_E(\Gamma)$ external edges, is of multi-vertex-degree $(j_3, j_4, \ldots)$ if the number of $m$-valent vertices is equal to $j_m + \delta_{m,E_E(\Gamma)}$. One can check that this grading is compatible with the coproduct. Moreover, the two degrees are related via $\sum_m (m-2) j_m(\Gamma) = 2|\Gamma|$. This grading can be extended to involve other types of vertices ---such as $\tilde j$ ghost-gluon vertices--- cf. \cite{Sui08} for full details.

\subsection{Series of graphs}
As said, from a physical point of view, it is not so interesting to study individual graphs; rather, one considers whole sums of graphs with the same number of external lines. In this section, we will study series of 1PI graphs in the Hopf algebra $H$:
\begin{equation*}
G_{0}^{k,n}=\sum_{|\Gamma|=n,|E_E(\Gamma)|=k}\Gamma\frac{\mathrm{colour}(\Gamma)}{\mathrm{sym}(\Gamma)},
\end{equation*} 
This is the sum of all 1PI 3-regular (0 4-valent vertices) graphs with first Betti number $n$ and $k$ external gluon edges (which fixes the amplitude $r$ under consideration), normalized by their symmetry factors $\mathrm{sym}(\Gamma)$, the rank of their automorphism groups, in the denominator, and also weighted in the numerator by the corresponding colour factor 
$\mathrm{colour}(\Gamma)$:
\begin{equation*}
\mathrm{colour}(\Gamma):=\prod_{v\in V^\Gamma}R_v\prod_{e\in E_I^\Gamma}\delta_{s(e),t(e)}.
\end{equation*}
Here, $R_v$ is determined by a choice of a representation of th gauge group at $v$, and $s(e),t(e)$ are the vertex labels for source and target 
 of the internal edge $e$. Typical, $R_v$ is the adjoint representation for gluon self-interactions  or the fundamental representation for a gluon interacting with fermionic matter
 fields.
 
Similarly, we write $G_{j}^{k,n}$ for series of graphs which have $j$ 4-valent vertices, with all other vertices 3-valent. 
Also, we consider external ghost edges and loops. We let $G_{j;\tilde n}^{k,\tilde k,n}$ denote the sum of graphs which have $k$ external gluon edges, $\tilde k$ external ghost edges, $j$ 4-valent vertices, with all other vertices 3-valent, and $\tilde n$ ghost cycles (Figure \ref{fig:ghostcycle}). 
We let $\hat{G}_{j; \tilde n}^{k,n}$ be the same sum where we consider all $j$ 4-valent vertices, and all $\tilde n$ ghost cycles as marked.

Summarizing, the superscript on $G$ always indicate the external structure of the graphs in the series, whereas the subscripts indicate the 4-vertex degree, the loop number, or the ghost cycle degree.

\bigskip

\begin{figure}
\begin{center}
\includegraphics[scale=.2]{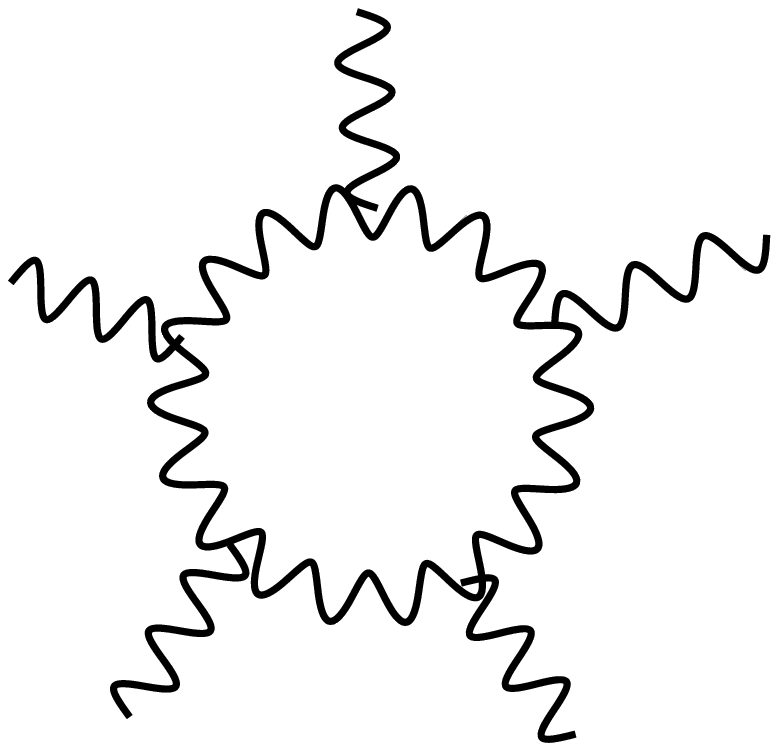}
\hspace{1cm}
\includegraphics[scale=.2]{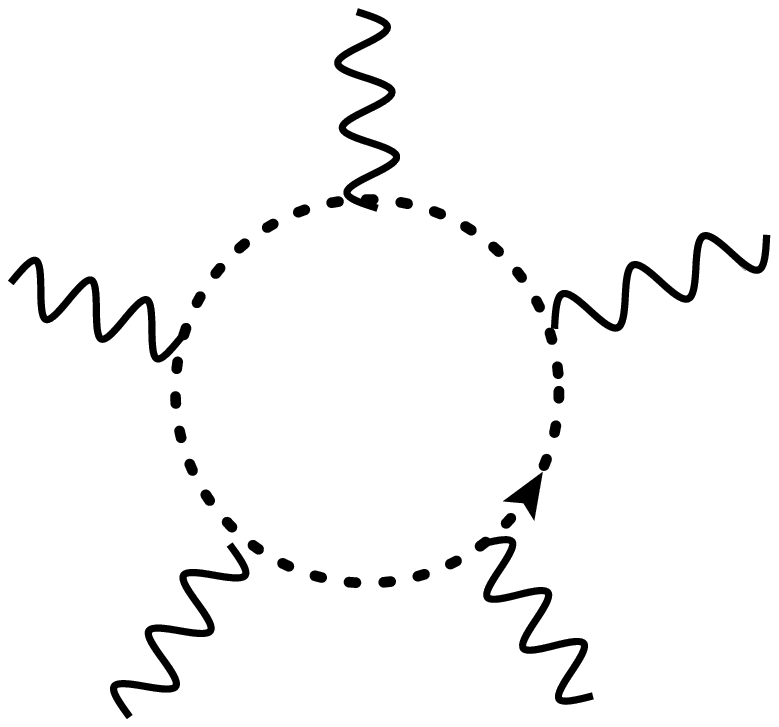}
\end{center}
\caption{A 3-regular gluon cycle (left) and an oriented  ghost cycle (right)}
\label{fig:ghostcycle}
\end{figure}

We have shown in \cite{Sui07} that we can impose the Slavnov--Taylor identities on the Hopf algebra $H$, compatibly with the coproduct, equating all of the following formal elements:
\begin{equation}
Q^{k,\tilde k} := \left(\frac{G^{k,\tilde k,/}_{/}}{(G^{2,0,/}_{/})^{k/2}(G^{0,\tilde 2,/}_{/})^{\tilde k/2}}\right)^{1/(k+\tilde k-2)} ,
\label{eq:ST}
\end{equation}
independent of the numbers $k$ and $\tilde k$ of external gluon and ghost edges,  respectively. The thus-defined single formal series $Q \equiv Q^{k,\tilde k}$ will play the role of a `charge' element in the Hopf algebra. 

\begin{prop}
The coproduct on the Green's functions read
$$
\Delta(G^{k,n}_{j_3j_4; \tilde n}) = \sum_{\begin{smallmatrix} j_m=j'_m+j''_m \\ 
n = n' + n'' \\ \tilde n = \tilde n' + \tilde n''
\end{smallmatrix}} (G^{k,n'} Q^{2n''})_{j_3'j_4'; \tilde n' } \otimes G^{k,n''}_{j_3''j_4''; \tilde n''}
$$
with $G^{k,n}_{j_3j_4; \tilde n}$ the above series of graphs of vertex multidegree $(j_3, j_4)$, first Betti number $n$ and $\tilde n$ ghost cycles. 

After taking the Slavnov--Taylor identities \eqref{eq:ST} into account, the coproduct reads on the above series of graphs
$$
\Delta(G^{k,n}) = \sum_{\begin{smallmatrix} 
n = n' + n'' 
\end{smallmatrix}} (G^k Q^{2n''})_{n'} \otimes G^{k,n''}.
$$
\end{prop}
\begin{rem}
Note that neither the lhs nor the rhs depend on $\tilde k$ in the above proposition, as 
$Q \equiv Q^{k,\tilde k}$, $\forall k,\tilde k$.
\end{rem}
\begin{rem}
The inclusion of fermions is parallel to the study of ghost edges and loops, and a mere notational exercise.
\end{rem}
Another way to describe the Green's function $G^k$ is in terms of so-called grafting operators, defined in terms of 1PI primitive graphs.
We start by considering maps $B_+^\gamma: H\to \mathrm{Aug}$, with $\mathrm{Aug}$ the augmentation ideal, which will soon lead us to non-trivial one co-cycles in the Hochschild cohomology of $H$. They are defined as follows.
$$
B_+^\gamma(h)=\sum_{\Gamma\in \langle\Gamma\rangle}\frac{{\textbf{bij}(\gamma,h,\Gamma)}}{|h|_\vee}\frac{1}{\textrm{maxf}(\Gamma)}\frac{1}{(\gamma|h)}\Gamma,
$$
where maxf$(\Gamma)$ is the number of maximal forests of $\Gamma$, $|h|_\vee$ is the number of distinct graphs obtainable by permuting edges of $h$, $\textbf{bij}(\gamma,h,\Gamma)$ is the number of bijections of external edges of $h$ with an insertion place in $\gamma$ such that the result is $\Gamma$,
and finally $(\gamma|h)$ is the number of insertion places for $h$ in $\gamma$ \cite{Kre05}. $\sum_{\Gamma\in <\Gamma>}$ indicates a sum over the linear span $\langle\Gamma\rangle$ of generators of $H$.

The sum of the $B^\gamma_+$ over all primitive 1PI Feynman graphs at a given loop order and with given residue will be denoted by $B^{l;n}_+$, as in \cite{Kre05}. More precisely,
$$
B^{k;n}_+ = \sum_{\begin{smallmatrix} \gamma ~\prim \\ |\gamma|=n \\ E_E(\gamma)=k \end{smallmatrix}} \frac{1}{\Sym(\gamma)} B^\gamma_+.
$$
With this and the above Proposition, we can show \cite[Theorem 5]{Kre05}:
\begin{align}
\label{eq:DS}
G^k &= \sum_{l=0}^\infty B_+^{k;n} (G^k Q^{2n}); \\
\label{eq:Hochschild}
\Delta( B_+^{k;n} (G^k Q^{2n})) &= B_+^{k;n} (G^k Q^{2n}) \otimes \One + (\textup{id} \otimes B_+^{k;n}) \Delta (G^k Q^{2n}).
\end{align}
Equation \eqref{eq:DS} is known as the combinatorial Dyson--Schwinger equation, while \eqref{eq:Hochschild} shows that $B_+^{k;n}$ is a Hochschild cocycle for the Hopf algebra $H$.

\subsection{The generator of ghost loops}
We again consider the  map $\delta_+: H \to H$ that replaces gluon loops in a Feynman graph by ghost loops.
\begin{rem}
In accordance with our previous definition of $\delta_+$,  it becomes an algebra derivation $\delta_+:H \to H$ by the assignment
$$
\delta_+(\Gamma) =( \tilde  l+1)\sum_{ g \subset \Gamma}
\Gamma_{g \mapsto \tilde g } .
$$
for a 1PI Feynman graph $\Gamma$ at ghost loop order $\tilde l$. The sum is over all oriented 3-regular gluon cycles $g$, and $\Gamma_{g \mapsto \tilde g }$ denotes the graph $\Gamma$ with the 3-regular gluon cycle $g$ replaced by a ghost cycle $\tilde g$ (cf. Figure \ref{fig:ghostcycle}), of the same orientation. 
\end{rem}
The notation $\delta_+$ suggests that there is also a $\delta_-$. In fact,  such an operator can be defined and would replace a ghost loop by a gluon loop. We will not further study such an operator, since our interest lies in generating physical amplitudes from zero-ghost-loop amplitudes. 


\begin{ex}
Consider the following one-loop gluon self-energy graph:
$$
\Gamma = \parbox{10mm}{\includegraphics[scale=.1]{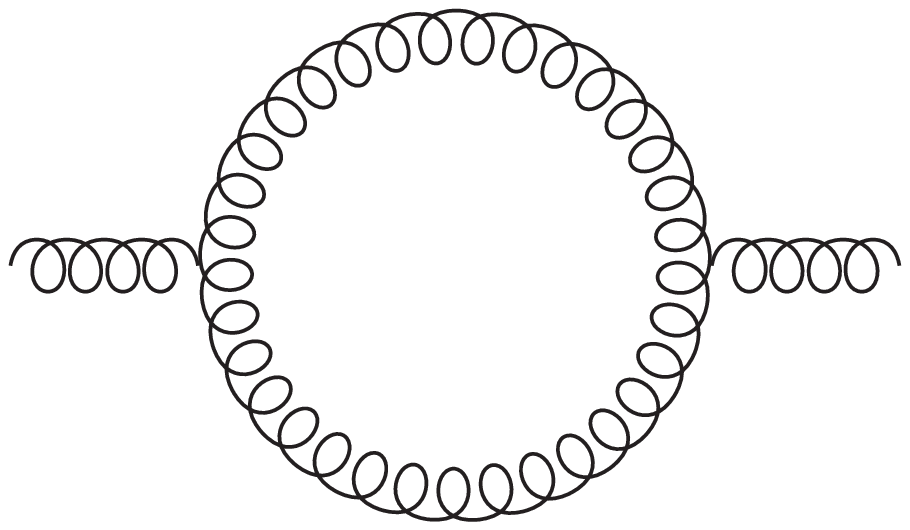}}.
$$
Its symmetry factor is $\Sym(\Gamma) = 2$ so that 
$$
\delta_+\left( \parbox{10mm}{\includegraphics[scale=.1]{gluon1l2e.eps}} \right) = 2\parbox{10mm}{\includegraphics[scale=.1]{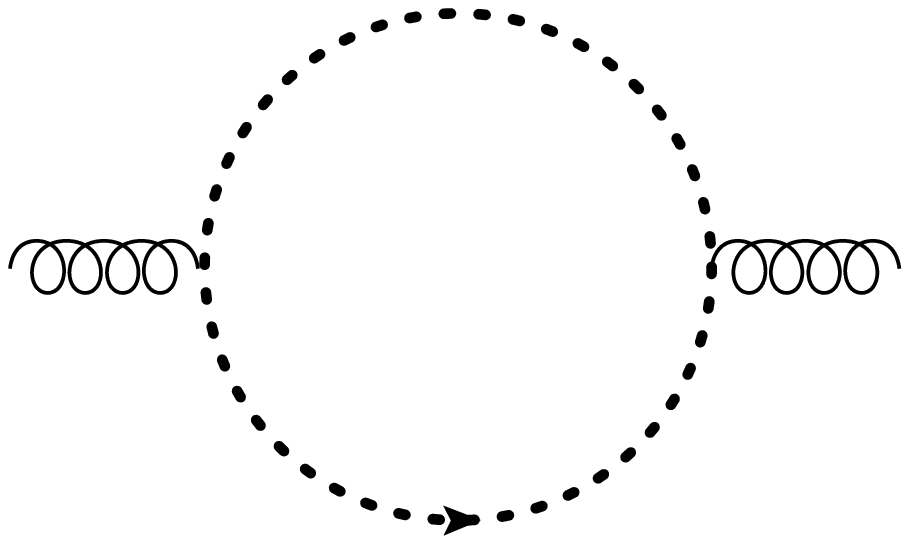}}.
$$
A two loop example is given by the graph 
$$
\Gamma' = \parbox{15mm}{\includegraphics[scale=.15]{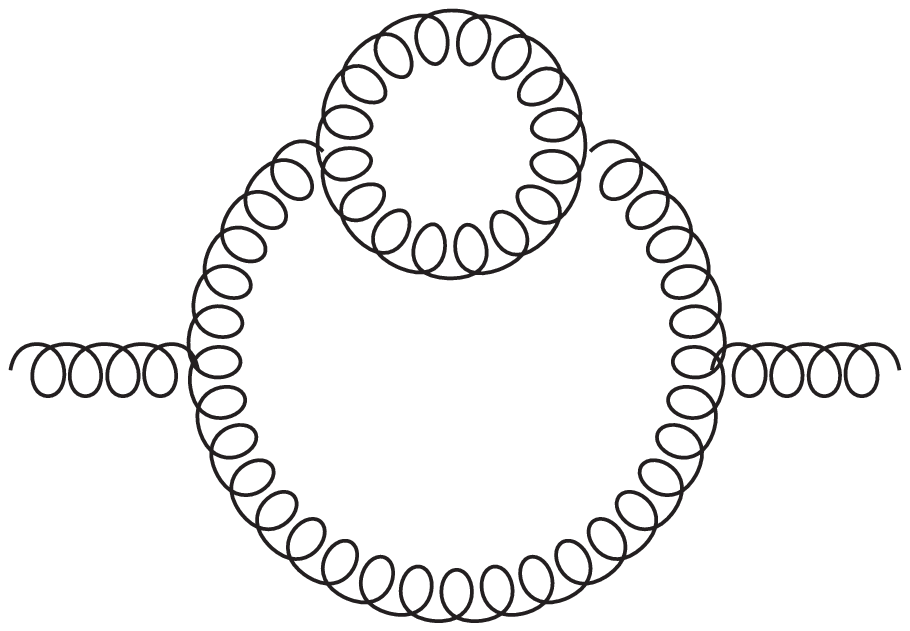}}
$$
for which $\Sym(\Gamma' ) = 2$. Now, 
$$
\delta_+\left( \parbox{15mm}{\includegraphics[scale=.15]{gluon2l2eA.eps}} \right) = 2\parbox{20mm}{\includegraphics[scale=.15]{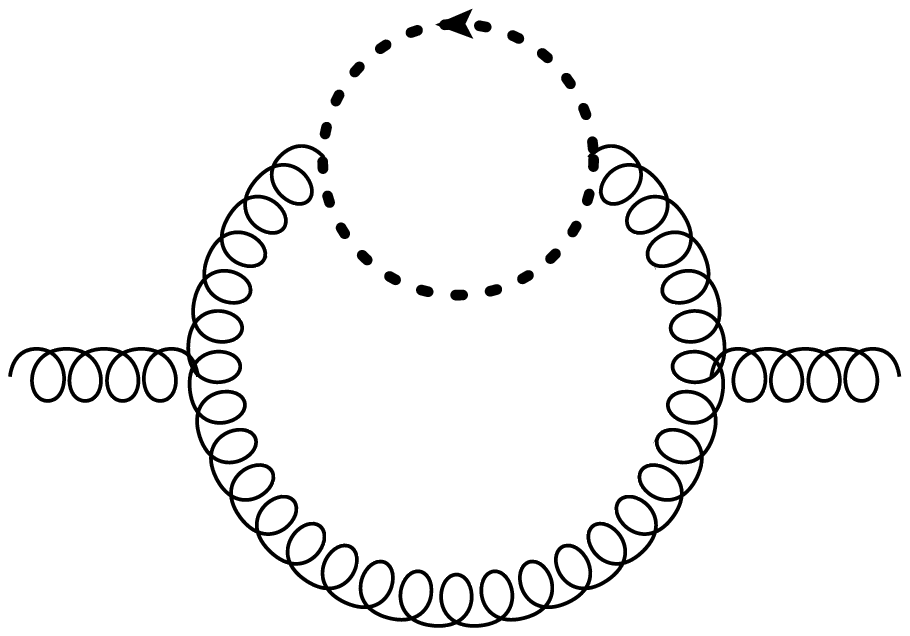}}+ 2  \parbox{20mm}{\includegraphics[scale=.15]{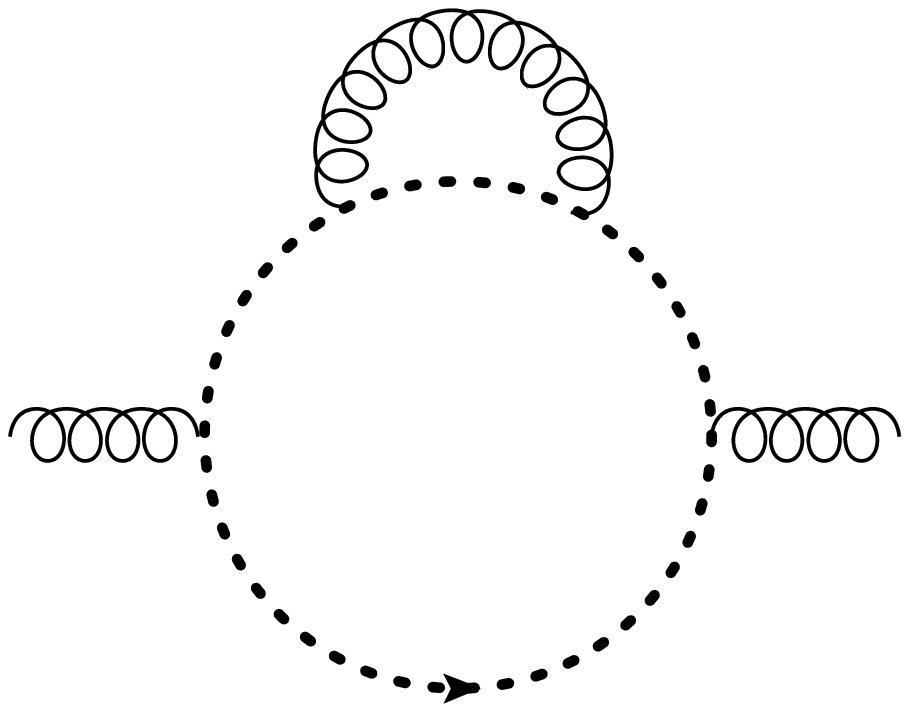}} + 2  \parbox{20mm}{\includegraphics[scale=.15]{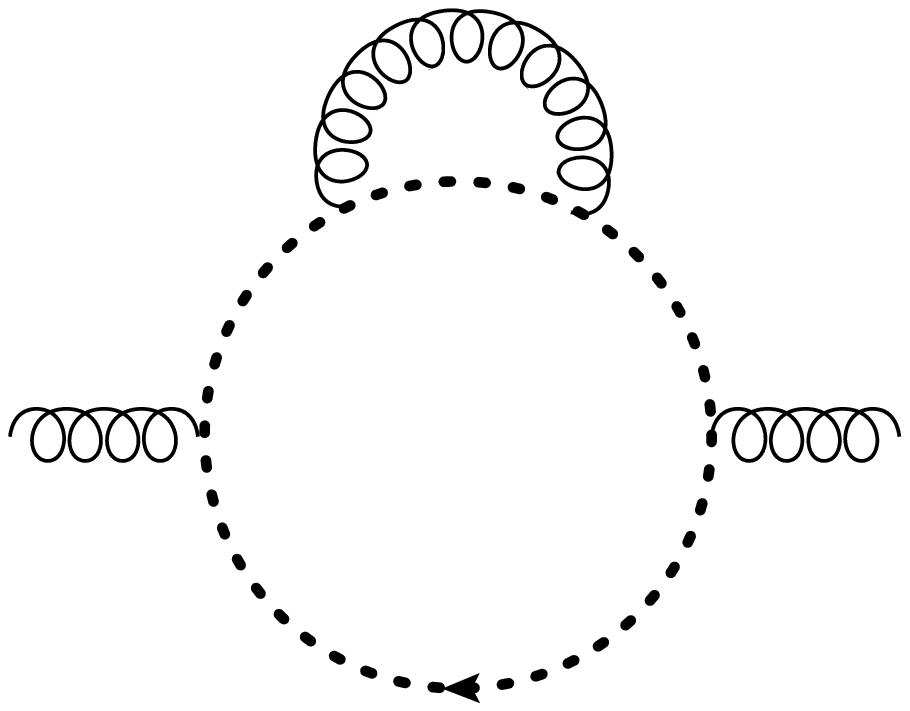}}.
$$
\end{ex}
The first graph on the rhs obtains a factor of two because the two orientation of the ghost loop
both reproduce this graph when the little gluon loop is replaced by a ghost loop.

The other two possible gluon cycles give the same graphs
with again a coefficient of two for each of them, with the two orientations of the ghost cycle
now resulting in those two remaining  graphs on the rhs. In full accordance with 
Lemma \ref{lem:skeletongraphs-delta},
the ratio of the symmetry factor of a graph on the left by the symmetry factor of a  graph on the right  counts such multiplicities. 

With that lemma we  conclude:
\begin{prop}
When acting on series of graphs with no ghost cycles ($\tilde n =0$):
$$
e^{\delta^+} \circ B_+^{k;n}
= \sum_{\begin{smallmatrix} \gamma ~\prim \\ |\gamma|=n, \tilde n(\gamma) =0 \\ E_E(\gamma)=k \end{smallmatrix}} \frac{1}{\Sym(\gamma)}B_+^{e^{\delta^+}(\gamma)} \circ e^{\delta_+}
$$
where the sum is over graphs $\gamma$ with no ghost cycles. 
\end{prop}
\begin{rem}
\label{prop:exp-chi}
There is a similar result for connected graphs on the exponentiation of $\chi_+$. We give it here without proof. It follows directly though from extending the definition 
of graph Hopf algebras and their Hochschild cohomology from 1PI to connected graphs.
When acting on series of graphs with no marked edges:
$$
e^{\chi^+} \circ B_+^{k;n}
= \sum_{\begin{smallmatrix} \gamma ~\prim \\ |\gamma|=n, j(\gamma) =0 \\ E_E(\gamma)=k \end{smallmatrix}} \frac{1}{\Sym(\gamma)}B_+^{e^{\chi^+}(\gamma)} \circ e^{\chi_+}
$$
where the sum is over graphs $\gamma$ with no marked edges and $j(\gamma)$ is the number of 4-valent vertices. 

Together, the two results on the interplay of Hochschild cohomology and exponentiation  show that gauge invariant combinatorial Green functions are obtained from gauge invariant skeleton graphs into which gauge invariant subgraphs are inserted.
\end{rem}

\begin{ex}
Let us consider the example of the gluon self-energy at two loops:
\begin{align}
\label{eq:Ge2l}
G_{n=2}^2 &= 
+ \frac16 \parbox{18mm}{\includegraphics[scale=.15]{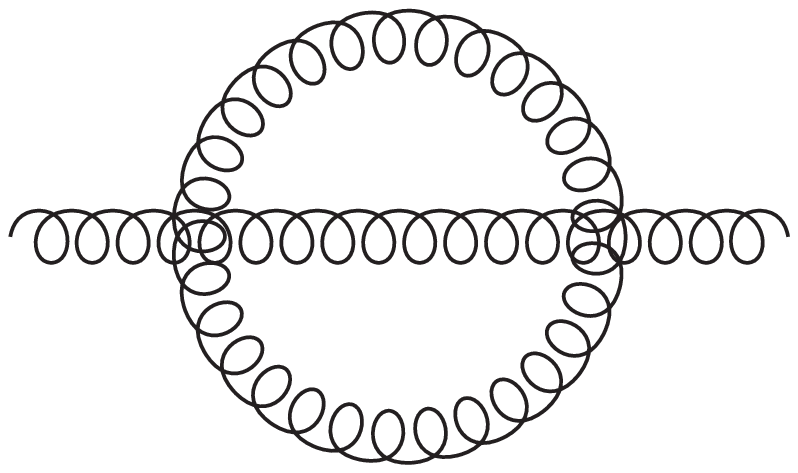}} 
 + \frac12 \parbox{20mm}{\includegraphics[scale=.15]{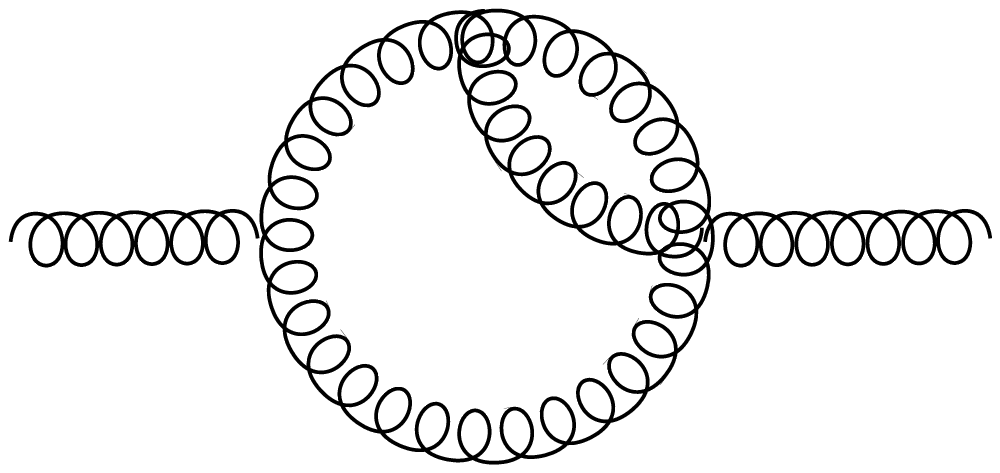}} 
+ \frac12 \parbox{20mm}{\includegraphics[scale=.15]{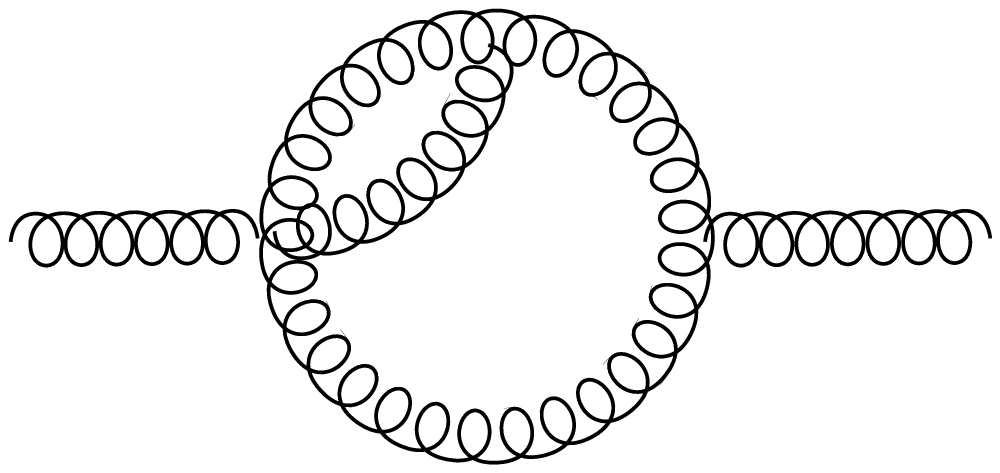}} \\
&\quad +
\frac12 \parbox{20mm}{\includegraphics[scale=.15]{gluon2l2eA.eps}} 
+ \parbox{20mm}{\includegraphics[scale=.15]{ghost2l2eA.eps}} 
+ \parbox{20mm}{\includegraphics[scale=.15]{ghost2l2eD.eps}}
+ \parbox{20mm}{\includegraphics[scale=.15]{ghost2l2eDR.eps}}
\nonumber \\
& \quad +  \frac 12 \parbox{20mm}{\includegraphics[scale=.15]{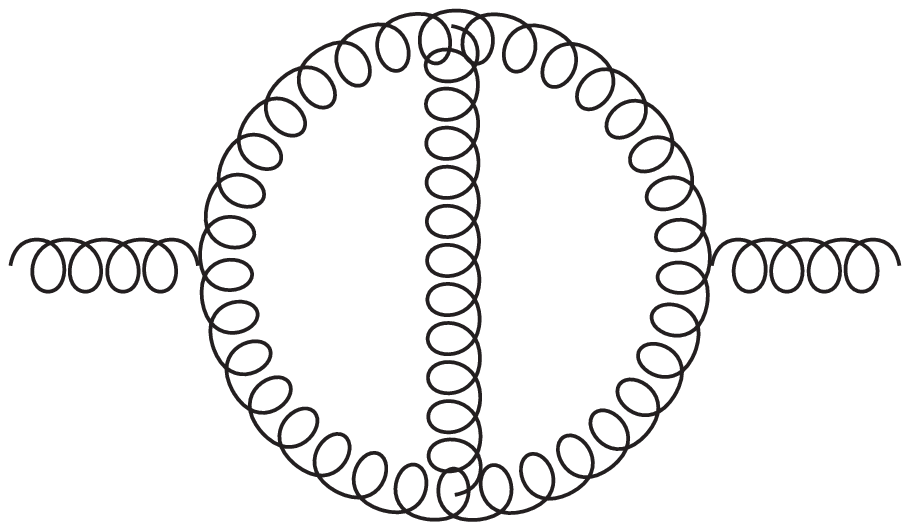}} 
+ \parbox{20mm}{\includegraphics[scale=.15]{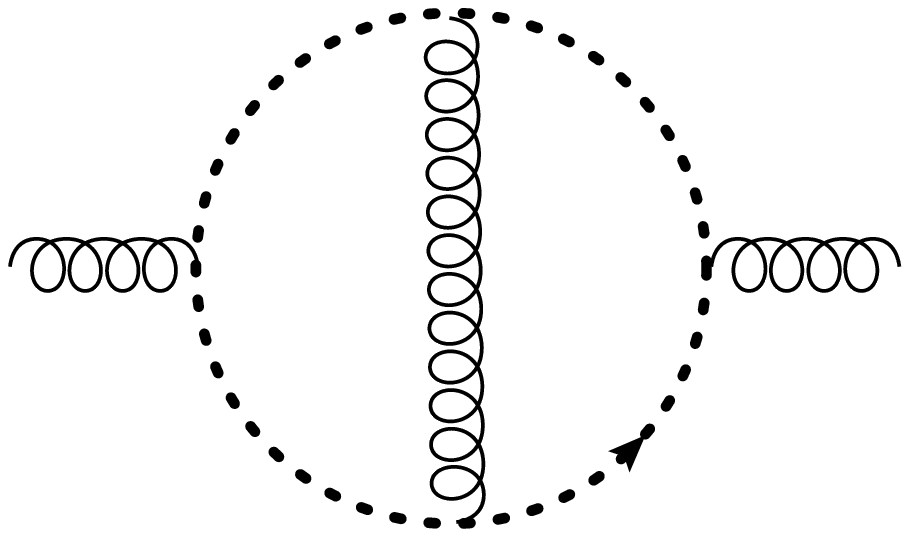}} 
+ \parbox{20mm}{\includegraphics[scale=.15]{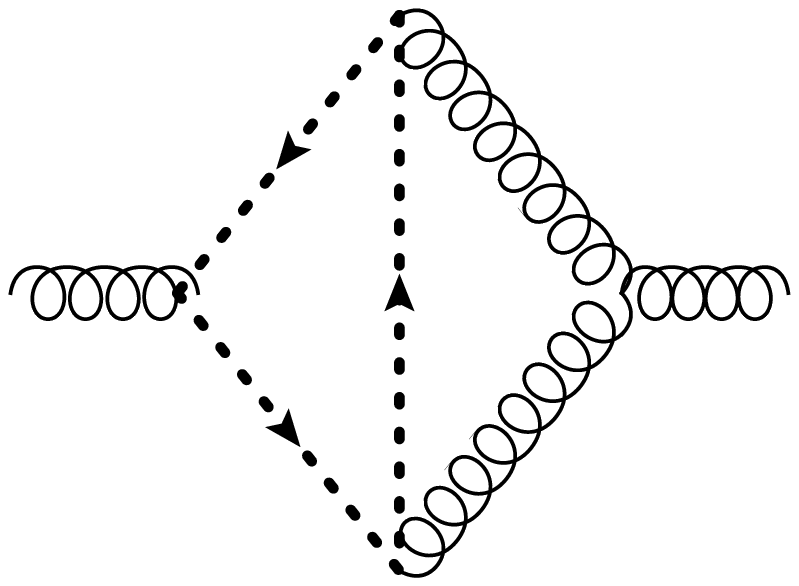}} 
+ \parbox{20mm}{\includegraphics[scale=.15]{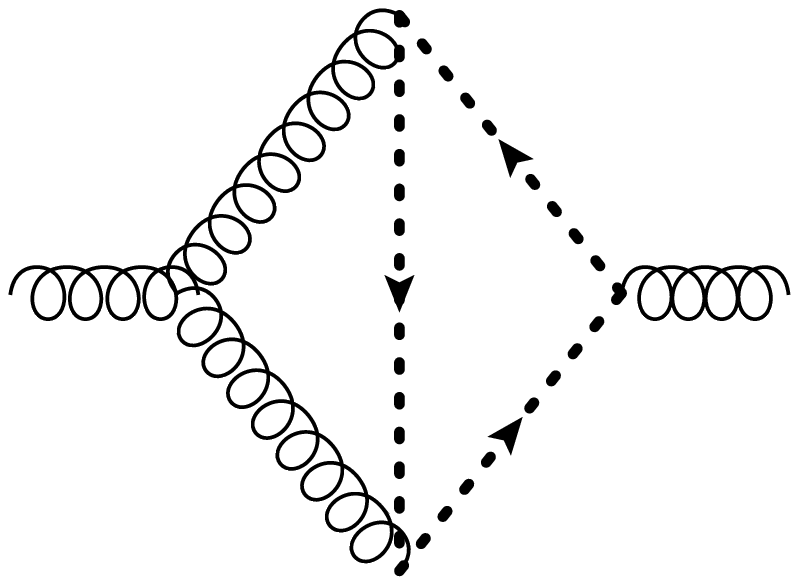}}\nonumber
\end{align}
whose zero-ghost-loop part is
\begin{align}
\label{eq:Ge2l0ghosts}
G_{n=2, \tilde n = 0}^2 = 
\frac16 \parbox{18mm}{\includegraphics[scale=.15]{tad2l2eB.eps}} 
+ \frac12 \parbox{20mm}{\includegraphics[scale=.15]{gluon2l2eE.eps}} 
+ \frac12 \parbox{20mm}{\includegraphics[scale=.15]{gluon2l2eF.eps}} \\ \nonumber
\qquad 
+\frac12 \parbox{20mm}{\includegraphics[scale=.15]{gluon2l2eA.eps}} 
+  \frac 12 \parbox{20mm}{\includegraphics[scale=.15]{gluon2l2eB.eps}} 
\end{align}
One readily checks that $(1+\delta_+) G_{n=2, \tilde n= 0}^2 = G_{n=2}^2$. 
\end{ex}

\begin{thm}
Let $\tilde H$ be the Hopf subalgebra of $H$ generated by $G^{k,n}$ for all $n \geq 0$ and $k=2,3,4$. Then $\exp{\delta_+}$ is an automorphisms of the graded Hopf algebra $\tilde H$:
$$
 \exp{\delta_+} (x_1 x_2) =  \exp{\delta_+} (x_1)\exp{\delta_+} (x_2); \qquad
\Delta( \exp{\delta_+} (x)) = (\exp{\delta_+} \otimes \exp{\delta_+}  ) \Delta (x).
$$
for $x_1,x_2,x \in \tilde H$.
\end{thm}
\proof
By definition, $\delta_+$ is an algebra derivation so that $\exp \delta_+$ is an algebra automorphism. Note that at a given loop order $l$, the exponential terminates at that power $n$ and is thus well-defined on the graded algebra underlying $\tilde H$.

Let us then consider the compatibility of $\delta_+$ with the coproduct structure.  
Recall from \cite{Sui08} the formula 
$$
\Delta(G^k) = \sum_{j_3,j_4,\tilde j \geq 0} G^k (Q^3)^{j_3} (Q^4)^{2 j_4} (Q^{1, \tilde 2})^{\tilde n} \otimes G^k_{j_3j_4\tilde j}. 
$$
which holds even without the Slavnov--Taylor identities. It continues to hold when restricting to graphs with zero ghost loops:
$$
\Delta(G^k_{\tilde n = 0}) = \sum_{j_3,j_4} G^k_{\tilde n = 0} (Q_{\tilde n = 0}^{3})^{j_3} (Q_{\tilde n = 0}^{4})^{2 j_4} \otimes G^r_{j_3,j_4;\tilde n = 0}. 
$$
We now apply $\exp \delta_+ \otimes \exp \delta_+$ to this equation to obtain after imposing the Slavnov--Taylor-identities $Q^3 = Q^4$:
\begin{align*}
(\exp \delta_+ \otimes \exp \delta_+ )\Delta(G^k_{\tilde n = 0}) &= \sum_{j_3,j_4} G^k(Q^{3})^{j_3} (Q^{4})^{2 j_4} \otimes \exp \delta_+ \left(G^k_{j_3j_4;\tilde n = 0}\right)\\
&=\sum_{n \geq 0} G^k Q^{2n} \otimes \exp \delta_+ \left(G^k_{n,\tilde n = 0} \right)
\end{align*}
since in the absence of ghost vertices $j_3 + 2 j_4 = 2n$ in terms of the first Betti number $n$.  Lemma(\ref{lem:skeletongraphs-delta}) then yields $\exp \delta_+ (G^k_{n,\tilde n = 0}) = G^k_n$, which completes the proof.
\endproof

This can be extended to the {\it connected} Green's functions $X_{k,n}$, where also a similar result can  be shown for $\exp \chi_+$. 

\begin{ex}
First, recall the Slavnov--Taylor identities $G^3 G^{\tilde 2} = G^{1,\tilde 2} G^{2}$ which at one-loop order become:
\begin{multline*}
 \parbox{20mm}{\includegraphics[scale=.15]{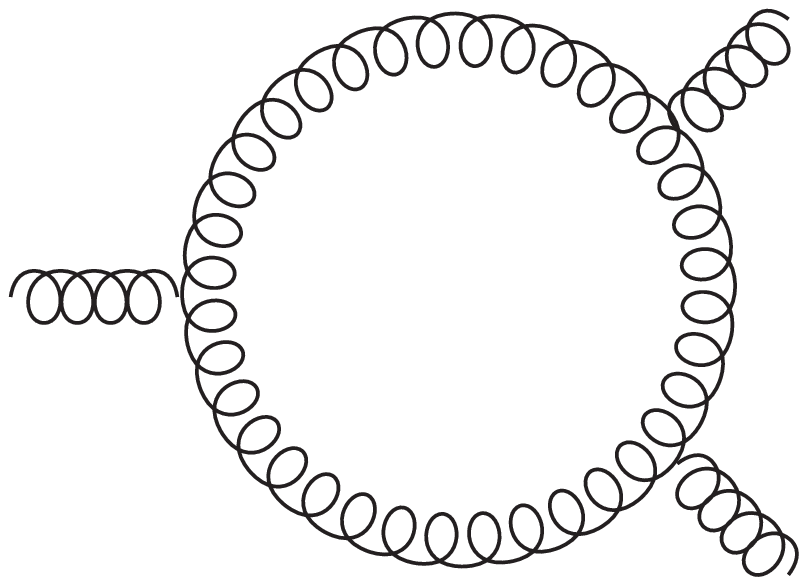}}
+ \frac12 \parbox{20mm}{\includegraphics[scale=.15]{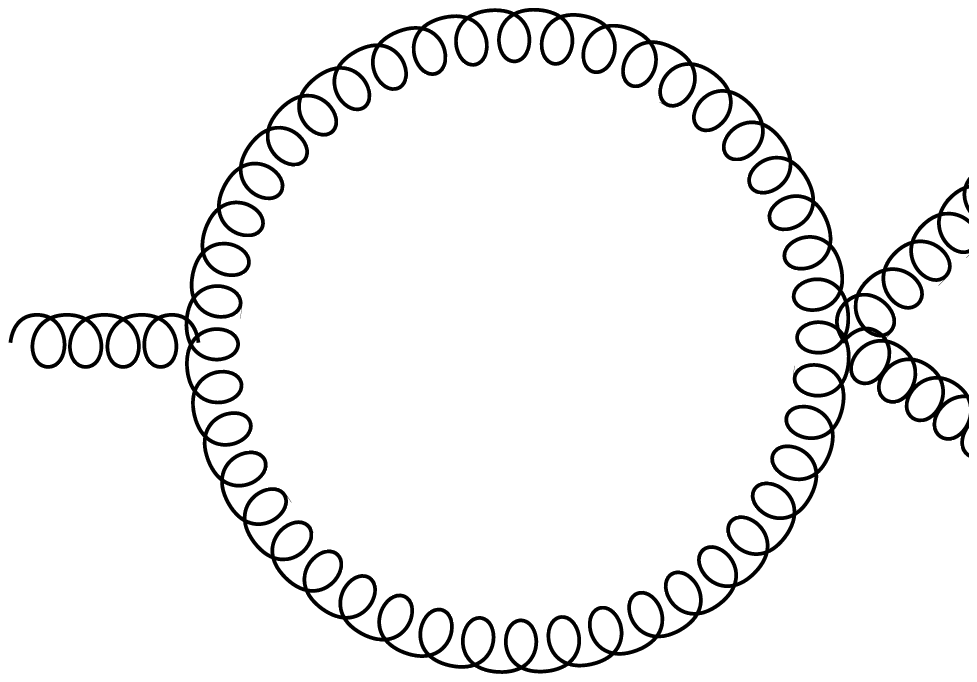}}
+ \parbox{15mm}{\includegraphics[scale=.15]{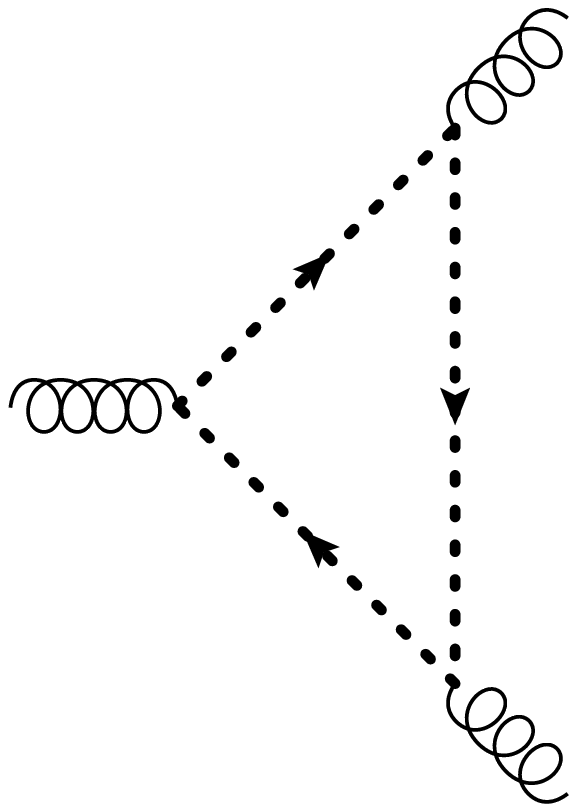}}
+ \parbox{15mm}{\includegraphics[scale=.15]{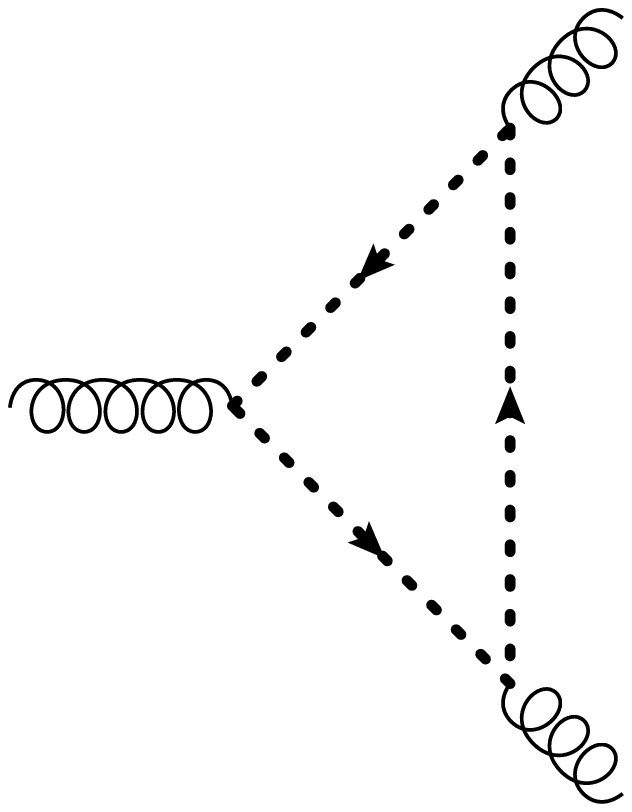}}
-\parbox{15mm}{\includegraphics[scale=.15]{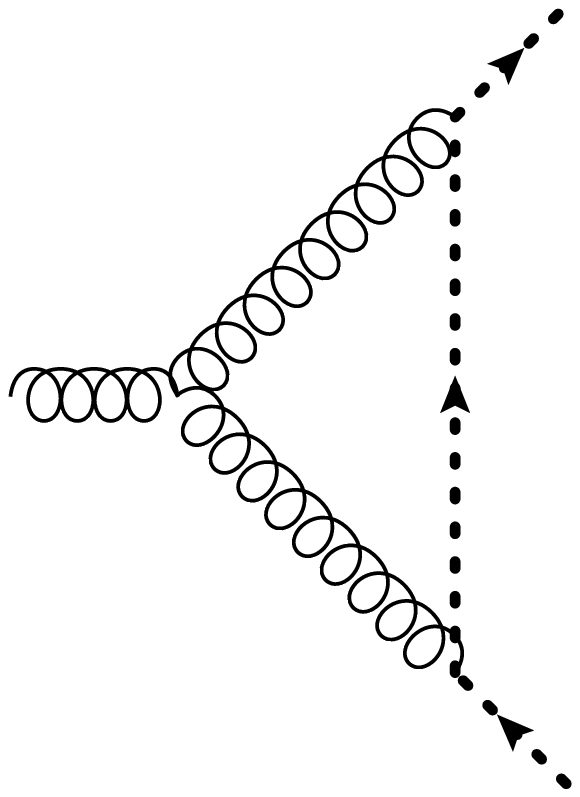}}
-\parbox{15mm}{\includegraphics[scale=.15]{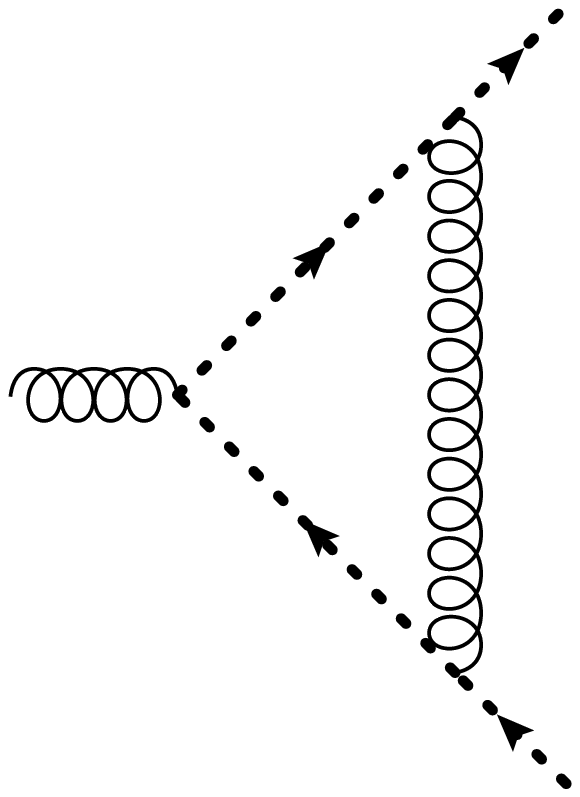}} 
- \parbox{20mm}{\includegraphics[scale=.15]{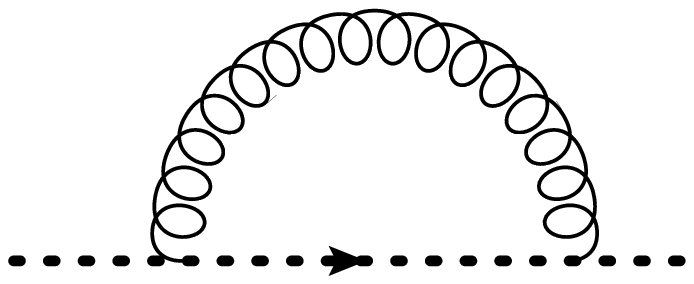}} \\ 
+ \frac12 \parbox{20mm}{\includegraphics[scale=.15]{gluon1l2e.eps}} 
+ \parbox{20mm}{\includegraphics[scale=.15]{ghost1l2e.eps}}
= 0.
\end{multline*}
We compute $\Delta'(G^{2,n=2}_{\tilde n=0})$ with $G^{2,n=2}_{\tilde n=0}$ given in Eq.\eqref{eq:Ge2l0ghosts}. For the first graph on the last line, we have
$$
\Delta'\left( \frac12 \parbox{20mm}{\includegraphics[scale=.15]{gluon2l2eA.eps}}\right) 
=  \frac12 \parbox{20mm}{\includegraphics[scale=.15]{gluon1l2e.eps}} \otimes \parbox{20mm}{\includegraphics[scale=.15]{gluon1l2e.eps}}
$$
If we apply $\exp \delta_+ \otimes \exp \delta_+$ to this expression, we obtain 
$$
\frac12 \left( \parbox{20mm}{\includegraphics[scale=.15]{gluon1l2e.eps}} + 2 \parbox{20mm}{\includegraphics[scale=.15]{ghost1l2e.eps}} \right)\otimes \left(\parbox{20mm}{\includegraphics[scale=.15]{gluon1l2e.eps}} + 2 \parbox{20mm}{\includegraphics[scale=.15]{ghost1l2e.eps}}\right)
$$
For the coproduct on the last graph in Eq.\eqref{eq:Ge2l0ghosts} we have
$$
\Delta' \left(\frac 12 \parbox{20mm}{\includegraphics[scale=.15]{gluon2l2eB.eps}} \right)
= \parbox{20mm}{\includegraphics[scale=.15]{gluon1l3e.eps}} \otimes \parbox{20mm}{\includegraphics[scale=.15]{gluon1l2e.eps}}
$$
and applying $\exp \delta_+ \otimes \exp \delta_+$ to this expression yields
$$
\left( \parbox{20mm}{\includegraphics[scale=.15]{gluon1l3e.eps}}+\parbox{15mm}{\includegraphics[scale=.15]{ghost1l3e.eps}}+\parbox{15mm}{\includegraphics[scale=.15]{ghost1l3eR.eps}} \right)\otimes \left(\parbox{20mm}{\includegraphics[scale=.15]{gluon1l2e.eps}} + 2 \parbox{20mm}{\includegraphics[scale=.15]{ghost1l2e.eps}}\right)
$$
On the other hand, $\Delta' (\exp \delta_+ G^{2,n=2}_{\tilde n=0}) = \Delta'(G^{2,n=2})$ is computed from Eq.\eqref{eq:Ge2l}:
\begin{multline*}
\Delta'(G^{2,n=2}) =\left( \frac 12  \parbox{20mm}{\includegraphics[scale=.15]{gluon1l2e.eps}}+\parbox{20mm}{\includegraphics[scale=.15]{ghost1l2e.eps}} \right)\otimes \parbox{20mm}{\includegraphics[scale=.15]{gluon1l2e.eps}}\\
+ 2 \parbox{15mm}{\includegraphics[scale=.15]{gluon1l2ge.eps}} \otimes \parbox{20mm}{\includegraphics[scale=.15]{ghost1l2e.eps}}\\
+ \left(   \parbox{20mm}{\includegraphics[scale=.15]{gluon1l3e.eps}}+ \parbox{15mm}{\includegraphics[scale=.15]{ghost1l3e.eps}} + \parbox{15mm}{\includegraphics[scale=.15]{ghost1l3eR.eps}} \right)\otimes \parbox{20mm}{\includegraphics[scale=.15]{gluon1l2e.eps}}\\
 +\left( 2  \parbox{12mm}{\includegraphics[scale=.15]{gluon1l1e2ge.eps}}+ 2 \parbox{15mm}{\includegraphics[scale=.15]{ghost1l1e2ge.eps}} \right)\otimes \parbox{20mm}{\includegraphics[scale=.15]{ghost1l2e.eps}}
\end{multline*}
We conclude that
\begin{multline*}
(\exp \delta_+ \otimes \exp \delta_+)\Delta'(G^{2,n=2}_{\tilde n=0})- \Delta' (\exp \delta_+ (G^{2,n=2}_{\tilde n=0})) = \\
2 \left(
\frac12 \parbox{20mm}{\includegraphics[scale=.15]{gluon1l2e.eps}} 
+\parbox{20mm}{\includegraphics[scale=.15]{ghost1l2e.eps}} 
+  \parbox{15mm}{\includegraphics[scale=.15]{ghost1l3e.eps}} 
+  \parbox{15mm}{\includegraphics[scale=.15]{ghost1l3eR.eps}} \right.\\
\left.-  \parbox{15mm}{\includegraphics[scale=.15]{gluon1l2ge.eps}} 
+  \parbox{20mm}{\includegraphics[scale=.15]{gluon1l3e.eps}}    
- \parbox{15mm}{\includegraphics[scale=.15]{gluon1l1e2ge.eps}}
- \parbox{15mm}{\includegraphics[scale=.15]{ghost1l1e2ge.eps}} 
 \right) \otimes \parbox{20mm}{\includegraphics[scale=.15]{ghost1l2e.eps}}
\end{multline*}
which vanishes by the Slavnov--Taylor identities upon adding the contribution of 4-valent vertices.
\end{ex}

\section{The corolla polynomial and differentials}\label{Corolla}
\subsection{The Corolla Polynomial}
Finally, we introduce the Corolla Polynomial (\cite{KrY}). It is a polynomial based on half-edge variables
$a_{v,j}$ assigned to any half-edge $(v,j)$ determined by a vertex $v$ and an edge $j$.
We need the following definitions:
\begin{itemize}
\item For a vertex $v\in V$ let $n(v)$ be the set of edges incident to $v$ (internal or external).
\item For a vertex $v\in V$ let $D_v = \sum_{j \in n(v)} a_{v,j}$.
\item Let $\mathcal{C}$ be the set of all cycles of $\Gamma$ (cycles, not circuits). This is a finite set.
\item For $C$ a cycle and $v$ a vertex in $V$, since $\Gamma$ is 3-regular, there is a unique edge of $\Gamma$ incident to $v$ and not in $C$, let $v_C$ be this edge.
\item For $i\geq 0$ let 
  \[
  C^i = \sum_{\substack{C_1,C_2,\ldots C_i \in \mathcal{C} \\ C_j \text{pairwise disjoint}}} \left(\left(\prod_{j=1}^{i} \prod_{v \in C_j}a_{v,v_C}\right)\prod_{v \not\in C_1\cup C_2\cup \cdots \cup C_i} D_v\right)
  \]
\item Let
  \[
   C = \sum_{j \geq 0} (-1)^j C^j
  \]
\end{itemize}

For any finite graph $\Gamma$, this is a polynomial $C=C(\Gamma)$ ---the corolla polynomial--- because $C^i=0$ for $i>|\mathcal{C}|$.

\begin{thm}(\cite{KrY})
\label{bin reformulation}
  Let $\mathcal{T}^\Gamma$ be the set of sets $T$ of half edges of $\Gamma$ with the property that
  \begin{itemize}
    \item every vertex of $\Gamma$ is incident to exactly one half edge of $T$
    \item $\Gamma\smallsetminus T$ has no cycles
  \end{itemize}
  Then
  \[
  C(\Gamma) = \sum_{T\in \mathcal{T}^\Gamma} \prod_{h\in T} a_h
  \]
\end{thm}
\begin{rem}
This shows that the corolla polynomial is strictly positive.
As it applies in this form as a corolla differential to pure Yang--Mills theory, this results
in a positivity statement on Yang--Mills theory which does not hold for gauge fields coupled to matter fields.
Accordingly, the sign of the $\beta$-function in gauge theory becomes dependent on the number of fermion families, and their representations.
\end{rem}
\begin{rem}\label{remthree}
For a graph $\Gamma$, let $E$ be a set of 
pairwise disjoint internal edges of $\Gamma$.
For $i\geq 0$ let 
  \[
  C^i_E(\Gamma) = \sum_{\substack{C_1,C_2,\ldots C_i \in \mathcal{C} \\ C_j \text{ pairwise disjoint} \\ C_j\cap E=\emptyset}} \left(\left(\prod_{j=1}^{i} \prod_{v \in C_j}a_{v,v_C}\right)\prod_{v \not\in C_1\cup C_2\cup \cdots \cup C_i\cup E} 
  D_v\right)
  \]
  where the sum forbids cycles from sharing either vertices or edges with $E$.
\item Let
  \[
   C_E(\Gamma) = \sum_{j \geq 0} (-1)^j C^j_E(\Gamma).
  \]
Then,
\[
C_E(\Gamma)=C(\Gamma- E)
\]  
where $\Gamma-E$ is the graph with the edges and vertices involved in $E$ removed.  Removing a vertex removes all its incident half-edges so that $2|E|$ new external edges are generated. 
Note that $C_\emptyset(\Gamma)=C(\Gamma)$.
\end{rem}

Define
\begin{equation*}
C^{\mathrm{fr}}(\Gamma):=\sum_E \left(C_E(\Gamma)\prod_{e\in E}W_e\right),
\end{equation*}
where $W_e$ is defined in (\ref{eq:W}).
\begin{cor}
\[
C^{\mathrm{fr}}(\Gamma)=\sum_E \left( \sum_{T\in\mathcal{T}^E}(\prod_{h\in T} a_h) \prod_{e\in E}W_e\right).\]
\end{cor}
\begin{proof}
Immediate.
\end{proof}
\begin{rem}
Consider a 3-regular graph $\Gamma$ which has $j$, $j\geq 2$, 3-valent vertices, and let $\mathcal{P}$ be a set of $m$ paths, $2m\leq j$, on internal edges and 3-valent vertices in $\Gamma$ which each connect 
two external 3-valent vertices (a 3-valent vertex $v$ is external if $n(v)$ contains an external edge) with $p_i\cap p_j=\emptyset$, $\forall p_i,p_j\in \mathcal{P}$.

Consider for chosen set $E$ and $\mathcal{P}$ as above, with $E\cap\mathcal{P}=\emptyset$,
for $i\geq 0$, 
  \begin{eqnarray*}
  C^i_{E,\mathcal{P}}(\Gamma) & = &  
  \sum_{\substack{C_1,C_2,\ldots C_i \in \mathcal{C} \\ C_j \text{ pairwise disjoint} \\ C_j\cap E=\emptyset\\
  C_j\cap\mathcal{P}=\emptyset}} \left(\left(\prod_{j=1}^{i} \prod_{v \in C_j}a_{v,v_C}\right)\prod_{v \not\in C_1\cup C_2\cup \cdots \cup C_i\cup E\cup\mathcal{P}} 
  D_v\right)\times\\ & & \times \left( \prod_{p\in\mathcal{P}}\prod_{v\in p}a_{v,v_p}\right),
  \end{eqnarray*}
  where the sum forbids cycles from sharing either vertices or edges with $E$, and $v_p$ is the unique half-edge at $v$ not in $p$.
\item Let
  \[
   C_{E,\mathcal{P}}(\Gamma) = \sum_{j \geq 0} (-1)^j C^j_{E,\mathcal{P}}(\Gamma).
  \]
Finally, we set
\begin{equation}\label{openlines}
C^{\mathrm{fr}}_{\mathcal{P}}(\Gamma):=\sum_E \left(C_{E,\mathcal{P}}(\Gamma)\prod_{e\in E}W_e\right).
\end{equation}  
\end{rem}
\subsection{Corolla differentials}
Our main use of the corolla polynomial is to construct differential operators with it. These operators differentiate with respect to momenta $\xi_e$ assigned to edges $e$ of a graph, and act on the second Kirchhoff polynomial written for generic edge momenta $\xi_e$, that is on $|N|_{\mathrm{Pf}}$.

Only at the end of the computation will we employ the map \[
Q: \xi_e\to \xi_e+q(e).\]
We then set  $\xi_e=0$ after we have applied the corolla differentials 
so that we obtain the standard second Symanzik polynomial for specific external momenta as prescribed by gauge theory amplitudes.

For a half edge $h\equiv (w,f)\in H^\Gamma$, we let $e(h)=f$ and $v(h)=w$.  We remind the reader that $h_+$ and $h_-$ are the successor and the precursor of $h$  in the oriented corolla at $v(h)$, and that we assign to a graph $\Gamma$:
\begin{enumerate}
\item to each (possibly external) edge $e$, a variable $A_e$ and a 4-vector $\xi_e$;\\
\item to each half edge $h$, a Lorentz  index $\mu(h)$;\\
\item a factor $\mathrm{colour}(\Gamma)$.
\end{enumerate}

\subsection{The differential $D^0$}
The corolla polynomial is an alternating sum over terms $C^i$, where $i$ counts the number of loops.
Similarly, the corolla differentials are a sum of terms $D^i$. We start with $D^0$.

Let $\Gamma\in\mathcal G^{n,l}$:
\begin{samepage}
\begin{equation*} U^0(\Gamma) = \int_E\frac{\mathrm d\underline k_L}{(2\pi)^{dl}} C_\Gamma^0(\underline D) e^{ -\sum_{e\in\Gamma^{[1]}} A_e\xi_e^{\prime2} },  \end{equation*}
\end{samepage}where
\begin{equation*} C_\Gamma^0(\underline D) = \prod_{v\in\Gamma^{[0]}} D_v, \qquad D_v = D_{v1} + D_{v2} + D_{v3} \end{equation*}
(the edges incident on $v$ are labelled $1,2,3$), 
\begin{equation*} D_{v1} = -\tfrac12 g^{\mu_2\mu_3} \Big( \epsilon_{v2}\frac1{A_2}\frac\partial{\partial\xi_{2\mu_1}}-\epsilon_{v3}\frac1{A_3}\frac\partial{\partial\xi_{3\mu_1}} \Big). \end{equation*}
Using that all corollas are oriented, we can write this as
\begin{equation*} D_{g}(h) := -\tfrac12 g^{\mu_{h_+}\mu_{h_-}} \Big( \epsilon_{h_+}\frac1{A_{e(h_+)}}\frac\partial{\partial\xi(h_+)_{\mu_h}}-\epsilon_{h_-}\frac1{A_{e(h_-)}}\frac\partial{\partial\xi(h_-)_{\mu_h}} \Big), \end{equation*}
for any half-edge $h$.
The operator $D_v$ is such that if it acts on $e^{-\sum_{e\in\Gamma^{[1]}} A_e\xi_e^{\prime2} }$, it gives the 3-vertex Feynman rule of Eq.\eqref{3-vertex}:
\begin{equation*} D_v e^{-\sum_{e\in\Gamma^{[1]}} A_e\xi_e^{\prime2}} = V^{(3)}_v e^{-\sum_{e\in\Gamma^{[1]}} A_e\xi_e^{\prime2}}. \end{equation*}
In order to calculate $C_\Gamma^0(\underline D)\equiv C_\Gamma^0(h\to D_g(h))$, we also need to know the Leibniz terms $D_vV^{(3)}_w$, where $v,w\in\Gamma^{[0]}$.
\begin{itemize}
\item If $v$ and $w$ do not share an edge, $D_vV^{(3)}_w=0$.

\item Suppose they share exactly one edge; we give it label $5$. Let $1$ and $2$ be the other edges at $v$ and $3$ and $4$ the other ones at $w$:
\begin{equation*} \parbox[c][35pt]{35pt}{\centering\begin{fmfgraph*}(30,30)
	\fmfleft{2,1}
	\fmfright{3,4}
	\fmfv{label=\scriptsize$1$,label.dist=1}{1}
	\fmfv{label=\scriptsize$2$,label.dist=1}{2}
	\fmfv{label=\scriptsize$3$,label.dist=1}{3}
	\fmfv{label=\scriptsize$4$,label.dist=1}{4}
	\fmfv{label=\scriptsize$v$,label.dist=3}{v}
	\fmfv{label=\scriptsize$w$,label.dist=3}{w}
	\fmf{boson,tension=2}{1,v}
	\fmf{boson,tension=2}{2,v}
	\fmf{boson,tension=2}{3,w}
	\fmf{boson,tension=2}{4,w}
	\fmf{boson,label=\scriptsize$5$,label.dist=3}{v,w}
\end{fmfgraph*}}
.\end{equation*}
Then:
\begin{equation*}
D_vV^{(3)}_w
= \frac1{A_5} (g^{\mu_4\mu_2}g^{\mu_3\mu_1}-g^{\mu_2\mu_3}g^{\mu_4\mu_1})\equiv  \frac{W_e}{A_e}, \end{equation*}
where $W_e$ is the Feynman rule for a marked edge (equation (\ref{eq:W})). Note that thus
\begin{equation*} D_vV^{(3)}_w = D_wV^{(3)}_v . \end{equation*}

\item Suppose that $v$ and $w$ share two edges, $3$ and $4$. Let $1$ be the other edge at $v$ and $2$ the other one at $w$:
\begin{equation*} \parbox[c][35pt]{35pt}{\centering\begin{fmfgraph*}(30,15)
	\fmfleft 1
	\fmfright 2
	\fmfv{label=\scriptsize$1$,label.dist=1}{1}
	\fmfv{label=\scriptsize$2$,label.dist=1}{2}
	\fmfv{label=\scriptsize$v$,label.dist=-5}{v}
	\fmfv{label=\scriptsize$w$,label.dist=-6}{w}
	\fmf{boson,tension=2}{1,v}
	\fmf{boson,tension=2}{2,w}
	\fmf{boson,right,label=\scriptsize$3$,label.dist=3}{v,w}
	\fmffreeze
	\fmf{boson,left,label=\scriptsize$4$,label.dist=3}{v,w}
\end{fmfgraph*}}
.\end{equation*}
Then:
\begin{equation*}\begin{split}
D_vV^{(3)}_w
&= \Big( \frac1{A_3} ( g^{\mu_4\mu_4}g^{\mu_2\mu_1} - g^{\mu_2\mu_4}g^{\mu_4\mu_1} )+ \frac1{A_4} ( g^{\mu_3\mu_3}g^{\mu_2\mu_1} - g^{\mu_3\mu_2}g^{\mu_3\mu_1} ) \Big) \\
&= \frac{W_3}{A_3} + \frac{W_4}{A_4}.
\end{split}\end{equation*}
where we have used equation (\ref{eq:W}). Note that also in this case
\begin{equation*} D_vV^{(3)}_w = D_wV^{(3)}_v . \end{equation*}
Contracting the indices further gives self-loops which can be omitted:
\begin{equation*} D_vV^{(3)}_w = 3C_2\delta^{a_1a_2} g^{\mu_1\mu_2} \Big(\frac1{A_3}+\frac1{A_4}\Big). \end{equation*}
\end{itemize}

\subsection{Regular terms and residues}
We can now compute immediately the application of $D^0$ to the scalar integrand $I_\Gamma$, that is, $U^0(\Gamma):=D^0I_\Gamma$.
\begin{equation*}\begin{split}
U^0(\Gamma)
&= \int\frac{\mathrm d\underline k_L}{(2\pi)^{dl}} \bigg(\prod_{v\in\Gamma^{[0]}}D_v\bigg) e^{ -\sum_{e\in\Gamma^{[1]}} A_e\xi_e^{\prime2} } \\
&= \int\frac{\mathrm d\underline k_L}{(2\pi)^{dl}} \bigg[ \bigg(\prod_{v\in\Gamma^{[0]}}V^{(3)}_v\bigg) \\
&\qquad+ \sum_{\substack{w,w'\in\Gamma^{[0]}\\ \text{$w$ and $w'$ share an edge}\\ w<w'}}
(D_wV^{(3)}_{w'})
\bigg(\prod_{\substack{v\in\Gamma^{[0]}\\ v\neq w,w'}} V^{(3)}_v\bigg) \\
&\qquad+ \sum_{\substack{w,w',x,x'\in\Gamma^{[0]}\\ \text{$w$ and $w'$ and $x$ and $x'$ share an edge}\\ w<w'<x<x' }}
(D_wV^{(3)}_{w'})(D_xV^{(3)}_{x'})
\bigg(\prod_{\substack{v\in\Gamma^{[0]}\\ v\neq w,w',x,x'}} V^{(3)}_v\bigg) \\
&\qquad+ \cdots \bigg] e^{ -\sum_{e\in\Gamma^{[1]}} A_e\xi_e^{\prime2} }.
\end{split}\end{equation*}
With the result of the previous subsection we get (recall that we exclude graphs with self-loops):
\begin{equation*}\begin{split}
U^0(\Gamma) &= \int\frac{\mathrm d\underline k_L}{(2\pi)^{dl}} \bigg[ \bigg(\prod_{v\in\Gamma^{[0]}}V^{(3)}_v\bigg)
+ \sum_{e\in\Gamma^{[1]}_{\rm int}}
\frac{W_e}{A_e}
\bigg(\prod_{\substack{v\in\Gamma^{[0]}\\ v \text{ not adj.\ to } e }} V^{(3)}_v\bigg) \\
&\qquad+ \sum_{\substack{\{e_1,e_2\}\subset\Gamma^{[1]}_{\rm int}\\ \text{$e_1$ and $e_2$ do not share a vertex}}}
\frac{W_{e_1}W_{e_2}}{A_{e_1}A_{e_2}}
\bigg(\prod_{\substack{v\in\Gamma^{[0]}\\ v \text{ not adj.\ to } e_1,e_2 }} V^{(3)}_v\bigg) \\
&\qquad+ \cdots \bigg] e^{ -\sum_{e'\in\Gamma^{[1]}} A_{e'}\xi_{e'}^{\prime2} }. \\
&= \sum_{k\geq0} \sum_{\substack{\{e_1,\ldots,e_k\}\subset\Gamma^{[1]}_{\rm int}\\ \text{$e_1,\ldots,e_k$ do not share a vertex}}} \frac{W_{e_1}\cdots W_{e_k}}{A_{e_1}\cdots A_{e_k}}
\int\frac{\mathrm d\underline k_L}{(2\pi)^{dl}} \\
&\qquad\times \bigg(\prod_{\substack{v\in\Gamma^{[0]}\\ v \text{ not adj.\ to } e_1,\ldots,e_k }} V^{(3)}_v\bigg) e^{ -\sum_{e\in\Gamma^{[1]}} A_{e}\xi_{e}^{\prime 2}} .\\
\end{split}\end{equation*}

The first term we recognise as the Feynman-Schwinger integrand of $\Gamma$. The other terms we can write as the integrands of marked versions of $\Gamma$ (equation (\ref{markedintegrand})). More precisely,
\begin{equation*}
U^0(\Gamma)
= \sum_{k\geq0} \sum_{\{e_1,\ldots,e_k\}\subset\Gamma^{[1]}_{\rm int}} \frac1{A_{e_1}\cdots A_{e_k}} \int\frac{\mathrm d\underline k_L}{(2\pi)^{dl}} \mathcal I(\chi_+^{e_1}\cdots\chi_+^{e_k}\Gamma)
e^{-A_{e_1}\xi_{e_1}^{\prime2}-\cdots-A_{e_k}\xi_{e_k}^{\prime2}},
\end{equation*}
where $\mathcal I(\Gamma)$ is given in equation \eqref{eq:mathcalI}. Recall that in the exponent in the integrand only the unmarked edges are included. 
That is why the factor $e^{-A_{e_1}\xi_{e_1}^{\prime2}-\cdots-A_{e_k}\xi_{e_k}^{\prime2}}$ appears. This factor does not change the residue along $\prod_{e\in\Gamma^{[1]}_{\rm int}}A_e=0$.

Each subset of edges here is accompanied by a corresponding set of poles. By construction, the residues along these poles correspond to integrands where the edges shrink to 
form 4-valent vertices with the correct Feynman rules.

Using the $\chi_+^e$-operator, we can write the integral $\tilde U^0(\Gamma)$ as:
\begin{equation*}\begin{split}
\tilde U^0(\Gamma)
&= \sum_{k\geq0} \sum_{\{e_1,\ldots,e_k\}\subset\Gamma^{[1]}_{\rm int}}
\int\mathrm d\underline A_{\Gamma^{[1]}\smallsetminus\{e_1,\ldots,e_k\}}
I(\chi_+^{e_1}\cdots\chi_+^{e_k}\Gamma) .
\end{split}\end{equation*}
In terms of Feynman amplitudes, this is
\begin{equation} \tilde U^0(\Gamma)
= \sum_{k\geq0} \sum_{\{e_1,\ldots,e_k\}\subset\Gamma^{[1]}_{\rm int}}
\Phi (\chi_+^{e_1}\cdots\chi_+^{e_k}\Gamma)
= \Phi(e^{\chi_+}\Gamma).
\label{eq:tildeU-chi}\end{equation}

Instead of applying $\tilde U^0$ to a single graph, we can do this to the combinatorial Green's function $X^{n,l}$. This gives us the Green's function for all graphs in Yang--Mills theory without the ghosts, but including the 4-valent vertices:
\begin{prop}
Collecting residues as above produces the evaluation by the Feynman rules of all 3- and 4-valent graphs in gauge theory without internal ghost or fermion edges:
\begin{equation*} \tilde U^0(X^{n,l}) = \Phi(e^{\chi_+}X^{n,l}) = \Phi(X^{n,l}_{/\fourvertex}). \end{equation*}
\end{prop}
\begin{proof}
The above equation (\ref{eq:tildeU-chi}) is used, together with Lemma \ref{lem:skeletongraphs}.iii.
\end{proof}

\subsection{Exponentiating residues}
Let us discuss the pairing between the integrand with poles along the boundaries of the simplex, with boundaries given by $\sigma_\Gamma: \prod_{i=1}^{|\Gamma_I^{[1]}|}A_i=0$,
and the Feynman integrand $U^0(\Gamma)$ in more detail.

The amplitude $\tilde U^0(\Gamma)$ can be obtained from $U^0(\Gamma)$ by taking residues along hypersurfaces $\prod_{e\in E}A_e=0$ and regular parts and integrating:
\begin{equation*}\begin{split}
\tilde U^0(\Gamma)
&= \sum_{k\geq0} \sum_{\{e_1,\ldots,e_k\}\subset\Gamma^{[1]}_{\rm int}} \int\mathrm d\underline A_{\Gamma^{[1]}\smallsetminus\{e_1,\ldots,e_k\}} \Reg_{A_1,\ldots,\hat{A_{e_1}},\ldots,\hat{A_{e_k}},\ldots=0} 
 \Res_{A_{e_1},\ldots,A_{e_k}=0}U^0(\Gamma) .
\end{split}\end{equation*}

For a function $f=f(\{A_e\})$ of graph polynomial variables $A_e,e\in \gamma^{[1]}_I$ with at most simple poles at the origin localized in disjoint sets of edges $E$, we can write 
\[
f=\sum_E f^E,
\]
where the sum is over all such sets and $f_E$ is the part of $f$ which is regular upon setting variables $A_e, e\in (\Gamma^{[1]}_I-E)$ to zero.

For any set $E$ of mutually disjoint internal edges of $\Gamma$, consider 
$
\prod_{e\in E}\oint_{\gamma_e}f,
$
and let $f^E$ be its regular part.
For any finite graph $\Gamma$, let $\mathcal{E}^\gamma$ be the set of all sets of mutually disjoint edges ($\emptyset$ included).

Consider the differential form 
\[
J_\Gamma^f:=\Big(f^E \bigwedge_{e\in (\Gamma^{[1]}_I-E)}dA_e\Big)_{E\in\mathcal{E}^\Gamma}.
\]
Let $M_\Gamma^E$ be the hypercube
\[
M_\Gamma^E:=\mathbb{R}_+^{|\Gamma^{[1]}_I|-|E|},
\] 
and the corresponding vector
\[
H_\Gamma(M_\Gamma^E)_{E\in\mathcal{E}^\Gamma}.
\]
Then, there is a natural pairing 
\[
\int H_\Gamma\cdot J_\Gamma^f:=\sum_{E\in\mathcal{E}^\Gamma} \int_{M_\Gamma^E} \Big(f^E \bigwedge_{e\in (\Gamma^{[1]}_I-E)}dA_e\Big).
\]

\subsection{Graph homology and the residue map}
Note that in parametric integration we integrate against the simplex $\sigma\equiv \sigma_\Gamma$ with boundary $\prod_{e\in\Gamma^{[1]}} A_e=0$.
We have co-dimension $k$-hypersurfaces given by 
\[
A_{i_1}=\cdots=A_{i_k}=0.
\]
The Feynman integrand we have constructed above comes from regular parts, and residues along these hypersurfaces.
It can be described by the following commutative diagram.

\begin{equation*}\label{graphhomres}\begin{CD}
\mathcal{G}_{j\fourvertex,/\ghostloop}\ni\Gamma@>
                                                                                                                                                                                                                                                            \chi_+>>\chi_+(\Gamma)\in \mathcal{G}_{(j+1)\fourvertex,/\ghostloop}\\
@VV\Phi V@VV\Phi V\\
\Phi(\Gamma)@>\sum_e  \mathrm{Res}_e>>\Phi(\chi_+(\Gamma))
\end{CD}\end{equation*}
The underlying geometry will be interpreted elsewhere.
\subsection{Covariant gauges}
For an edge $e$, let 
\[
G_{\mu\nu}^\rho(e):=\frac{g_{\mu\nu}}{{{\xi_e^\prime}}^2}-2\rho\frac{{\xi_e^\prime}_\mu{\xi_e^\prime}_\nu}{{{\xi_e^\prime}}^4}
\]
the corresponding gluon propagator in a covariant gauge ($\rho=1/2$ being the transversal Landau gauge, $\rho=0$ the Feynman gauge).
One computes
\[
G_{\mu\nu}^\rho(e)=\int_0^\infty \frac{-1}{2A\rho}\frac{\partial}{\partial\xi^\prime_{e\mu}}\frac{\partial}{\partial\xi^\prime_{e\nu}}e^{-\rho A{{\xi_e^\prime}^2}}d\!A=:\int_0^\infty g_{\mu\nu}^\rho(e) d\!A.
\]
We  set 
\[
G^\rho_\Gamma:=\prod_{e\in\Gamma^{[1]}_I} G^\rho_{\mu_{(s(e),e)}\mu_{(t(e),e)}}(e),
\]
for half-edges $(s(e),e)$ and $(t(e),e)$,
and $g^\rho_\Gamma$ accordingly. 

We let $I_\Gamma(\rho)$ be the corresponding scalar integrand
obtained by substituting $A_e\to \rho A_e$ for each internal edge $e$.  

$G^\rho_\Gamma$ acts as a differential operator so that 
\[
G^\rho_\Gamma I_\Gamma(\rho)=F_G(\rho) I_\Gamma(\rho),
\]
with $F_G(\rho)$ a polynomial in $\rho$, edge variables $A_e$ and 4-momenta $\xi_e$.

Similarly, the corolla differential $D_\Gamma$  acts as a differential operator so that 
\[
D_\Gamma I_\Gamma(1)=F_D I_\Gamma(1),
\]
with $F_D$ a polynomial in the 4-momenta $\xi_e$ and a rational function in the edge variables $A_e$.

To compute in an arbitrary covariant gauge, we then work with
\[
F_G(\rho) F_D I_\Gamma(1).
\]

\subsection{Yang--Mills theory}
Consider a cycle $C$ through 3-valent vertices in a graph $\Gamma$, and consider 
\[
\prod_{v\in C}D_g(v_C).\
\]
This is a differential operator with coefficients which are monomials in variables $1/A_e$, where $e\in C$.

Let $D_C$ be the part in this  differential operator
which is linear in all variables $1/A_e$, for $e\in C$.
Let $\phi_C$ be the Feynman rule for a ghost loop on $C$, summed over both orientations.
\begin{lem}\label{ghostly}
\[
\phi_C:=D_C e^{-\sum_{e\in C}A_e{\xi_e^\prime}^2}.
\]
\end{lem}
\begin{proof}
This follows directly from the Feynman rules for ghost propagators and ghost-gluon vertices. Linearization eliminates all poles with residues corresponding to 4-valent 2-ghost-2-gluon vertices.
\end{proof}

Now consider the corolla polynomial $C(\Gamma)$ and replace each half-edge variable $h$ by the differential $D_g(h)$.
This defines a differential operator 
\[
d(\Gamma)^{\mathrm{YM}}:=C(\Gamma)(h\to D_g(h)),
\]
We consider $d(\Gamma)^{\mathrm{YM}} (I_\Gamma)$ where
\[
I_\Gamma:=\frac{e^{-\frac{|N_\Gamma|_{\mathrm{Pf}}}{\psi_\Gamma}}}{\psi_\Gamma^2},
\]
is the scalar integrand for a graph $\Gamma$.
\begin{prop}
All poles in $d(\Gamma)^{\mathrm{YM}} I_\Gamma$ are   located along co-dimension $|E|$ hypersurfaces $A_e=0,\;e\in E $ 
for subsets $E$ of mutually disjoint edges are simple poles.
\end{prop}
\begin{proof}
Corollary \ref{linear} ensures that poles are at most of first order and appear only when two derivatives act on the same edge. By the definition of the corolla
polynomial this can only appear in mutually disjoint ordered pairs of corollas. All poles coming from divergent subgraphs are located along subsets of connected edges, as divergent subgraphs have more than a single edge.
\end{proof}
By our previous results on the Leibniz terms we can summarize now for the parametric integrand:
\begin{cor}
The residues of these poles correspond to graphs where each corresponding pair of corollas $P_e$ is replaced by a 4-valent vertex.
\end{cor}
\begin{proof}
Setting an edge variable to zero shrinks that edge in the two Symanzik polynomials by the standard contraction-deletion identities
\cite{BrI,BrII,AluffMarc}.
\end{proof}

The Leibniz terms serve the useful purpose to shrink an edge between two 3-gluon vertices.
They provide a residue which corresponds to  the integrand where the corresponding edge is a marked edge in our
conventions. it is hence part of the integrand for a graph
with a corresponding 4-valent vertex. As we have checked before, when summing over all connected 3-regular graphs, we correctly reproduce the Feynman integrand for all gluon self-interactions.

We stress that in doing so we want to shrink edges only between pairs of corollas which both are corollas for 3-gluon vertices, and will not mark edges between
other type of vertices.
This leads us to
\begin{defi}
We let $D(\Gamma)^{\mathrm{YM}} I_\Gamma$ be the part $d(\Gamma)^{\mathrm{YM}} I_\Gamma$ which is linear in all variables $1/A_e$.
\end{defi}
This eliminates all  poles in $D(\Gamma)^{\mathrm{YM}} I_\Gamma$ of the form $1/A_e$.  We can regain then the contribution of 4-valent 4-gluon vertices
by using Theorem \ref{bin reformulation} together with  Remark \ref{remthree}:
\begin{lem} Let
\[U_\Gamma  = 
g^\rho_\Gamma C^{\mathrm{fr}}(\Gamma)(a_h\to D_g(h))I_\Gamma.\]
Then $\bar{U}_\Gamma$ (cf. Eq.\eqref{corrfac}) generates the integrand for the complete contribution of $\Gamma$ to the full Yang--Mills  theory amplitude. 
$\bar{U}_\Gamma^R$ generates the corresponding integrand for the renormalized contribution.
\end{lem}
\begin{proof}
Immediate application of Lemma \ref{ghostly} and Theorem \ref{allthm}.
\end{proof}
This also proves Theorem \ref{GTFR} in the context of Yang--Mills theory.
\begin{rem}
If we were to work with non-linear gauges, we could avoid this linearization and use the Leibniz terms for the graphs with 2-gluon 2-ghost and 4-ghost vertices. Also, note that $\bar{U}_\Gamma^R=\bar{U}_\Gamma^R(\rho)$ depends on the gauge parameter.
\end{rem}   
\subsection{Amplitudes with open ghost or fermion lines}
For $k$ open ghost  lines we have a straightforward generalization of these differentials  by using 
$C^{\mathrm{fr}}_{\mathcal{P}}(G)$, see Eq.(\ref{openlines}), where each half edge $h$ is again replaced by $D_g(h)$ and linearization is understood as before. 
For fermion lines, see below.
\subsection{Gauge Theory}
If we include matter fields, we need to add a second differential in particular for fermion fields:
\begin{eqnarray*}
D_f(h) & := & \left(\frac{1}{A(e(h_+))}\frac{\partial}{\partial \xi(e(h_+))_{\mu(h_+)}}\gamma_{\mu(h_+)}\gamma_{\mu(h)}\right.\nonumber\\
 & & -
\left. \frac{1}{A(e(h_-))}\frac{\partial}{\partial \xi(e(h_-))_{\mu(h_-)}}\gamma_{\mu(h)}\gamma_{\mu(h_-)}\right).
\end{eqnarray*}
Now we must carefully distinguish between fermion and ghost cycles.

For a collection of cycles $C_1,
\cdots,C_j$ contributing to $C^j$, consider partitions of this set into two subsets $I_f,I_g$ containing $|I_f|+|I_g|=j$ cycles.
Replace $a_{v,v_C}\to b_{v,v_C}$ for each $C\in I_f$. This defines $C^{I_g,I_f}(\Gamma)(a_h,b_h)$.
Upon summing over all possible partitions $I_g,I_l$ of the cycles for each $j$, this gives a further corolla polynomial for which we write in slight abuse of notation  $C(\Gamma)(a_h,b_h)$.
Assign a differential operator as follows:
\begin{equation*}
U_\Gamma=g^\rho_\Gamma\sum_{j\geq 0} \sum_{|I_g|+|I_f|=j} C^{I_g,I_f}(\Gamma)(D_g(h),D_f(h))\mathrm{colour}^{I_g,I_f}(\Gamma),
\end{equation*}
where in $C^{I_g,I_f}$, for $I_g\cup I_f\not=\emptyset$, we keep only terms which are linear in variables $1/A_e$ for edges $e\in C_1\cup\cdots\cup C_{j}$. 
We can now proceed with $\bar{U}_\gamma$ as before.

Note that the restriction to $I_l=\emptyset$ gives back the corresponding operator for Yang--Mills theory.
From here on, Theorem \ref{GTFR} follows for gauge theory  as before for Yang--Mills theory.
\begin{rem}
Note that all this can be turned into a projective integrand, illuminating the slots in the period matrix which are filled in a gauge theory 
as compared to a scalar field theory. In particular, one hopes that the geometry of Eq.(\ref{graphhomres}) is helpful
to explain appearances and disappearances of periods in gauge theory.
\end{rem}
\begin{rem}
Putting fermions into the same colour rep as gauge bosons allows for immediate cancellations between $D_g$ and $D_f$. This can be illuminating in studying the simplifications for supersymmetric gauge theories.
\end{rem}
\subsection{Examples: QED and Yang--Mills theory} In the following two examples, we compute the one-loop vacuum polarization in quantum electrodynamics, and then the one-loop gluon vacuum polarization in Yang--Mills theory. Both examples can be obtained from corolla differentials acting on the simplest possible 3-regular graph:
\[
\Gamma:=
\parbox{45pt}{\centering\scriptsize\begin{fmfgraph*}(35,17.5)
	\fmfleft 1 \fmfright 2
	\fmf{plain,tension=2}{1,a}
	\fmf{plain,tension=2}{2,b}
	\fmf{plain,left,tension=.5,label=$1$,label.dist=2}{a,b}
	\fmf{plain,left,tension=.5,label=$2$,label.dist=2}{b,a}
	\fmfv{label=$3$,label.dist=2}{1}
	\fmfv{label=$4$,label.dist=2}{2}
	\fmfv{label=a,label.dist=2,label.angle=0}{a}
	\fmfv{label=b,label.dist=2,label.angle=180}{b}
\end{fmfgraph*}}
.
\]
We label its two internal edges $1,2$, and the external edges $3,4$. We also label the two vertices $a,b$.
Edge 3 is oriented from vertex $a$ to vertex $b$, and edge 4 vice versa, say.

We have six half-edges: $h_1:=(a,3),h_2:=(a,2),h_3:=(a,1)$ and $h_4:=(b,1),h_5:=(b,2),h_6:=(b,4)$, with corresponding half-edge variables $a_{a3}, a_{a2}$ etc.

We have four 4-vectors $\xi_1,\xi_2,\xi_3,\xi_4$, with $\xi_e\in \mathbb{M}^4$, Minkowski space, 
with scalar product $\xi_e^2\equiv \xi_e\cdot\xi_e={\xi_e}_0^2-{\xi_e}_1^2-{\xi_e}_2^2-{\xi_e}_3^2$.

\begin{exmp}
In order to compute the one-loop vacuum polarisation in massless QED, 
\[
\Pi_1= \parbox{25pt}{\begin{fmfgraph*}(25,12.5)
	\fmfleft 1 \fmfright 2
	\fmf{boson,tension=2}{1,a}
	\fmf{boson,tension=2}{2,b}
	\fmf{fermion,left,tension=.5}{a,b,a}
\end{fmfgraph*}}
\]
we proceed as follows. We have  for the corolla polynomial
\begin{equation*}
C_1(\Gamma)=a_{a3} a_{b4}.
\end{equation*}
The scalar integrand is
\begin{equation*}
I(\Gamma)=\frac{1}{2}{\xi_3}^2{\xi_4}^2\frac{e^{-\frac{({\xi_1}-{\xi_2})^2A_1A_2+(A_3{\xi_3}^2+A_4{\xi_4}^2)(A_1+A_2)}{A_1+A_2}}}{(A_1+A_2)^2}dA_1 dA_2dA_3 dA_4.
\end{equation*}
We can directly integrate $A_3,A_4$ eliminating any  appearance of ${\xi_3}^2,{\xi_4}^2$ as in this example no derivatives with
respect to external edges appear in the corolla differential.

Indeed, replacing the two half-edge variables in $C_1(\Gamma)$ by the fermion differential and using the linearized corolla differential (we symmetrize below in $\mu(3),\mu(4)$ when allowed)
\[
\frac{1}{4A_1A_2}\left(\frac{\partial}{\partial{\xi_3}_{\mu(3)}}\frac{\partial}{\partial{\xi_4}_{\mu(4)}}
+\frac{\partial}{\partial{\xi_4}_{\mu(4)}}\frac{\partial}{\partial{\xi_3}_{\mu(3)}}\right)
=\frac{1}{2A_1A_2}\frac{\partial}{\partial{\xi_3}_{\mu(3)}}\frac{\partial}{\partial{\xi_4}_{\mu(4)}}\]
 delivers $\pi_1$, the integrand for $\Pi_1$:
\begin{eqnarray*}
\pi_1 & := & - \frac{1}{4}\Tr(\gamma_{\mu(3)}\gamma_{\mu(2)}\gamma_{\mu(4)}\gamma_{\mu(1)})
\frac{\partial}{A_1\partial{\xi_1}_{\mu(1)}}\frac{\partial}{A_2\partial{\xi_2}_{\mu(2)}}I(\Gamma)\\
 & = & 
-\Tr(\gamma_{\mu(3)}\gamma_{\mu(2)}\gamma_{\mu(4)}\gamma_{\mu(1)})({\xi_1}-{\xi_2})_{\mu(1)}({\xi_2}-{\xi_1})_{\mu(2)}A_1A_2\times\nonumber\\
& & 
\times\frac{e^{-\frac{({\xi_1}-{\xi_2})^2A_1A_2}{A_1+A_2}}}{(A_1+A_2)^4}dA_1 dA_2\,
\left(A_\gamma^{F_1}=\frac{A_1A_2}{(A_1+A_2)^4},\,|A_\gamma^{F_1}|_\gamma=0 \right)\nonumber\\
 & & +\Tr(\gamma_{\mu(3)}\gamma_{\mu(2)}\gamma_{\mu(4)}\gamma_{\mu(1)})\frac{1}{2}g_{\mu(1)\mu(2)}\times\nonumber\\
  & & 
\times\frac{e^{-\frac{({\xi_1}-{\xi_2})^2A_1A_2}{A_1+A_2}}}{(A_1+A_2)^3}dA_1 dA_2\,
\left(A_\gamma^{F_2}=\frac{1}{(A_1+A_2)^3},\,|A_\gamma^{F_2}|_\gamma=2 \right).
\end{eqnarray*}
Partially integrating the metric tensor term (equivalently, multiplying $A_\gamma^{F_2}$ by
$\frac{A_1A_2}{(A_1+A_2)A_4}$ before integrating $A_4$, see Eq.(\ref{corrfac})) gives
\begin{eqnarray*}
\pi_1 & = & \Tr(\gamma_{\mu(3)}\gamma_{\mu(2)}\gamma_{\mu(4)}\gamma_{\mu(1)}) ({\xi_1}-{\xi_2})_{\mu(1)}({\xi_2}-{\xi_1})_{\mu(2)}A_1A_2\times\nonumber\\
& & 
\times\frac{e^{-\frac{({\xi_1}-{\xi_2})^2A_1A_2}{A_1+A_2}}}{(A_1+A_2)^4}dA_1 dA_2\nonumber\\
 & & +\Tr(\gamma_{\mu(3)}\gamma_{\mu(2)}\gamma_{\mu(4)}\gamma_{\mu(1)})
 \frac{1}{2}g_{\mu(1)\mu(2)}({\xi_1}-{\xi_2})^2A_1A_2\times\nonumber\\
  & & 
\times\frac{e^{-\frac{({\xi_1}-{\xi_2})^2A_1A_2}{A_1+A_2}}}{(A_1+A_2)^4}dA_1 dA_2.
\end{eqnarray*}
Evaluating the trace, contracting indices and integrating delivers ($Q$ replaces ${\xi_1}-{\xi_2}$ by $q$, subtraction at $q^2=\mu^2$ understood, {\it i.e.} $Q_0$ replaces ${\xi_1}-{\xi_2}$ by $\mu$)
\begin{equation*}
\Pi_1= 8(q^2g_{\mu(3)\mu(4)}-q_{\mu(3)}q_{\mu(4)})\int \frac{A_1A_2e^{-\frac{q^2A_1A_2}{A_1+A_2}}}{(A_1+A_2)^4}-\cdots_{|q^2=\mu^2}dA_1dA_2
\end{equation*}
which can be written projectively
\begin{equation*}
\Pi_1 = 8(q^2g_{\mu(3)\mu(4)}-q_{\mu(3)}q_{\mu(4)})\ln\frac{q^2}{\mu^2} \int_{\mathbb{P}^1(\mathbb{R}_+)} 
\frac{A_1A_2}{(A_1+A_2)^4}(A_1dA_2-A_2dA_1)
\end{equation*}
and which correctly evaluates to the expected transversal result
\begin{equation*}
\Pi_1= 
\frac{4}{3}(q^2g_{\mu(3)\mu(4)}-q_{\mu(3)}q_{\mu(4)})\ln\frac{q^2}{\mu^2}.
\end{equation*}
\end{exmp}

Next, we turn to Yang--Mills theory.
\begin{exmp}
\newcommand\examplegraph[1]{{
\parbox{#1 pt}{\begin{fmfgraph}(#1,#1)
	\fmfleft 1 \fmfright 2
	\fmf{boson,tension=2}{1,a}
	\fmf{boson,tension=2}{2,b}
	\fmf{boson,left,tension=.5}{a,b}
	\fmf{boson,left,tension=.5}{b,a}
\end{fmfgraph}}
}}
\newcommand\examplegraphpl[1]{{
\parbox{#1 pt}{\begin{fmfgraph}(#1,#1/2)
	\fmfleft 1 \fmfright 2
	\fmf{plain,tension=2}{1,a}
	\fmf{plain,tension=2}{2,b}
	\fmf{plain,left,tension=.5}{a,b}
	\fmf{plain,left,tension=.5}{b,a}
\end{fmfgraph}}
}}
We have
\begin{align*} \frac{|N_{\examplegraphpl{10}}|_{\mathrm{Pf}}} {\psi_{\examplegraphpl{10}}} &= -\xi_3^2A_3-\xi_4^2A_4-\frac{(\xi_1-\xi_2)^2A_1A_2}{A_1+A_2} 
\intertext{while the corolla polynomials read}
C_{\examplegraphpl{10}}^0 (\underline a) &= (a_{\rm a3}+a_{\rm a1}+a_{\rm a2})(a_{\rm b4}+a_{\rm b1}+a_{\rm b2})\\
C_{\examplegraphpl{10}}^1 (\underline a) &= a_{\rm a3} a_{\rm b4} \end{align*}

The corresponding differentials then become
\begin{equation*}\begin{split}
C_{\examplegraphpl{10}}^0 (\underline D)
&= (D_{\rm a3}+D_{\rm a1}+D_{\rm a2})(D_{\rm b4}+D_{\rm b2}+D_{\rm b1}) \\
&= \big(-\tfrac12)^2 \bigg(
g^{\mu_1\mu_2} \Big( -\frac1{A_1}\frac\partial{\partial\xi_{1\mu_3}} - \frac1{A_2}\frac\partial{\partial\xi_{2\mu_3}} \Big) \\
&\qquad\quad +g^{\mu_2\mu_3} \Big( \frac1{A_2}\frac\partial{\partial\xi_{2\mu_1}} - \frac1{A_3}\frac\partial{\partial\xi_{3\mu_1}} \Big) \\
&\qquad\quad +g^{\mu_3\mu_1} \Big( \frac1{A_3}\frac\partial{\partial\xi_{3\mu_2}} + \frac1{A_1}\frac\partial{\partial\xi_{1\mu_2}} \Big) \bigg) \\
&\qquad\times\bigg(
g^{\mu_2\mu_1} \Big( -\frac1{A_2}\frac\partial{\partial\xi_{2\mu_4}} - \frac1{A_1}\frac\partial{\partial\xi_{1\mu_4}} \Big) \\
&\qquad\quad +g^{\mu_1\mu_4} \Big( \frac1{A_1}\frac\partial{\partial\xi_{1\mu_2}} - \frac1{A_4}\frac\partial{\partial\xi_{4\mu_2}} \Big) \\
&\qquad\quad +g^{\mu_4\mu_2} \Big( \frac1{A_4}\frac\partial{\partial\xi_{4\mu_1}} + \frac1{A_2}\frac\partial{\partial\xi_{2\mu_1}} \Big) \bigg) \\
C_{\examplegraphpl{10}}^1 (\underline D)
&= D_{\rm a3} D_{\rm b4} \\
&= 4\big(-\tfrac12\big)^2
\Big( -\frac1{A_1}\frac\partial{\partial\xi_{1\mu_3}} - \frac1{A_2}\frac\partial{\partial\xi_{2\mu_3}} \Big) \\
&\qquad\times \Big( -\frac1{A_2}\frac\partial{\partial\xi_{2\mu_4}} - \frac1{A_1}\frac\partial{\partial\xi_{1\mu_4}} \Big)
\end{split}\end{equation*}

for which the linear part, without the factor $4$ (the space-time dimension), is
\begin{equation*}\begin{split}
\tilde C_{\examplegraphpl{10}}^1 (\underline D)
&= \big(-\tfrac12\big)^2 \frac1{A_1A_2}
\Big( \frac{\partial^2}{\partial\xi_{1\mu_3}\partial\xi_{2\mu_4}}
+ \frac{\partial^2}{\partial\xi_{2\mu_4}\partial\xi_{1\mu_3}} \Big)
\end{split}\end{equation*}

We compute 
\begin{align*}
U_{\examplegraphpl{10}}^0
&= C_{\examplegraphpl{10}}^0(\underline D) \frac{e^{\bar\phi_{\examplegraphpl{10}}/\psi_{\examplegraphpl{10}}}}{\psi_{\examplegraphpl{10}}^2} \\
&= \frac1{(A_1+A_2)^4}\Big( (A_1-A_2)g^{\mu_1\mu_2}q^{\mu_3} - (2A_1+A_2)g^{\mu_2\mu_3}q^{\mu_1} \\
&\qquad\qquad + (A_1+2A_2)g^{\mu_3\mu_1}q^{\mu_2} \Big)
\Big( (A_1-A_2)g^{\mu_2\mu_1}q^{\mu_4} \\
&\qquad\qquad+ (A_1+2A_2) g^{\mu_1\mu_4}q^{\mu_2}
- (2A_1+A_2) g^{\mu_4\mu_2}q^{\mu_1} \Big) \\
&\qquad\quad\times e^{-q^2\big(\frac{A_1A_2}{A_1+A_2}+A_3+A_4\big)} \\
&\qquad+ \frac3{(A_1+A_2)^3} \Big(1-\overbrace{\frac{A_1}{A_2}-\frac{A_2}{A_1}}^{\to 0,\,\mathrm{as\,residues\,
are\, scale-independent\, self-loops}}\Big) g^{\mu_1\mu_2}
e^{-q^2\big(\frac{A_1A_2}{A_1+A_2}+A_3+A_4\big)} 
\intertext{and so}
\tilde U_{\examplegraphpl{10}}^0
&= \left( \frac1{(A_1+A_2)^4}\big(
-(2A_1^2+2A_2^2+14A_1A_2)q^{\mu_3}q^{\mu_4}\right. \\
&\qquad\quad \left. +(5A_1^2+5A_2^2+8A_1A_2)q^2g^{\mu_3\mu_4}
\big)\right)_{\mathrm{hence}\,|A_\gamma^{F_1}|_\gamma=0}
e^{-q^2\frac{A_1A_2}{A_1+A_2}} \mathrm dA_1\mathrm dA_2 \\
&\qquad+ \left( \frac3{(A_1+A_2)^3} g^{\mu_3\mu_4}\right)_{\mathrm{hence}\,|A_\gamma^{F_2}|_\gamma=2}
e^{-q^2\frac{A_1A_2}{A_1+A_2}} \mathrm dA_1\mathrm dA_2 
\end{align*}
Similarly, 
\begin{equation*}\begin{split}
\tilde U_{\examplegraphpl{10}}^1
&= \Big(\frac{2A_1A_2}{(A_1+A_2)^4}q^{\mu_3}q^{\mu_4} -\frac1{(A_1+A_2)^3} g^{\mu_3\mu_4} \Big)
e^{-q^2\frac{A_1A_2}{A_1+A_2}} \mathrm dA_1\mathrm dA_2
\end{split}\end{equation*}

We thus obtain for the corresponding integrals:
\begin{equation*}\begin{split}
\int\tilde U_{\examplegraphpl{10}}^{0\rm R}
&= \big(\tfrac{11}3 q^{\mu_3}q^{\mu_4}-\tfrac{25}6 q^2g^{\mu_3\mu_4}) \ln\Big(\frac{q^2}{\mu^2}\Big),
\end{split}\end{equation*}
which corresponds to the gauge boson loop:
\[
\parbox{25pt}{\begin{fmfgraph*}(25,12.5)
	\fmfleft 1 \fmfright 2
	\fmf{boson,tension=2}{1,a}
	\fmf{boson,tension=2}{2,b}
	\fmf{boson,left,tension=.5}{a,b,a}
\end{fmfgraph*}}
,
\]
and
\begin{equation}
\int\tilde U_{\examplegraphpl{10}}^{1\rm R}
= -\big(\tfrac13q^{\mu_3}q^{\mu_4}+\tfrac16 q^2g^{\mu_3\mu_4})
\ln\Big(\frac{q^2}{\mu^2}\Big),\label{ratiotwo}
\end{equation}
which corresponds to the ghost loop:
\[
\parbox{25pt}{\begin{fmfgraph*}(25,12.5)
	\fmfleft 1 \fmfright 2
	\fmf{boson,tension=2}{1,a}
	\fmf{boson,tension=2}{2,b}
	\fmf{ghost,left,tension=.5}{a,b,a}
\end{fmfgraph*}}
.
\]
They combine to a transversal result:
\begin{equation*}\begin{split}
\int\tilde U_{\examplegraphpl{10}}^{\rm R}
&= \int\tilde U_{\examplegraphpl{10}}^{0\rm R}
- \int\tilde U_{\examplegraphpl{10}}^{1\rm R} \\
&= 4(q^{\mu_3}q^{\mu_4}-q^2g^{\mu_3\mu_4})
\ln\Big(\frac{q^2}{\mu^2}\Big).
\end{split}\end{equation*}
Multiplying with $\frac{\mathrm{colour}(\gamma)}{\mathrm{sym}(\gamma)}=\frac{1}{2}f^{h_1h_2h_3}f^{h_2h_3h_6}$, this is the result 
for the 1-loop gluon self-energy in Yang--Mills theory. The gauge theory result is immediate from including the previous example with a suitable colour factor for the fermion loop.
\end{exmp}

\section{Conclusion}
\subsection{Covariant quantization without ghosts}
Consider $\bar{U}_\Gamma^R=:\sum_{i=0}^\infty (-1)^i \bar{U}_\Gamma^{i,R}$ 
in a notation which reflects the alternating structure of the corolla polynomial. Set $\bar{U^{\mathrm{gh}}}_\Gamma^R:=\sum_{i=1}^\infty (-1)^i \bar{U}_\Gamma^{i,R}$.

Covariant quantization delivers naively the integrand $\bar{U}_\Gamma^{0,R}$.
Let $P_L$ be a projector onto longitudinal degrees of freedom so that a physical amplitude is in the kernel of $P_L$,
$P_T$ the corresponding projector such that $P_L+P_T=\mathrm{id}$. 

Summing over connected  graphs contributing to a physical amplitude $X^{r,n}$ at $n$ loops, we know that 
\[
P_L\left( \bar{U}_{X^{r,n}}^{0,R}\right)=-P_L\left( \bar{U^{\mathrm{gh}}}_{X^{r,n}}^{R}\right).
\]
The undesired longitudinal part of the ghost free sector determines the longitudinal part of the ghost contribution by definition.

But also, to compute the ratio
\[
\frac{P_L\left( \bar{U^{\mathrm{gh}}}_{X^{r,n}}^{R}\right)}{P_T\left( \bar{U^{\mathrm{gh}}}_{X^{r,n}}^{R}\right)}
\]
is a combinatorial exercise in determining the interplay of these projectors with the Leibniz terms originating
from the corolla differentials in the various topologies. These longitudinal and transversal differentials 
are determined by the same scalar integrand, and hence are not independent. Eq.(\ref{ratiotwo}) with the ratio two between 
the $qq$ and $g$ form-factor is a typical example. 

So the transversal part of the ghost sector is determined by the combinatorics of scalar graphs and the longitudinal part. It hence is implicitly determined by the ghost free sector. 
\subsection{Slavnov--Taylor Identities}
Slavnov--Taylor identities are treated here as originating from co-ideals in the corresponding Hopf algebras. We reproduce the Feynman rules in four dimensions as renormalized integrands, and can similarly reproduce them
in dimensional regularization, and checked that  our renormalized Feynman integrand vanishes on the corresponding co-ideals, as required.

In future work, we will directly demonstrate the validity of Slavnov--Taylor identities from the structure of the corolla polynomial.
\end{fmffile}

\end{document}